\documentclass[reprint,superscriptaddress,
 amsmath,amssymb,
 aps,
pra,
longbibliography
]{revtex4-2}

\usepackage{graphicx}
\usepackage{dcolumn}
\usepackage{bm}
\usepackage{physics}
\usepackage{amsfonts, amsmath}
\usepackage{dsfont}
\usepackage{amsthm}
\usepackage{bm}
\usepackage{hyperref}
\usepackage{bbold}
\usepackage{color}
\usepackage{xcolor}
\usepackage{lipsum,babel}
\usepackage[normalem]{ulem}
\usepackage{enumerate}
\usepackage{comment}

\newtheorem{theorem}{Theorem}
\newtheorem{corollary}{Corollary}

\newtheorem{lemma}{Lemma}
\newtheorem{property}{Property}

\begin{document}

\definecolor{navy}{RGB}{46,72,102}
\definecolor{pink}{RGB}{219,48,122}
\definecolor{grey}{RGB}{184,184,184}
\definecolor{yellow}{RGB}{255,192,0}
\definecolor{grey1}{RGB}{217,217,217}
\definecolor{grey2}{RGB}{166,166,166}
\definecolor{grey3}{RGB}{89,89,89}
\definecolor{red}{RGB}{255,0,0}

\preprint{APS/123-QED}

\title{Criteria for unbiased estimation: applications to noise-agnostic sensing and {quantum channel estimation}}
\author{Hyukgun Kwon}
\thanks{These authors contributed equally to this work: HK (\href{mailto:kwon37hg@sejong.ac.kr}{kwon37hg@sejong.ac.kr}); KT (\href{mailto:tsubouchi@noneq.t.u-tokyo.ac.jp}{tsubouchi@noneq.t.u-tokyo.ac.jp}).}
\affiliation{Department of Physics and Astronomy, Sejong University, 209 Neungdong-ro Gwangjin-gu, Seoul 05006, Republic of Korea}
\affiliation{Pritzker School of Molecular Engineering, University of Chicago, Chicago, Illinois 60637, USA}
\affiliation{Center for Quantum Technology, Korea Institute of Science and Technology, Seoul 02792, Republic of Korea}

\author{Kento Tsubouchi}
\thanks{These authors contributed equally to this work: HK (\href{mailto:kwon37hg@sejong.ac.kr}{kwon37hg@sejong.ac.kr}); KT (\href{mailto:tsubouchi@noneq.t.u-tokyo.ac.jp}{tsubouchi@noneq.t.u-tokyo.ac.jp}).}
\affiliation{Department of Applied Physics, University of \mbox{Tokyo, 7-3-1} Hongo, Bunkyo-ku, Tokyo 113-8656, Japan}

\author{Chia-Tung Chu}
\affiliation{Pritzker School of Molecular Engineering, University of Chicago, Chicago, Illinois 60637, USA}

\author{Liang Jiang}
\email{liangjiang@uchicago.edu}
\affiliation{Pritzker School of Molecular Engineering, University of Chicago, Chicago, Illinois 60637, USA}

\begin{abstract} 
We establish the necessary and sufficient conditions for {local} unbiased estimation in multi-parameter estimation tasks.
More specifically, we first consider quantum state estimation, where multiple parameters are encoded in a quantum state, and derive simple and intuitive necessary and sufficient conditions for a {local} unbiased estimation based on the derivatives of the encoded state.
To demonstrate the utility of our framework, we consider phase estimation under unknown Pauli noise.
We show that while {local} unbiased phase estimation is infeasible with a naive scheme, employing an entangled probe with a noiseless ancilla enables {local} unbiased estimation.
{We then extend our analysis to quantum channel estimation and derive a necessary and sufficient condition for local unbiased estimability of channel parameters, allowing noiseless ancillae, general CPTP controls, and multiple uses of the channel. As a concrete application, we analyze unbiased estimation of noise parameters affecting non-Clifford gates via cycle benchmarking under SPAM errors.}
\end{abstract}

\maketitle

\section{Introduction}
In many practical applications, quantum systems are influenced by multiple parameters, and accurate estimation of the parameters is essential for the development of quantum science and technology.
For instance, in quantum metrology, phase estimation—relevant to many important physical applications such as gravitational wave detection and magnetic field sensing~\cite{intro-caves1981quantum, intro-abadie2011gravitational,  intro-aasi2013enhanced,  intro-PhysRevLett.71.1355,intro-taylor2008high, intro-baumgratz2016quantum}—is governed not only by the phase itself but also by the noise parameters of the system~\cite{phase-PhysRevA.89.023845, phase-vidrighin2014joint}.
In quantum channel learning, characterizing the Pauli coefficients of the given Pauli channel is essential for advancing quantum hardware development and realizing the quantum computation~\cite{wallman2016noise}.

Multi-parameter estimation, however, presents significant challenges that are not present in single-parameter estimation scenarios.
One of the critical issues arises in quantum state estimation, where multiple parameters are encoded in a quantum state: certain parameters cannot be {locally} unbiasedly estimated when the associated quantum Fisher information matrix (QFIM) is non-invertible~\cite{lemma-PhysRevX.10.031023, albarelli2019upper}.
{Locally} unbiased estimation is essential for meaningful and reliable parameter estimation, as it ensures that the estimated value equals the true parameter on average.
In this context, Refs.~\cite{lemma-PhysRevX.10.031023, albarelli2019upper} establish a necessary and sufficient condition for the existence of {locally} unbiased estimators; however, verifying this condition requires diagonalizing the QFIM—a procedure that is computationally demanding and may pose challenges for practical applications.

{Another critical issue arises in quantum channel estimation, where the purpose is to estimate parameters characterizing an unknown quantum channel: certain parameters of a quantum channel are fundamentally unlearnable. For instance, Ref.~\cite{learnability-chen2023learnability} showed that, in the presence of state-preparation and measurement (SPAM) errors, certain Pauli noise parameters affecting Clifford gates cannot be identified, since different combinations of SPAM errors and Pauli noise can generate the same effective noise process. Such exact non-identifiability immediately precludes unbiased estimation of the corresponding parameters. Motivated by this observation, we address the following general question: \textit{For an arbitrary parameterized quantum channel, which parameters admit locally unbiased estimators with finite variance, and which do not?}}

{In this work, we address these challenges by establishing necessary and sufficient conditions for local unbiased estimation that are formulated independently of the QFIM, relying instead directly on derivatives of the quantum state and channel with respect to the parameters.}
In the first part, we focus on quantum state estimation.
We derive a necessary and sufficient condition for the {local} unbiased estimation of the parameters characterizing a quantum state. In particular, in Theorem \ref{theorem1}, we show that a {locally} unbiased estimator of a parameter exists if and only if the derivative of the encoded state with respect to that parameter cannot be expressed as a linear combination of the derivatives with respect to the other parameters. Conversely, no {locally} unbiased estimator exists for the parameters that violate the condition.
Notably, this condition does not require the explicit calculation of the QFIM, therefore, it is readily verifiable.
As an application of our main results in quantum metrology, we investigate phase estimation in the presence of Pauli noise, focusing on the scenario where all the noise parameters are unknown. We show that {local} unbiased phase estimation is infeasible by naively using a fixed quantum probe without ancilla.
However, employing an entangled quantum probe with a noiseless ancilla enables {local} unbiased phase estimation.

{In the second part, we consider quantum channel estimation. We extend Theorem \ref{theorem1} beyond a quantum state, to the quantum channel.
In Corollary \ref{corollary1}, we establish the necessary and sufficient condition for the local unbiased estimation of the parameters characterizing a quantum channel.
If the condition is violated, local unbiased estimation of the corresponding parameters of the quantum channel becomes fundamentally impossible, even when noiseless ancillae, general CPTP control, and multiple uses of the channel are allowed. Importantly, the applicability of Corollary \ref{corollary1} does not rely on any special structural property of the channel under consideration. Since the criterion is formulated solely in terms of derivatives of the quantum channel with respect to the parameters, it applies broadly to general quantum channel estimation scenarios. As an application, we examine the unbiased estimation of the noise parameters affecting non-Clifford gates via cycle benchmarking.}

\section{Quantum state estimation and noise agnostic sensing}
\subsection{Unbiased quantum state estimation}
{We impose the following standard regularity assumptions. 
Let $\{\hat{\rho}_{\vb*{\theta}}\}_{\vb*{\theta}\in\Theta}$ be a twice continuously differentiable parametric family of quantum states on a finite-dimensional Hilbert space, where $\Theta\subset\mathbb{R}^M$ with $M<\infty$ denotes a finite-dimensional parameter space.
Throughout this work, we assume that the functional form of the map $\vb*{\theta}\mapsto \hat{\rho}_{\vb*{\theta}}$ is known, whereas the true value of the parameter vector $\vb*{\theta}=(\theta_1,\theta_2,\ldots,\theta_M)^T$ is unknown.
Thus, our setting is parameter estimation within a specified quantum statistical model.}
To estimate $\vb*{\theta}$, a POVM measurement denoted as $\{\hat{\Pi}_{\vb*{x}}\}_{\vb*{x}{\in X}}$ is performed, where each element $\hat{\Pi}_{\vb*{x}}$ corresponds to a possible measurement outcome ${\vb*{x}\in X}$ {from the outcome set $X$}.
Based on the measurement outcome $\vb*{x}$, an estimator $\tilde{\theta}_{i}=\tilde{\theta}_{i}(\vb*{x})$ for each parameter $\theta_{i}$ is constructed.
{We define} locally unbiased estimator $\tilde{\theta}_i$ of $\theta_i$ {as an} estimator satisfying 
\begin{align}
    \big<\tilde{\theta_i}\big>=\theta_i,~~\partial_{\theta_j}\big<\tilde{\theta_i}\big>=\delta_{ij}
\end{align}
for some specific value of $\vb*{\theta}$ \cite{helstrom1969quantum, holevo2011probabilistic, est-paris2009quantum}.
{Local unbiasedness is a standard benchmark in local quantum estimation theory and is the condition under which the standard quantum Cramér--Rao bound (QCRB) introduced below is formulated. 
It ensures that, at the parameter point under consideration, the estimator has the correct zeroth- and first-order response to the target parameter while being first-order insensitive to nuisance parameters.}

{For a locally unbiased estimator $\tilde{\theta}_i$,} the estimation error can be lower bounded by the QCRB as
\begin{align}
    \Delta^{2}\theta_i:=\big<\big(\tilde{\theta}_i-\theta_i\big)^{2}\big> \geq \vb*{w}_i^{\mathrm{T}}\mathbf{J}^{-1}\vb*{w}_i, \label{direQCRB}
\end{align}
where $\mathbf{J}$ is the QFIM with elements defined as $\mathbf{J}_{ij} = \mathrm{Tr}[\hat{\rho}_{\vb*{\theta}}\{\hat{L}_{i},\hat{L}_{j}\}{/2}]$, $\vb*{w}_i\in\mathbb{R}^M$ is a column vector with its $j$-th element defined as $(\vb*{w}_i)_j=\delta_{ij}$, and $\big<\cdot \big>$ denotes the average over all possible measurement outcomes~{\cite{helstrom1969quantum, holevo2011probabilistic, est-paris2009quantum}}.
Here, $\hat{L}_{i}$ is the symmetric logarithmic derivative operators satisfying $\pdv{\hat{\rho}_{\vb*{\theta}}}{\theta_{i}}=\frac{1}{2}\{\hat{L}_{i},\hat{\rho}_{\vb*{\theta}}\}$.
The equality in the QCRB can be attained by the POVM measurement that consists of the eigenbasis of $\hat{L}_{i}$~\cite{est-Suzuki_2020}. This guarantees that there exists a locally unbiased estimator of $\theta_i$ with finite estimation error $\vb*{w}_i^{\mathrm{T}}\mathbf{J}^{-1}\vb*{w}_i$, whenever the QFIM $\mathbf{J}$ is invertible. Here, it is worth emphasizing that although such a POVM saturates the equality for a specific $i$, a single POVM that simultaneously saturates the equality for all parameters generally does not exist, since the symmetric logarithmic derivative operators $\hat{L}_{i}$ and $\hat{L}_{j\neq i}$ may not commute. Nevertheless, achievability of Eq. \eqref{direQCRB} for each $\theta_i$ guarantees the {local} unbiased estimation of all parameters when we have access to multiple copies of $\hat{\rho}_{\vb*{\theta}}$, even though it may not yield the optimal estimation error for their simultaneous estimation. A straightforward approach in the multiparameter setting is to prepare $M$ copies of $\hat{\rho}_{\vb*{\theta}}$ and estimate each parameter $\theta_{i}$ individually by performing the POVM corresponding to the eigenbasis of $\hat{L}_{i}$ that saturates Eq. \eqref{direQCRB} for that parameter. The corresponding estimation error is then given by $\vb*{w}_i^{\mathrm{T}}\mathbf{J}^{-1}\vb*{w}_i$, which remains finite whenever the QFIM $\mathbf{J}$ is invertible. Repeating this procedure for all $1\leq i \leq M$ is equivalent to independent estimation of each parameter within the locally unbiased framework.

However, when the QFIM $\mathbf{J}$ is not invertible, the existence of a locally unbiased estimator for $\theta_i$ with finite estimation error is no longer guaranteed.
In particular, such an estimator exists if and only if $\vb*{w}_i$ lies in the support of $\mathbf{J}$~\cite{lemma-PhysRevX.10.031023, albarelli2019upper}; that is described by the following lemma:
\begin{lemma}\label{lemma1}
There exists a {locally} unbiased estimator of $\phi=\vb*{w}^{\mathrm{T}}\vb*{\theta}$ with finite estimation error if and only if $\vb*{w} \in \mathrm{supp}\left(\mathbf{J}\right)$.
In other words
\begin{align}
    \mathbf{J}^{+}\mathbf{J}\vb*{w}  = \vb*{w}. \label{NScondit1}
\end{align}
When Eq.~\eqref{NScondit1} is satisfied, the achievable lower bound of the estimation error of $\phi=\vb*{w}^{\mathrm{T}}\vb*{\theta}$ is given by generalized QCRB:
\begin{align}
    \Delta^{2}\phi=\vb*{w}^{\mathrm{T}}\mathbf{C}\vb*{w} \geq \vb*{w}^{\mathrm{T}}\mathbf{J}^{+}\vb*{w},\label{gqcrb} 
\end{align}
Here, $\mathbf{J}^{+}$ is the Moore-Penrose pseudo inverse of $\mathbf{J}$ and $\mathrm{supp}\left(\mathbf{J}\right)$ is support of $\mathbf{J}$.
\end{lemma}
{We note that Lemma~\ref{lemma1} is a parameter-wise statement. In particular, a singular QFIM does not imply that all parameters fail to admit locally unbiased estimators. Among the parameters \((\theta_1,\ldots,\theta_M)\), those satisfying Eq.~\eqref{NScondit1} admit locally unbiased estimators with finite variance, whereas those violating the condition do not.}
This result was originally proved in Refs. \cite{lemma-PhysRevX.10.031023, albarelli2019upper}. For completeness, we also provide independent proof in the Supplemental Material (SM) Sec. S2 \cite{Note1}.

While Eq.~\eqref{NScondit1} establishes the necessary and sufficient condition for the existence of a locally unbiased estimator, its application involves the diagonalization of the QFIM $\mathbf{J}$, which requires demanding calculation in general. 
{We now present an equivalent tangent-space formulation in Theorem \ref{theorem1} that is independent of QFIM. This formulation translates the condition $\vb*{w} \in \mathrm{supp}\left(\mathbf{J}\right)$ into a linear-independence condition among the state derivatives. It does not require the explicit calculation or diagonalization of the QFIM and will also be useful for formulating the channel-estimation criterion discussed below.} 
Notably, Theorem \ref{theorem1} offers a more accessible formulation compared to Eq.~\eqref{NScondit1} in certain estimation scenarios, particularly in the \emph{noise-agnostic sensing}, which will be discussed in the later part of this section. {Nevertheless, for the clarification, we note that this does not imply that verifying Theorem \ref{theorem1} is universally easier than checking the support condition in Lemma \ref{lemma1}; the relative convenience depends on the specific estimation problem under consideration.}
\begin{theorem}\label{theorem1}
    There exists a locally unbiased estimator of $\theta_{i}$ with finite estimation error if and only if $\pdv{\hat{\rho}_{\vb*{\theta}}}{\theta_{i}}$ cannot be expressed with a linear combinations of $\pdv{\hat{\rho}_{\vb*{\theta}}}{\theta_{j}}\big\vert_{j \neq i}$, i.e., for all $c_j\in\mathbb{C}$,
    \begin{align}
    \pdv{\hat{\rho}_{\vb*{\theta}}}{\theta_{i}} \neq \sum_{j\neq i}^{M} c_{j} \pdv{\hat{\rho}_{\vb*{\theta}}}{\theta_{j}}. \label{NScondit2}
    \end{align}
\end{theorem}
See Methods Sec. \ref{methods:prooftheorem1} for the proof. 
The intuitive explanation of Theorem \ref{theorem1} is as follows: in the context of local parameter estimation, where all $\theta_{i}$ varies near $0$, the encoded state can be approximated as 
\begin{align}
    \hat{\rho}_{\vb*{\theta}} \approx \hat{\rho}_{\vb*{0}}+\sum_{i = 1}^{M} \theta_{i} \pdv{\hat{\rho}_{\vb*{\theta}}}{\theta_{i}},
\end{align}
with $\vb*{0}$ denoting the zero vector. 
{If Eq. \eqref{NScondit2} is violated, the first-order change induced by $\theta_{i}$ can be mimicked by suitable first-order changes of the nuisance parameters. Therefore, $\theta_{i}$ has no distinguishable tangent-space signature in the presence of the other unknown parameters. In this differential local sense, $\theta_{i}$ is not identifiable, and Theorem \ref{theorem1} shows that a locally unbiased estimator with finite variance does not exist.}

Let us provide remarks on Theorem~\ref{theorem1}.
First, although Eq. \eqref{NScondit2} appears to concern a single parameter $\theta_{i}$, however, by applying Theorem~\ref{theorem1} to each parameter individually, we can classify which parameters can be {locally} unbiasedly estimated and which cannot be even though we have access to multiple copies of $\hat{\rho}_{\vb*{\theta}}$. Second, while we focus on the {local} unbiased estimation of $\theta_{i}$ for simplicity, Theorem \ref{theorem1} can be straightforwardly generalized to an arbitrary parameter $\phi_{i}$, which is an $i$th element of a new parameter set $\vb*{\phi}:=\mathbf{O}\vb*{\theta}$, where $\mathbf{O} \in \mathbb{R}^{M \times M}$ is an orthogonal matrix, by invoking the standard change-of-variables rule \cite{est-Suzuki_2020}. {Third, we note that Theorem \ref{theorem1} is equivalent to the support condition in Lemma \ref{lemma1}; its purpose is to express the same estimability criterion directly in the quantum tangent space, rather than introducing a fundamentally different criterion. Finally, Eq.~\eqref{NScondit2} can be viewed as a differential, or tangent-space, local identifiability condition: the infinitesimal effect of $\theta_i$ cannot be reproduced by variations of the other parameters. This first-order notion is distinct from ordinary local or global identifiability \cite{iden-rothenberg1971identification}. For instance, the model
\begin{align}
\hat{\rho}_\theta=\left(\frac12+\theta^3\right)|0\rangle\langle0|
+\left(\frac12-\theta^3\right)|1\rangle\langle1|
\end{align}
is locally one-to-one at $\theta=0$, whereas $\partial_\theta\hat{\rho}_\theta|_{\theta=0}=0$; hence it is not differentially identifiable and admits no locally unbiased estimator with finite variance at $\theta=0$.}

\subsection{Application to noise-agnostic sensing}
As an application of our main results, we consider noise-agnostic sensing. More specifically, we consider $n$-qubit canonical phase estimation scenario, where an unknown signal $\phi$ is encoded by the quantum channel 
\begin{align}
    \mathcal{U}_{\phi}(\cdot):=\hat{U}_{\phi}(\cdot)\hat{U}_{\phi}^{\dagger}
\end{align}
with unitary operator $\hat{U}_\phi := e^{-i\phi/2 \sum_{i=1}^{n} \hat{Z}_i}$. Here, $\hat{Z}_{i}$ denotes the Pauli-$Z$ operator acting on $i$th qubit, and $\phi$ is the parameter that we want to estimate. To estimate $\phi$, a quantum probe $\hat{\rho}_0$ is prepared and evolves under the unitary dynamics, resulting in the ideal signal state $\hat{\rho}_{\mathrm{id}}:=\mathcal{U}_\phi(\hat{\rho}_{0})$. Here, we consider the scenario where Pauli noise $\mathcal{N}$ occurs after the channel $\mathcal{U}_{\phi}$, (resulting in a noisy state given by $\hat{\rho}:=\mathcal{N}\circ \mathcal{U}_{\phi}(\hat{\rho}_{0})$,) and all the relevant parameters of $\mathcal{N}$ are unknown to us. 

Let us clarify the meaning of \emph{unknown parameters} of the Pauli noise, in terms of the Pauli error rates and the Pauli eigenvalues, which are introduced below. The Pauli noise can be described as 
\begin{align}
    \mathcal{N}(\cdot) = (1 - \sum_{\vb*{a}\neq (\vb*{0},\vb*{0})} p_{\vb*{a}})\cdot + \sum_{\vb*{a}\neq (\vb*{0},\vb*{0})} p_{\vb*{a}} \hat{P}_{\vb*{a}} \cdot \hat{P}_{\vb*{a}},
\end{align}
where $\hat{P}_{\vb*{a}}$ is the Pauli operator that is expressed as 
\begin{align}
    \hat{P}_{\vb*{a}} = \prod_{{j}=1}^n i^{x_{{j}}z_{{j}}}\hat{X}_{{j}}^{x_{{j}}}\hat{Z}_{{j}}^{z_{{j}}},    
\end{align}
and the relevant coefficients $\{p_{\vb*{a}}\}_{\vb*{a} \neq (\vb*{0},\vb*{0})}$ are referred to as the \emph{Pauli error rates}. Here, the index $\vb*{a}:=(\vb*{x},\vb*{z})$ consists of the pair of $n$-bit binary vectors $\vb*{x} = (x_1, \ldots, x_n)$ and $\vb*{z} = (z_1, \ldots, z_n)$, with $x_{{j}},z_{{j'}} \in \{0,1\}, ~\forall {j},{j'}$, and $\sum_{\vb*{a}\neq (\vb*{0},\vb*{0})}$ denotes the summation over all possible $\vb*{x}$ and $\vb*{z}$ except $(\vb*{x}, \vb*{z})=(\vb*{0},\vb*{0})$. Alternatively, the Pauli noise $\mathcal{N}$ can also be characterized using {\it Pauli eigenvalues} 
$\{\lambda_{\vb*{a}}\}_{\vb*{a} \neq (\vb*{0},\vb*{0})}$ which satisfy 
\begin{align}
    \mathcal{N}(\hat{P}_{\vb*{a}}) = \lambda_{\vb*{a}} \hat{P}_{\vb*{a}}.
\end{align}
\emph{Unknown parameters} indicates that all the Pauli error rates (or equivalently, eigenvalues) are unknown to us. 

Let us first consider the naive estimation protocol {under an unknown noise channel $\mathcal{N}$} depicted in Fig.~\ref{fig:metrology} (a), where we simply apply the noisy unitary $\mathcal{N}\circ \mathcal{U}_\phi$ to a fixed $n$-qubit probe state $\hat{\rho}_0$.
Without the noise, the ideal state $\hat{\rho}_{\mathrm{id}}=\mathcal{U}_\phi(\hat{\rho}_{0})$ {for an arbitrary input probe $\hat{\rho}_0$} can be represented as 
\begin{align}
    \hat{\rho}_{\mathrm{id}} = \frac{1}{2^n}\qty(\hat{I} + \sum_{\vb*{a}\neq (\vb*{0},\vb*{0})} u_{\vb*{a}}(\phi)\hat{P}_{\vb*{a}}),    
\end{align}
where $u_{\vb*{a}}(\phi)$ is a real function of $\phi$.
{This is simply the Pauli-basis expansion of a general density operator and does not imply that the ideal state is necessarily mixed; for a pure input probe \(\hat{\rho}_{0}\), \(\hat{\rho}_{\mathrm{id}}=\mathcal{U}_\phi(\hat{\rho}_0)\) is also pure.}
When Pauli noise $\mathcal{N}$ occurs on the ideal state, we obtain the noisy state 
\begin{align}
    \label{eq_noisy_probe}
    \hat{\rho} = \frac{1}{2^n}\qty(\hat{I} + \sum_{\vb*{a} \neq (\vb*{0},\vb*{0})} \lambda_{\vb*{a}}u_{\vb*{a}}(\phi)\hat{P}_{\vb*{a}}),
\end{align}
which has $4^{n}$ unknown parameters $\phi$ and $\{\lambda_{\vb*{a}}\}_{\vb*{a}\neq (\vb*{0},\vb*{0})}$. Next, let us apply Theorem \ref{theorem1} to detect whether there exists a {locally} unbiased estimator of $\phi$ with finite estimation error. One can easily find that $\partial_{\phi}\hat{\rho}$ can be expressed as
\begin{align}
    \partial_{\phi}\hat{\rho} =  \sum_{\vb*{a} \neq (\vb*{0},\vb*{0})} \frac{\partial_{\phi}u_{\vb*{a}}(\phi)}{u_{\vb*{a}}(\phi)}\lambda_{\vb*{a}}\partial_{\lambda_{\vb*{a}}}\hat{\rho}. \label{naiveapptheo}
\end{align}
As a consequence, according to Theorem 1, we conclude that the {locally} unbiased estimator of $\phi$ does not exist {for arbitrary probe state $\hat{\rho}_{0}$}. 

In contrast to the naive estimation, leveraging an entangled quantum probe with a noiseless ancilla enables a {locally} unbiased estimation of $\phi$ with finite estimation error, {even under the unknown noise channel $\mathcal{N}$}. As depicted in Fig.~\ref{fig:metrology} (b), let us prepare an $(n+1)$-qubit GHZ state 
\begin{align}
    \frac{1}{\sqrt{2}}(\ket{0^n}\ket{0}_a + \ket{1^n}\ket{1}_a)
\end{align}
as a quantum probe, where the last qubit denotes the noiseless ancilla. In the absence of noise, we obtain the ideal signal state
\begin{align}
    \frac{1}{\sqrt{2}}(\ket{0^n}\ket{0}_a + e^{in\phi}\ket{1^n}\ket{1}_a).
\end{align}
For a simpler description, let us adopt a quantum error correction perspective and interpret the state as 
\begin{align}
    \frac{1}{\sqrt{2}}(\ket{0}_\mathrm{L} + e^{in\phi}\ket{1}_\mathrm{L})
\end{align}
in the $(n+1)$-qubit repetition code, where $\ket{0,1}_{\mathrm{L}}$ denote the logical states of the repetition code.
The repetition code has a set of stabilizer generators $\{\hat{Z}_i\hat{Z}_{i+1}\}_{i=1}^n$ and logical operators $\hat{X}_\mathrm{L} = \hat{X}_1\hat{X}_2\cdots \hat{X}_{n+1}$ and $\hat{Y}_\mathrm{L} = \hat{X}_1\cdots\hat{X}_n\hat{Y}_{n+1}$.
Therefore, {in terms of the density operator}, the ideal state of the entangled case can be represented as
\begin{align}
    \hat{\rho}^{\mathrm{ent}}_{\mathrm{id}} = \hat{\Pi}_{\vb*{0}}\frac{1}{2} (\hat{I} + \cos(n\phi) \hat{X}_\mathrm{L} +  \sin(n\phi) \hat{Y}_\mathrm{L} ).
\end{align}

\begin{figure}[t]
    \centering
    \includegraphics[width=0.9\linewidth]{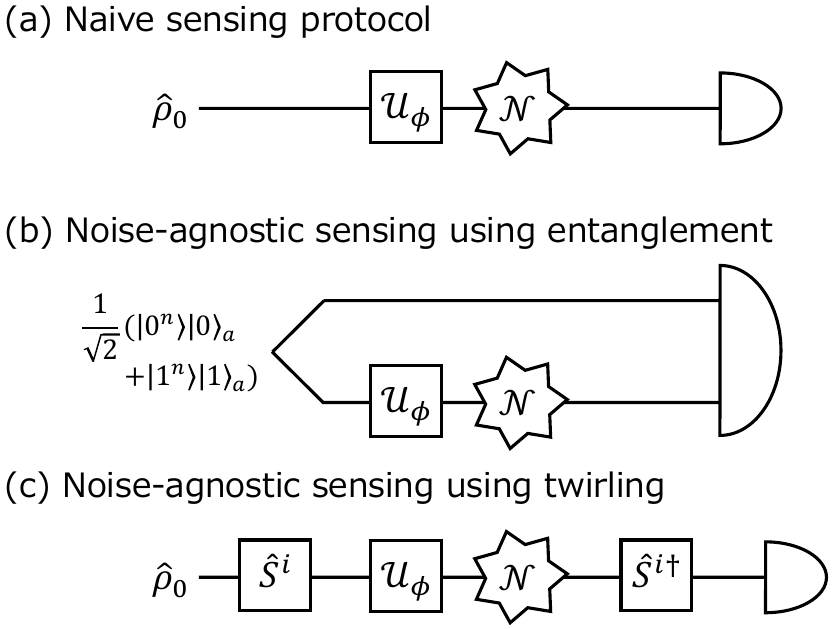}
    \caption{(a) Schematic of the naive estimation. (b) Schematic of the estimation that exploits an entangled state with a noiseless ancilla. {(c) Schematic of the estimation using symmetric Clifford twirling.}
    } 
    \label{fig:metrology}
\end{figure}

A Pauli operator $\hat{P}_{(\vb*{x}, \vb*{z})}$ commutes with the logical operators $\hat{X}_\mathrm{L}$ and $\hat{Y}_\mathrm{L}$ when $\vb*{z}\cdot\vb*{1}$ is even, while it anti-commutes with both of them when $\vb*{z}\cdot\vb*{1}$ is odd.
Here, $\vb*{1} = (1,\ldots,1)$.
Moreover, conjugating $\hat{\Pi}_{\vb*{0}}$ with $\hat{P}_{(\vb*{x}, \vb*{z})}$ results in $\hat{\Pi}_{\vb*{x}}$, where we define the projector to $2^n$ orthogonal code space as $\hat{\Pi}_{\vb*{x}} :=\prod_{i=1}^{n} \frac{1}{2} (\hat{I} +(-1)^{x_i+x_{i+1}} \hat{Z}_i\hat{Z}_{i+1})$ with $x_{n+1} = 0$.
Therefore, when the ideal state $\hat{\rho}^{\mathrm{ent}}_{\mathrm{id}}$ is affected by Pauli error $\hat{P}_{(\vb*{x}, \vb*{z})}$, the resulting state $\hat{P}_{(\vb*{x}, \vb*{z})}\hat{\rho}^{\mathrm{ent}}_{\mathrm{id}}\hat{P}_{(\vb*{x}, \vb*{z})}$ is represented as
\begin{equation}
     \hat{\Pi}_{\vb*{x}}\frac{1}{2} (\hat{I} + (-1)^{\vb*{z}\cdot\vb*{1}}\cos(n\phi) \hat{X}_\mathrm{L} + (-1)^{\vb*{z}\cdot\vb*{1}}\sin(n\phi) \hat{Y}_\mathrm{L}).
\end{equation}
Thus, when unknown Pauli noise occurs on $\hat{\rho}^{\mathrm{ent}}_{\mathrm{id}}$, we obtain a noisy state $\hat{\rho}^{\mathrm{ent}}$ represented as
\begin{align}
    \label{eq_noisy_entangled}
    \hat{\rho}^{\mathrm{ent}}:=\sum_{\vb*{x}\in\{0,1\}^n}p_{\vb*{x}}\hat{\Pi}_{\vb*{x}}\frac{1}{2} \left(\hat{I} + \lambda_{\vb*{x}}\cos(n\phi) \hat{X}_\mathrm{L} +  \lambda_{\vb*{x}}\sin(n\phi) \hat{Y}_\mathrm{L} \right),
\end{align}
where $p_{\vb*{x}} := \sum_{\vb*{z}\in\{0,1\}^n} p_{(\vb*{x},\vb*{z})}$ and $\lambda_{\vb*{x}} = \sum_{\vb*{z}\in\{0,1\}^n}$ $(-1)^{\vb*{z}\cdot\vb*{1}}p_{(\vb*{x},\vb*{z})}/p_{\vb*{x}}$ with $p_{(\vb*{0},\vb*{0})} = 1-\sum_{\vb*{a} \neq (\vb*{0},\vb*{0})}p_{\vb*{a}}$. 

Unlike the naive noisy state Eq.~\eqref{eq_noisy_probe}, two logical Pauli operators $\hat{X}_\mathrm{L}$ and $\hat{Y}_\mathrm{L}$ are affected by the same factor $\lambda_{\vb*{x}}$ for individual code space $\vb*{x}$.
Because of this, in contrast to the naive estimation case,  $\partial_{\phi}\hat{\rho}^{\mathrm{ent}}$ cannot be expressed with a linear combination of the derivatives with respect to other parameters $p_{\vb*{x}}$ and $\lambda_{\vb*{x}}$.
As a consequence, according to Theorem \ref{theorem1}, one can perform a {local} unbiased estimation of $\phi$.

{As another example of a noise-agnostic sensing protocol, we briefly describe a twirling-based protocol using \textit{symmetric Clifford twirling}~\cite{tsubouchi2024symmetric}. 
The key idea of symmetric Clifford twirling is to twirl the noisy sensing channel \(\mathcal{N}\circ\mathcal{U}_{\phi}\) using Clifford operators that commute with the noiseless signal unitary \(\hat{U}_{\phi}\). 
Because these twirling operators commute with \(\hat{U}_{\phi}\), the twirling effectively acts only on the noise channel while leaving the noiseless signal channel unchanged.

For clarity, we focus on the single-qubit case \(n=1\) and on a Pauli noise channel, as in the preceding example. 
The same idea extends to arbitrary \(n\) by using the \(Z\)-symmetric Clifford group. 
Let \(\hat{S}=\mathrm{diag}(1,i)\) and \(\mathcal{S}(\cdot)=\hat{S}(\cdot)\hat{S}^{\dagger}\). 
Since \([\hat{S},\hat{U}_{\phi}]=0\), choosing \(i\in\{0,1\}\) uniformly at random and applying \(\mathcal{S}^{i}\) before and \(\mathcal{S}^{-i}\) after the noisy sensing channel gives
\begin{equation}
    \frac{1}{2}\sum_{i=0}^{1}
    \mathcal{S}^{-i}\circ\mathcal{N}\circ\mathcal{U}_{\phi}\circ\mathcal{S}^{i}
    =
    \mathcal{N}^{\mathrm{twirl}}\circ\mathcal{U}_{\phi},
\end{equation}
where
\begin{equation}
    \mathcal{N}^{\mathrm{twirl}}
    :=
    \frac{1}{2}\sum_{i=0}^{1}
    \mathcal{S}^{-i}\circ\mathcal{N}\circ\mathcal{S}^{i}
\end{equation}
is the twirled noise channel.

Suppose that the original Pauli noise channel \(\mathcal{N}\) has Pauli eigenvalues
\(\{\lambda_x,\lambda_y,\lambda_z\}\), i.e.,
\(\mathcal{N}(\hat{P})=\lambda_P\hat{P}\) for
\(\hat{P}\in\{\hat{X},\hat{Y},\hat{Z}\}\). 
Since \(\{\lambda_x, \lambda_y, \lambda_z\}\) are unknown, locally unbiased estimation of \(\phi\) is impossible, as we have seen in Eqs.~\eqref{eq_noisy_probe} and \eqref{naiveapptheo}.
Meanwhile, the Pauli eigenvalues of the twirled channel \(\mathcal{N}^{\mathrm{twirl}}\) become \(\{\lambda_{xy}, \lambda_{xy}, \lambda_z\}\), where \(\lambda_{xy}=(\lambda_x+\lambda_y)/2\).
Thus, the two independent transverse nuisance parameters \(\lambda_x\) and \(\lambda_y\) are reduced to a single nuisance parameter \(\lambda_{xy}\).
This allows locally unbiased estimation: by taking the input state to be the \(\ket{+}\) state with density matrix \(\hat{\rho}_0 = \frac{1}{2} (\hat{I} + \hat{X})\), the output state \(\hat{\rho}^{\mathrm{twirl}} = (\mathcal{N}^{\mathrm{twirl}} \circ \mathcal{U}_{\phi})(\hat{\rho}_0)\) satisfies
\begin{align}
    \label{eq_noisy_twirled}
    \hat{\rho}^{\mathrm{twirl}} = \frac{1}{2}(\hat{I} + \lambda_{xy}\cos(\phi)\hat{X} + \lambda_{xy}\sin(\phi)\hat{Y}).
\end{align}
Unlike the naive case in Eq.~\eqref{eq_noisy_probe}, the two Pauli operators \(\hat{X}\) and \(\hat{Y}\) are affected by the same factor \(\lambda_{xy}\).
Because of this, \(\partial_{\phi}\hat{\rho}^{\mathrm{twirl}}\) cannot be expressed as a linear combination of the derivatives with respect to the parameter \(\lambda_{xy}\).
As a consequence, according to Theorem~\ref{theorem1}, one can perform a {locally} unbiased estimation of \(\phi\), even though the noise is unknown.}

{Notably, Eqs.~\eqref{naiveapptheo}, \eqref{eq_noisy_entangled}, and \eqref{eq_noisy_twirled} illustrate rather special settings in which Theorem \ref{theorem1} is substantially more convenient to apply than Lemma \ref{lemma1}. In both cases, the locally unbiased estimability of $\phi$ can be identified directly from the Pauli-basis tangent relations. In particular, Eq.~\eqref{naiveapptheo} immediately implies the impossibility of local unbiased estimation of $\phi$, whereas Eqs.~\eqref{eq_noisy_entangled} and \eqref{eq_noisy_twirled} directly establish the existence of a locally unbiased estimator. By contrast, a direct application of Lemma \ref{lemma1} would require constructing the QFIM involving $\phi$ together with all nuisance Pauli-noise parameters and subsequently verifying the corresponding support condition. We emphasize, however, that such simplifications are problem-dependent and do not occur in general.
We leave further discussions of the noise-agnostic sensing protocols in SM Sec. S4, S5, and S6~\cite{Note1}}

\section{Unbiased quantum channel estimation}
Thus far, we have primarily focused on estimating the parameters in quantum {\it states}.
Expanding the scope, in many applications, such as noise characterization, the main purpose is to estimate the parameters of quantum {\it channels}~\cite{learnability-chuang1997prescription, learnability-mohseni2008quantum}.
Therefore, given access to a parameterized quantum channel, it is crucial to determine whether a {local} unbiased estimation of the parameters is feasible.

{As in the state-estimation setting, we impose the corresponding standard regularity assumptions. 
Let $\{\mathcal{E}_{\vb*{\theta}}\}_{\vb*{\theta}\in\Theta}$ be a twice differentiable parametric family of CPTP maps between fixed input and output finite-dimensional Hilbert spaces, where $\Theta\subset\mathbb{R}^M$ with $M<\infty$ denotes the finite-dimensional parameter space. 
We assume that the functional form of the map $\vb*{\theta}\mapsto\mathcal{E}_{\vb*{\theta}}$ is known, whereas the true value of $\vb*{\theta}$ is unknown. 
The CPTP controls used in the sequential protocols below are taken to be known and independent of $\vb*{\theta}$.}
Given access to the channel $\mathcal{E}_{\vb*{\theta}}$, our goal is to determine whether a {locally} unbiased estimation of a certain parameter $\theta_1$ is possible. Notably, Theorem \ref{theorem1} makes no assumptions regarding the choice of a quantum probe or the specific encoding protocols for the parameters. Therefore, it can be directly extended to the channel perspective. This leads to the following corollary, where its proof is provided in SM Sec. S7~\cite{Note1}.

\begin{corollary}\label{corollary1}
    Let $\mathcal{E}_{\vb*{\theta}}$ be an unknown quantum channel parameterized by unknown parameters $\vb*{\theta}$.
    Then, given access to the unknown channel $\mathcal{E}_{\vb*{\theta}}$, $\theta_i$ can be locally unbiasedly estimated with a finite error if and only if for all $c_i\in\mathbb{C}$,
    \begin{align}
    \pdv{\mathcal{E}_{\vb*{\theta}}}{\theta_{i}} \neq \sum_{j \neq i} c_{j} \pdv{\mathcal{E}_{\vb*{\theta}}}{\theta_{j}}.\label{corollary1eq1}
    \end{align}
    If this condition is violated, {local} unbiased estimation of the parameters is fundamentally impossible even when the use of noiseless ancillae, general CPTP control, and multiple uses of the channel are allowed. Mathematically, {local} unbiased estimation remains impossible even when one has access to  the quantum channel $\mathcal{C}_{N+1}\circ(\mathcal{E}_{\vb*{\theta}}\otimes \mathcal{I})\circ \cdots\circ\mathcal{C}_2\circ(\mathcal{E}_{\vb*{\theta}}\otimes \mathcal{I})\circ \mathcal{C}_1$ where $\mathcal{C}_{i}$ are arbitrary CPTP channel and $N$ is arbitrary natural number. 
\end{corollary}
{Corollary~\ref{corollary1} provides a channel-level criterion for local unbiased estimability. Since the condition is formulated directly in terms of derivatives of the channel map, it does not rely on a particular probe state or measurement choice. If the derivative condition is violated, then no allowed protocol constructed from the channel, including protocols with noiseless ancillae, adaptive CPTP controls, and multiple channel uses, can yield a locally unbiased estimator with finite variance for the corresponding parameter.%
As in the state estimation, we note that Eq. \eqref{corollary1eq1} is the differential local identifiability condition for the channel parameter $\theta_{i}$, since it requires that the infinitesimal change of the channel generated by $\theta_{i}$ cannot be reproduced by infinitesimal changes of the nuisance channel parameters.}

\subsection{{Application to cycle benchmarking under SPAM errors}}
{As an example, we consider unbiased estimation of noise-channel parameters for non-Clifford gates via cycle benchmarking under state preparation and measurement (SPAM) errors. Specifically, we consider estimating the noise affecting the Pauli-Z rotation gate \(\hat{U}=\hat{R}_z(\phi)=e^{-i\hat{Z}\phi/2}\). As shown in SM Sec. S8~\cite{Note1}, we assume that the noise channel $\mathcal{N}$ affecting $\hat{U}$ has been twirled, such that its Pauli transfer matrix is expressed as}
\begin{align}
    \begin{pmatrix}
        1 & 0 & 0 & 0 \\
        0 & \lambda_1\cos\theta & \lambda_1\sin\theta & 0 \\
        0 & -\lambda_1\sin\theta & \lambda_1\cos\theta & 0 \\
        \alpha & 0 & 0 & \lambda_2 \\
    \end{pmatrix},
\end{align}
which is parameterized by four independent parameters $\lambda_1, \lambda_2, \alpha$, and $\theta$.
The noise channel $\mathcal{N}$ represents a mixture of depolarizing, dephasing, amplitude damping noise, and coherent noise due to over-rotation. In the absence of SPAM errors, noise channels can be characterized straightforwardly via quantum process tomography~\cite{learnability-chuang1997prescription, learnability-mohseni2008quantum}.
However, in the presence of SPAM errors, it remains unclear whether all parameters of the noise channel can be fully characterized~\cite{learnability-merkel2013self, learnability-blume2013robust, learnability-nielsen2021gate, learnability-chen2023learnability}.
State preparation errors $\mathcal{N}_S$ and measurement errors $\mathcal{N}_M$ are assumed to be Pauli-twirled and are represented by Pauli transfer matrices $\mathrm{diag}(1, \lambda_{1S}, \lambda_{2S}, \lambda_{3S})$ and $\mathrm{diag}(1, \lambda_{1M}, \lambda_{2M}, \lambda_{3M})$, respectively.

\begin{figure}[t]
    \centering
    \includegraphics[width=0.9\linewidth]{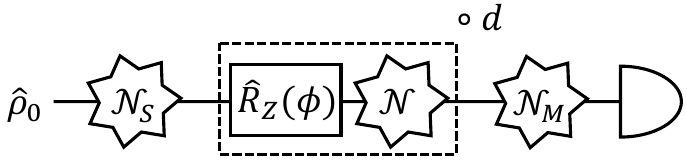}
    \caption{Schematic of cycle benchmarking. We aim to estimate the noise $\mathcal{N}$ affecting the $\hat{R}_z(\phi)$ gate by repeatedly applying the noisy $\hat{R}_z(\phi)$ gate. Here, we consider the case where state preparation of the initial state $\rho_0$ and measurement are affected by unknown Pauli noise $\mathcal{N}_S$ and $\mathcal{N}_M$.}
    \label{fig:cyclecircuit}
\end{figure}

Let us consider the {local unbiased estimability} of the noise parameters $\lambda_1, \lambda_2, \alpha, \theta$ of $\mathcal{N}$ via cycle benchmarking~\cite{learnability-erhard2019characterizing}. Cycle benchmarking identifies noise parameters by repeatedly applying a gate sequence and analyzing the resulting outputs to extrapolate the parameters. The application of the cycle benchmarking in our scenario can be described as follows. To estimate the noise parameters, we repeatedly apply the noisy $\hat{R}_z(\phi)$ gate, with the state-preparation error $\mathcal{N}_S$ occurring before the gate applications and the measurement error $\mathcal{N}_M$ taking place at the end. (See Fig.~\ref{fig:cyclecircuit}.) Mathematically, this sequence of operations for $d$-number of gate applications is described by the channel 
\begin{align}
    \mathcal{N}_{d} = \mathcal{N}_M \circ (\mathcal{N} \circ \mathcal{U})^{\circ d} \circ \mathcal{N}_S,
\end{align}
where 
\begin{align}
    \mathcal{U}(\cdot) := \hat{U}\cdot \hat{U}^\dag
\end{align}
represents the noiseless $\hat{R}_z(\phi)$ gate. To encapsulate $\mathcal{N}_{d}$ for all integer values of $d$ in a single description, we introduce an augmented channel $\mathcal{E}_{\mathrm{cycle}}$, and a classical register $\hat{\rho}_{c}$. This register keeps track of the number of applied gates $d$. Formally, we define
\begin{align}
    \mathcal{E}_{\mathrm{cycle}} : \hat{\rho} \otimes \hat{\rho}_c \mapsto \sum_{d}\mathcal{N}_{d}(\hat{\rho}) \otimes \ketbra{d}\hat{\rho}_c\ketbra{d}.
\end{align}
In this representation, each outcome $\ketbra{d}$ in the register corresponds to the channel $\mathcal{N}_{d}$ acting on the initial state $\hat{\rho}$.
For the quantum channel $\mathcal{E}_{\mathrm{cycle}}$, we find that
\begin{align}
    \alpha \partial_\alpha \mathcal{E}_{\mathrm{cycle}} = \lambda_{3M} \partial_{\lambda_{3M}} \mathcal{E}_{\mathrm{cycle}} - \lambda_{3S} \partial_{\lambda_{3S}} \mathcal{E}_{\mathrm{cycle}}. \label{augucha}
\end{align}
By Corollary~\ref{corollary1}, this implies that a {local} unbiased estimation of $\alpha$ is not feasible using $\mathcal{E}_{\mathrm{cycle}}$. Hence, $\alpha$ {does not admit a locally unbiased estimator} through cycle benchmarking when SPAM errors are unknown. See Methods Sec. \ref{methods:cycle} for the detailed proof. We further show that all the other parameters are {unbiasedly estimated} in the detailed analysis in SM Sec. S8~\cite{Note1}. As another example of {unbiased channel estimation}, we analyze Pauli noise affecting CNOT gate in the presence of SPAM errors in SM Sec. S9~\cite{Note1}. While this case was originally examined in Ref. \cite{learnability-chen2023learnability}, we provide an independent proof based on Collorary \ref{corollary1}.

\section{Conclusion}
We derive necessary and sufficient conditions for the {local} unbiased estimation for both the state and channel estimation. Our analysis implies the following two insights: First, much of the quantum metrology papers focus on analyzing the quantum Fisher information, as it provides a fundamental lower bound on estimation error for {locally} unbiased estimators. However, as illustrated by our noise-agnostic sensing example, when the noise affecting the estimation process is not fully characterized, {local} unbiased estimation may become infeasible. Therefore, before analyzing the quantum Fisher information, it is crucial to first assess the feasibility of {local} unbiased estimation. 
{Second, our analysis provides a local unbiased estimability in quantum channel estimation. The criterion determines whether a given channel parameter admits a locally unbiased estimator with finite variance, even in the presence of nuisance parameters and under general protocols involving noiseless ancillae, adaptive CPTP controls, and multiple channel uses. In the cycle-benchmarking application, this allows us to identify which noise parameters can be unbiasedly estimated under unknown SPAM errors and which parameter directions are obstructed by the effective channel structure. Future work may extend this unbiased-estimation perspective to broader classes of gates, noise models, and benchmarking protocols.}

\section{Acknowledgments}
We acknowledge useful discussions with Yuxiang Yang on state estimation and Senrui Chen on channel estimation.

H.K. was supported by IITP (RS-2025-02263264, RS-2025-25464252, RS-2024-00437191), the Education and Training Program of the Quantum Information Research Support Center (2021M3H3A1036573), and the NRF (RS-2025-25464492, RS-2024-00442710) funded by the Ministry of Science and ICT (MSIT), Korea.
K.T. is supported by the Program for Leading Graduate Schools (MERIT-WINGS), JST BOOST Grant Number JPMJBS2418, JST ASPIRE Grant Number JPMJAP2316, and JST ERATO Grant Number JPMJER2302.
L.J. and C.C. acknowledges support from the ARO(W911NF-23-1-0077), ARO MURI (W911NF-21-1-0325), AFOSR MURI (FA9550-19-1-0399, FA9550-21-1-0209, FA9550-23-1-0338), DARPA (HR0011-24-9-0359, HR0011-24-9-0361), NSF (OMA-1936118, ERC-1941583, OMA-2137642, OSI-2326767, CCF-2312755), NTT Research, and the Packard Foundation (2020-71479). This work was supported in part with funding from Google.org.

\section{Methods}
\subsection{Proof of Theorem \ref{theorem1}}\label{methods:prooftheorem1}
In this section, we provide a proof of Theorem~\ref{theorem1} in the main text, which is fomally stated as follows:
\setcounter{theorem}{0}
\begin{theorem}\label{M:theorem1}
    There exists a locally unbiased estimator of $\theta_{1}$ at $\vb*{\theta} = \vb*{\theta}_0$ with finite estimation error if and only if $\partial_{\theta_{1}}{\hat{\rho}_{\vb*{\theta}}}\vert_{\vb*{\theta} = \vb*{\theta}_0}$ cannot be expressed with a linear combinations of $\partial_{\theta_{i}}{\hat{\rho}_{\vb*{\theta}}}\vert_{\vb*{\theta} = \vb*{\theta}_0}$ for $i \neq 1$, i.e.,
    \begin{align}
    \left.\pdv{\hat{\rho}_{\vb*{\theta}}}{\theta_{1}}\right|_{\vb*{\theta} = \vb*{\theta}_0} \neq \sum_{i = 2}^{M} c_{i} \left.\pdv{\hat{\rho}_{\vb*{\theta}}}{\theta_{i}}\right|_{\vb*{\theta} = \vb*{\theta}_0}
    \end{align}
    for all $c_i\in\mathbb{C}$.
\end{theorem}

\subsubsection{proof of \emph{if} direction}
We aim to show that if $\left.\pdv{\hat{\rho}_{\vb*{\theta}}}{\theta_{1}}\right|_{\vb*{\theta} = \vb*{\theta}_0} \neq \sum_{i = 2}^{M} c_{i} \left.\pdv{\hat{\rho}_{\vb*{\theta}}}{\theta_{i}}\right|_{\vb*{\theta} = \vb*{\theta}_0}$, there exists a locally unbiased estimator of $\theta_{1}$ at $\vb*{\theta} = \vb*{\theta}_0$ with finite estimation error.
\begin{proof}
In this proof, we set $\vb*{\theta} = \vb*{\theta}_0$.
According to Lemma \ref{lemma1}, the existence of a {locally} unbiased estimator of $\theta_{1}$ with finite estimation error is equivalent to the condition $\vb*{w}_{1} \in \mathrm{supp}(\mathbf{J})$. Therefore, the statement is equivalent to proving if $\pdv{\hat{\rho}_{\vb*{\theta}}}{\theta_{1}} \neq \sum_{i \neq 1} c_{i} \pdv{\hat{\rho}_{\vb*{\theta}}}{\theta_{i}}$, then $\vb*{w}_{1} \in \mathrm{supp}(\mathbf{J})$. For this proof, we use the contrapositive method. Specifically, we demonstrate the contrapositive statement: if $\vb*{w}_{1} \notin \mathrm{supp}(\mathbf{J})$, then $\pdv{\hat{\rho}_{\vb*{\theta}}}{\theta_{1}} = \sum_{i \neq 1} c_{i} \pdv{\hat{\rho}_{\vb*{\theta}}}{\theta_{i}}$.

Let us consider $\pdv{\hat{\rho}_{\vb*{\theta}}}{\theta_{1}}$:
\begin{align}
    \begin{split}
    &\pdv{\hat{\rho}_{\vb*{\theta}}}{\theta_{1}} = \sum_{i=1}^{M}\delta_{1i} \pdv{\hat{\rho}_{\vb*{\theta}}}{\theta_{i}} =  \sum_{i=1}^{M}\big(\vb*{w}_{1}^{\mathrm{T}}\vb*{w}_{i}\big) \pdv{\hat{\rho}_{\vb*{\theta}}}{\theta_{i}} \\
    &= \sum_{i,k=1}^{M}  \big(\vb*{w}_{1}^{\mathrm{T}}\vb*{\lambda}_{k}\big) \big(\vb*{\lambda}_{k}^{\mathrm{T}}\vb*{w}_{i}\big)\pdv{\hat{\rho}_{\vb*{\theta}}}{\theta_{i}} = \sum_{k=1}^{M}  \big(\vb*{w}_{1}^{\mathrm{T}}\vb*{\lambda}_{k}\big) \hat{\rho}_{\vb*{\lambda}_{k}},
    \end{split}
\end{align}
where we define 
\begin{align}
    \hat{\rho}_{\vb*{\lambda}_{k}} :=\sum_{i=1}^{M}\big(\vb*{\lambda}_{k}^{\mathrm{T}}\vb*{w}_{i} \big)\partial_{i}\hat{\rho}_{\vb*{\theta}}= \sum_{i=1}^{M}\big(\vb*{w}_{i}^{\mathrm{T}} \vb*{\lambda}_{k}\big)\partial_{i}\hat{\rho}_{\vb*{\theta}}.
\end{align}
Here, $\vb*{w}_{i}$ is $M \times 1$ vector with elements $[\vb*{w}_{i}]_{k}=\delta_{ik}$ for all $1 \leq i \leq M$, and $\{\vb*{\lambda}_{i}\}_{i=1}^{M}$ is the eigenvector set of the QFIM $\mathbf{J}$. Each of the two vector sets forms an orthonormal basis for an $M$-dimensional real vector. 
Next, by leveraging Property \ref{M:prop4}, we complete our proof. The proof of Property \ref{M:prop4} will be provided at the end of this subsection.
\begin{property}\label{M:prop4}
    $\hat{\rho}_{\vb*{\lambda}_{k}}=0$ for $k \geq (s+1)$.
\end{property}
Using \emph{Property} \ref{M:prop4}, we further obtain
\begin{align}
    \begin{split}
    &\pdv{\hat{\rho}_{\vb*{\theta}}}{\theta_{1}}=\sum_{k=1}^{s}  \big(\vb*{w}_{1}^{\mathrm{T}}\vb*{\lambda}_{k}\big) \hat{\rho}_{\vb*{\lambda}_{k}} = \sum_{i=1}^{M}\sum_{k=1}^{s} \big(\vb*{w}_{1}^{\mathrm{T}} \vb*{\lambda}_{k})\big(\vb*{\lambda}_{k}^{\mathrm{T}} \vb*{w}_{i})\pdv{\hat{\rho}_{\vb*{\theta}}}{\theta_{i}}\\
    &= \sum_{k=1}^{s} \big(\vb*{w}_{1}^{\mathrm{T}} \vb*{\lambda}_{k})\big(\vb*{\lambda}_{k}^{\mathrm{T}} \vb*{w}_{1})\pdv{\hat{\rho}_{\vb*{\theta}}}{\theta_{1}} + \sum_{i=2}^{M}\sum_{k=1}^{s} \big(\vb*{w}_{1}^{\mathrm{T}} \vb*{\lambda}_{k})\big(\vb*{\lambda}_{k}^{\mathrm{T}} \vb*{w}_{i})\pdv{\hat{\rho}_{\vb*{\theta}}}{\theta_{i}}.
    \end{split}
\end{align}
Since we are considering the case where $\vb*{w}_{1} \notin \mathrm{supp}(\mathbf{J})$, it follows that $\sum_{k=1}^{s} \big(\vb*{w}_{1}^{\mathrm{T}} \vb*{\lambda}_{k})\big(\vb*{\lambda}_{k}^{\mathrm{T}} \vb*{w}_{1}) <1$. As a consequence, we can express $\pdv{\hat{\rho}_{\vb*{\theta}}}{\theta_{1}}$ as
\begin{align}
    \pdv{\hat{\rho}_{\vb*{\theta}}}{\theta_{1}} = \sum_{i=2}^{M}\left( \frac{\sum_{k=1}^{s} \big(\vb*{w}_{1}^{\mathrm{T}} \vb*{\lambda}_{k})\big(\vb*{\lambda}_{k}^{\mathrm{T}} \vb*{w}_{i})}{1-\sum_{k=1}^{s} \big(\vb*{w}_{1}^{\mathrm{T}} \vb*{\lambda}_{k})\big(\vb*{\lambda}_{k}^{\mathrm{T}} \vb*{w}_{1})}\right)\pdv{\hat{\rho}_{\vb*{\theta}}}{\theta_{i}}.
\end{align}
Therefore, if $\vb*{w}_{1} \notin \mathrm{supp}(\mathbf{J})$, there exists a set of coefficient $\{c_{i}\}_{i=1}^{M}$ that satisfies $\pdv{\hat{\rho}_{\vb*{\theta}}}{\theta_{1}} = \sum_{i = 2}^{M} c_{i} \pdv{\hat{\rho}_{\vb*{\theta}}}{\theta_{i}}$.
\end{proof}

\emph{proof of Property \ref{M:prop4}}.---Let us express the components of the QFIM using the eigenvectors of $n$-qubit encoded state $\hat{\rho}_{\vb*{\theta}}$. The diagonal form of the $d$-dimensional encoded state $\hat{\rho}_{\vb*{\theta}}$ is expressed as
\begin{align}
    \hat{\rho}_{\vb*{\theta}}=\sum_{\alpha=1}^{d}\rho_{\alpha}\dyad{\psi_{\alpha}}=\sum_{\alpha=1}^{S}\rho_{\alpha}\dyad{\psi_{\alpha}}+\sum_{\beta=S+1}^{d}0\dyad{\psi_{\beta}},
\end{align}
where $\rho_{\alpha}>0$ for all $ 1\leq \alpha \leq S$. If $\hat{\rho}_{\vb*{\theta}}$ is not full-rank in the Hilbert space, then $S < d$. In addition, we emphasize that $\{\ket{\psi_{\alpha}}\}_{\alpha=1}^{d}$ forms an orthonormal basis for the $d$-dimensional Hilbert space. By using $\{\ket{\psi_{\alpha}}\}_{\alpha=1}^{d}$ , the elements of the QFIM can be expressed as
\begin{align}
    \vb*{w}_{i}^{\mathrm{T}}\mathbf{J}\vb*{w}_{j} = [\mathbf{J}]_{ij}=\sum_{\rho_{\alpha}+\rho_{\beta} \neq 0}\frac{2\Re(\langle \psi_{\alpha} \vert \partial_{i}\hat{\rho}_{\vb*{\theta}} \vert \psi_{\beta} \rangle \langle \psi_{\beta} \vert \partial_{j}\hat{\rho}_{\vb*{\theta}} \vert \psi_{\alpha} \rangle)}{\rho_{\alpha}+\rho_{\beta}},
\end{align}
where the summation is over all $1 \leq \alpha,\beta \leq d$ such that $\rho_{\alpha}+\rho_{\beta} \neq 0$.

Next, let us consider $\vb*{\lambda}_{k}^{\mathrm{T}}\mathbf{J} \vb*{\lambda}_{k}$. For $k\geq (s+1)$, $\vb*{\lambda}_{k}$ is the eigenvector of $\mathbf{J}$ with eigenvalue $0$, therefore, $\vb*{\lambda}_{k}^{\mathrm{T}}\mathbf{J} \vb*{\lambda}_{k}=0$. This can be expressed in terms of the basis set $\{\vb*{w}_{i}\}_{i=1}^{M}$ as follows:
\begin{align}
    \begin{split}
    &0=\vb*{\lambda}_{k}^{\mathrm{T}}\mathbf{J} \vb*{\lambda}_{k} = \vb*{\lambda}_{k}^{\mathrm{T}}\left(\sum_{i=1}^{M}\vb*{w}_{i}\vb*{w}_{i}^{\mathrm{T}}\right) \mathbf{J} \left(\sum_{j=1}^{M}\vb*{w}_{j}\vb*{w}_{j}^{\mathrm{T}}\right)\vb*{\lambda}_{k}\\
    &=\sum_{i,j=1}^{M} (\vb*{\lambda}_{k}^{\mathrm{T}} \vb*{w}_{i})[\mathbf{J}]_{ij} (\vb*{w}_{j}^{\mathrm{T}}\vb*{\lambda}_{k} ) = 2\sum_{\rho_{\alpha}+\rho_{\beta} \neq 0}\frac{\abs{\langle \psi_{\alpha} \vert \hat{\rho}_{\vb*{\lambda}_{k}} \vert \psi_{\beta} \rangle}^{2}}{\rho_{\alpha}+\rho_{\beta}}, \label{QFIlambda}
    \end{split}
\end{align}
Here, the numerator is equal or greater than $0$ and the denominator is positive. Therefore, for $\vb*{\lambda}_{k}^{\mathrm{T}}\mathbf{J} \vb*{\lambda}_{k}=0$ to hold for $(s+1) \leq  k \leq M$, it must be true that $\langle \psi_{\alpha} \vert \hat{\rho}_{\vb*{\lambda}_{k}} \vert \psi_{\beta} \rangle=0$ for all $1 \leq \alpha,\beta \leq d$ such that $\rho_{\alpha}+\rho_{\beta} \neq 0$.

{Next, let us show 
\(\bra{\psi_{\alpha}}\hat{\rho}_{\vb*{\lambda}_{k}}\ket{\psi_{\beta}}=0\)
for \(\alpha\) and \(\beta\) satisfying
\(\rho_{\alpha}+\rho_{\beta}=0\). We note that
\(\rho_{\alpha}+\rho_{\beta}=0\) if and only if
\(\rho_{\alpha}=0\) and \(\rho_{\beta}=0\).
From the definition of \(\hat{\rho}_{\vb*{\lambda}_{k}}\), we have
\begin{align}
    \hat{\rho}_{\vb*{\theta}+\varepsilon\vb*{\lambda}_k}
    =
    \hat{\rho}_{\vb*{\theta}}
    +
    \varepsilon \hat{\rho}_{\vb*{\lambda}_k}
    +
    O(\varepsilon^2)
\end{align}
for sufficiently small \(\varepsilon\).
When we sandwich this state with an arbitrary normalized vector
\(\ket{\psi^0}\in\ker(\hat{\rho}_{\vb*{\theta}})\), we obtain
\begin{align}
    0
    \leq
    \ev*{\hat{\rho}_{\vb*{\theta}+\varepsilon\vb*{\lambda}_k}}{\psi^0}
    =
    \varepsilon
    \ev*{\hat{\rho}_{\vb*{\lambda}_k}}{\psi^0}
    +
    O(\varepsilon^2).
\end{align}
Dividing by \(\varepsilon>0\) and taking
\(\varepsilon\to0^+\), we obtain
\begin{align}
\ev*{\hat{\rho}_{\vb*{\lambda}_k}}{\psi^0}\ge 0 .
\end{align}
Let \(\hat{\Pi}_0\) denote the projector onto
\(\ker(\hat{\rho}_{\vb*{\theta}})\), and define
\begin{align}
\hat{A}_0
:=
\hat{\Pi}_0
\hat{\rho}_{\vb*{\lambda}_k}
\hat{\Pi}_0 .
\end{align}
Since the above inequality holds for every normalized
\(\ket{\psi^0}\in\ker(\hat{\rho}_{\vb*{\theta}})\),
the null-space block \(\hat{A}_0\) is positive semidefinite.

Moreover, since \(\hat{\rho}_{\vb*{\lambda}_k}\) is a linear
combination of derivatives of a normalized density operator,
we have
\begin{align}
\mathrm{Tr}\!\left[\hat{\rho}_{\vb*{\lambda}_k}\right]=0 .
\end{align}
From the previous argument,
\begin{align}
\bra{\psi_\alpha}
\hat{\rho}_{\vb*{\lambda}_k}
\ket{\psi_\beta}
=0
\end{align}
for all \(\alpha,\beta\) satisfying
\(\rho_\alpha+\rho_\beta\neq0\). In particular, all diagonal
elements on \(\operatorname{supp}(\hat{\rho}_{\vb*{\theta}})\)
vanish. Hence,
\begin{align}
0
=
\mathrm{Tr}\!\left[\hat{\rho}_{\vb*{\lambda}_k}\right]
=
\mathrm{Tr}\!\left[
\hat{\Pi}_0
\hat{\rho}_{\vb*{\lambda}_k}
\hat{\Pi}_0
\right]
=
\mathrm{Tr}\!\left[\hat{A}_0\right].
\end{align}
Since \(\hat{A}_0\) is positive semidefinite and has zero trace,
it must be the zero matrix. Therefore,
\begin{align}
\bra{\psi_\alpha}
\hat{\rho}_{\vb*{\lambda}_k}
\ket{\psi_\beta}
=0
\end{align}
also for \(\rho_\alpha=0\) and \(\rho_\beta=0\).
Combining this with the case \(\rho_\alpha+\rho_\beta\neq0\),
we conclude that
\begin{align}
\hat{\rho}_{\vb*{\lambda}_k}=0 .
\end{align}
}
\qed

\subsubsection{proof of \emph{only if} direction}
We aim to show that when there exists a locally unbiased estimator of $\theta_{1}$ at $\vb*{\theta} = \vb*{\theta}_0$ with finite estimation error, then $\left.\pdv{\hat{\rho}_{\vb*{\theta}}}{\theta_{1}}\right|_{\vb*{\theta} = \vb*{\theta}_0} \neq \sum_{i = 2}^{M} c_{i} \left.\pdv{\hat{\rho}_{\vb*{\theta}}}{\theta_{i}}\right|_{\vb*{\theta} = \vb*{\theta}_0}$.
\begin{proof}
In this proof, we set $\vb*{\theta} = \vb*{\theta}_0$ and show the proof by contradiction.
Let us assume $\mathbf{J}^{+}\mathbf{J}\vb*{w}_{1} = \vb*{w}_{1}$ and there exists $c_i\in\mathbb{R}$ such that $\pdv{\hat{\rho}_{\vb*{\theta}}}{\theta_{1}} = \sum_{i \neq 1} c_{i} \pdv{\hat{\rho}_{\vb*{\theta}}}{\theta_{i}}$.
From the assumption of $\pdv{\hat{\rho}_{\vb*{\theta}}}{\theta_{1}} = \sum_{i \neq 1} c_{i} \pdv{\hat{\rho}_{\vb*{\theta}}}{\theta_{i}}$ the elements of the QFIM satisfy:
\begin{align}
    [\mathbf{J}]_{k1}=\sum_{i \neq 1} c_{i} [\mathbf{J}]_{ki}, \forall k. \label{M:detectioncondition2}
\end{align}
This is because $\mathbf{J}_{ki} \equiv \mathrm{Tr} \left[\hat{L}_{k} \pdv{\hat{\rho}_{\vb*{\theta}}}{\theta_{i}} \right]$. Next, let us inspect $\mathbf{J}^{+}\mathbf{J}\vb*{w}_{1}$. Using Eq. \eqref{M:detectioncondition2}, we can express it as
\begin{align}
    \begin{split}
    & \left[\mathbf{J}^{+}\mathbf{J}\vb*{w}_{1}\right]_{a}  =\sum_{b,c=1}^{M} [\mathbf{J}^{+}]_{ab}[\mathbf{J}]_{bc}[\vb*{w}_{1}]_{c}\\
    &=\sum_{b=1}^{M}[\mathbf{J}^{+}]_{ab}[\mathbf{J}]_{b1}=\sum_{b=1}^{M}[\mathbf{J}^{+}]_{ab} \sum_{j \neq 1} c_{j} [\mathbf{J}]_{bj} =\left[\mathbf{J}^{+}\mathbf{J}\vb*{c}\right]_{a},
    \end{split}
\end{align}
where $\vb*{c}=(0,c_{2},\cdots,c_{M})^{\mathrm{T}}$. Notably, $\vb*{c}$ is orthogonal to $\vb*{w}_{1}$, as $\vb*{w}_{1}$ has non-zero support only in the first component. Now, consider the following inner product:
\begin{align}
    \vb*{w}_{1}^{\mathrm{T}}\left(\mathbf{J}^{+}\mathbf{J}\vb*{w}_{1}\right)=\vb*{w}_{1}^{\mathrm{T}}\left(\mathbf{J}^{+}\mathbf{J}\vb*{c}\right)=\left(\vb*{w}_{1}^{\mathrm{T}}\mathbf{J}^{+}\mathbf{J}\right)\vb*{c}=\vb*{w}_{1}^{\mathrm{T}} \vb*{c}=0. \label{pcontra2}
\end{align}
However, this contradicts with $\vb*{w}_{1}^{\mathrm{T}}\left(\mathbf{J}^{+}\mathbf{J}\vb*{w}_{1}\right) = \vb*{w}_1^T\vb*{w}_1 = 1$.
Therefore, when $\mathbf{J}^{+}\mathbf{J}\vb*{w}_{1} = \vb*{w}_{1}$, we have $\pdv{\hat{\rho}_{\vb*{\theta}}}{\theta_{1}} \neq \sum_{i \neq 1} c_{i} \pdv{\hat{\rho}_{\vb*{\theta}}}{\theta_{i}}$ for all $c_i\in\mathbb{R}$.
\end{proof}

\subsection{Proof of Corollary \ref{corollary1}}\label{methods:proofcorollary1}
We provide the proof of Corollary \ref{corollary1} stated in the main text.
For clarity,  we restate the corollary below.
We note that the Corollary stated in the main text is the simplified version, and the following Corollary is more general.

\setcounter{corollary}{0}
\begin{corollary}\label{M:corollary1}
    Let $\mathcal{E}_{\vb*{\theta}}$ be an unknown quantum channel parameterized by unknown parameters $\vb*{\theta} = (\theta_1, \ldots, \theta_M)^T$.
    Then, given multiple accesses to the unknown channel $\mathcal{E}_{\vb*{\theta}}$, a locally unbiased estimator of parameter $\theta_1$ at $\vb*{\theta} = \vb*{\theta}_0$ with finite estimation error exists if and only if
    \begin{align}
    \label{M:eq_cors1_1}
    \left.\pdv{\mathcal{E}_{\vb*{\theta}}}{\theta_{1}}\right|_{\vb*{\theta} = \vb*{\theta}_0} \neq \sum_{i \neq 1} c_{i} \left.\pdv{\mathcal{E}_{\vb*{\theta}}}{\theta_{i}}\right|_{\vb*{\theta} = \vb*{\theta}_0}
    \end{align}
    for all $c_i\in\mathbb{C}$. Here, $\vb*{\theta}_0=(\theta_{1,0},\theta_{2,0},\cdots,\theta_{M,0})^{\mathrm{T}}$ is a fixed point.
    
    If this condition is violated, {local} unbiased estimation of the parameters is fundamentally impossible even when the use of noiseless ancillae, general CPTP control, and multiple uses of the channel are allowed. Mathematically, {local} unbiased estimation remains impossible even when one has access to  the quantum channel $\mathcal{C}_{N+1}\circ(\mathcal{E}_{\vb*{\theta}}\otimes \mathcal{I})\circ \cdots\circ\mathcal{C}_2\circ(\mathcal{E}_{\vb*{\theta}}\otimes \mathcal{I})\circ \mathcal{C}_1$ where $\mathcal{C}_{i}$ are arbitrary CPTP channel and $N$ is arbitrary natural number.
\end{corollary}

\begin{proof}
    In this proof, we set $\vb*{\theta} = \vb*{\theta}_0$.
    We employ the fact that the most general estimation protocol for unknown quantum channels is represented by the \emph{sequential strategy}~\cite{giovannetti2006quantum, demkowicz2014using}.
    In this strategy, we access $\mathcal{E}_{\vb*{\theta}}$ for $N$ times and prepare the quantum state
    \begin{align}
        \label{M:eq_cors1_pr_1}
        \hat{\rho}_{\mathrm{seq}} = \mathcal{C}_{N+1}\circ(\mathcal{E}_{\vb*{\theta}}\otimes \mathcal{I})\circ \cdots\circ\mathcal{C}_2\circ(\mathcal{E}_{\vb*{\theta}}\otimes \mathcal{I})\circ \mathcal{C}_1(\hat{\rho}_0).
    \end{align}
    Here, $\hat{\rho}_0$ is the initial state of the system and ancilla, $\mathcal{C}_i$ is an arbitrary CPTP channel, and $\mathcal{I}$ denotes the identity channel on the ancilla.
    Since the sequential strategy is the most general estimation protocol, a {locally} unbiased estimator of $\theta_1$ with finite estimation error exists if and only if there exists a sequential strategy that allows {locally} unbiased estimation of $\theta_1$ from the output state $\rho_{\mathrm{seq}}$.
    This condition is equivalent to the existence of $N$, $\hat{\rho}_0$, and $\mathcal{C}_i$ such that
    \begin{align}
        \label{M:eq_cors1_pr_2}
        \begin{split}
        &\pdv{\theta_{1}}\mathcal{C}_{N+1}\circ(\mathcal{E}_{\vb*{\theta}}\otimes \mathcal{I})\circ \cdots\circ\mathcal{C}_2\circ(\mathcal{E}_{\vb*{\theta}}\otimes \mathcal{I})\circ \mathcal{C}_1(\hat{\rho}_0)\\
        &\neq \sum_{i \neq 1} d_{i} \pdv{\theta_{i}}\mathcal{C}_{N+1}\circ(\mathcal{E}_{\vb*{\theta}}\otimes \mathcal{I})\circ \cdots\circ\mathcal{C}_2\circ(\mathcal{E}_{\vb*{\theta}}\otimes \mathcal{I})\circ \mathcal{C}_1(\hat{\rho}_0)
        \end{split}
    \end{align}
    for all $d_i\in\mathbb{C}$.
    Therefore, it suffices to show the equivalence between Eq.~\eqref{M:eq_cors1_1} and Eq.~\eqref{M:eq_cors1_pr_2}.

    To establish this equivalence, we demonstrate that parameters $c_i\in\mathbb{C}$ satisfying
    \begin{align}
    \label{M:eq_cors1_pr_3}
    \pdv{\mathcal{E}_{\vb*{\theta}}}{\theta_{1}} = \sum_{i \neq 1} c_{i} \pdv{\mathcal{E}_{\vb*{\theta}}}{\theta_{i}}
    \end{align}
    exist if and only if, for all $N$, $\hat{\rho}_0$, and $\mathcal{C}_i$, parameters $d_i\in\mathbb{C}$ exist such that
    \begin{align}
        \label{M:eq_cors1_pr_4}
        \begin{split}
        &\pdv{\theta_{1}}\mathcal{C}_{N+1}\circ(\mathcal{E}_{\vb*{\theta}}\otimes \mathcal{I})\circ \cdots\circ\mathcal{C}_2\circ(\mathcal{E}_{\vb*{\theta}}\otimes \mathcal{I})\circ \mathcal{C}_1(\hat{\rho}_0)\\
        &= \sum_{i \neq 1} d_{i} \pdv{\theta_{i}}\mathcal{C}_{N+1}\circ(\mathcal{E}_{\vb*{\theta}}\otimes \mathcal{I})\circ \cdots\circ\mathcal{C}_2\circ(\mathcal{E}_{\vb*{\theta}}\otimes \mathcal{I})\circ \mathcal{C}_1(\hat{\rho}_0).
        \end{split}
    \end{align}

    We first prove the \emph{if} direction.
    Assume that for all $N$, $\hat{\rho}_0$, and $\mathcal{C}_i$, there exist parameters $d_i\in\mathbb{C}$ satisfying Eq.~\eqref{M:eq_cors1_pr_4}.
    In particular, consider the case where $N = 1$, $\hat{\rho}_0 = \ketbra{\Psi}$ with $\ket{\Psi} = \sum_i \ket{i}\ket{i}$ being the maximally entangled state, and where $\mathcal{C}_1$ and $\mathcal{C}_2$ are identity operations on the system and the ancilla.
    In this case, we obtain
    \begin{align}
        \label{eq_cors1_pr_5}
        \left(\pdv{\theta_{1}}\mathcal{E}_{\vb*{\theta}}\otimes \mathcal{I}\right)(\ketbra{\Psi})
        = \left(\sum_{i \neq 1} d_{i} \pdv{\theta_{i}}\mathcal{E}_{\vb*{\theta}}\otimes \mathcal{I}\right)(\ketbra{\Psi}),
    \end{align}
    for some $d_i\in\mathbb{C}$.
    This equality implies that the Choi state of the map $\partial_{\theta_{1}}\mathcal{E}_{\vb*{\theta}}$ coincides with the Choi state of the map $\sum_{i \neq 1} d_{i} \partial_{\theta_{i}}\mathcal{E}_{\vb*{\theta}}$.
    Therefore, we obtain Eq.~\eqref{M:eq_cors1_pr_3} with $c_i = d_i$.
    
    Next, we prove the \emph{only if} direction.
    Assume that there exist parameters $c_i\in\mathbb{C}$ satisfying Eq.~\eqref{M:eq_cors1_pr_3}.
    Then, since we have
    \begin{align}
        \label{eq_cors1_pr_6}
        \begin{split}
        &\;\;\;\pdv{\theta_{i}}\mathcal{C}_{N+1}\circ(\mathcal{E}_{\vb*{\theta}}\otimes \mathcal{I})\circ \cdots\circ\mathcal{C}_2\circ(\mathcal{E}_{\vb*{\theta}}\otimes \mathcal{I})\circ \mathcal{C}_1(\hat{\rho}_0) \\
        &= \mathcal{C}_{N+1}\circ(\partial_{\theta_i}\mathcal{E}_{\vb*{\theta}}\otimes \mathcal{I})\circ \cdots\circ\mathcal{C}_2\circ(\mathcal{E}_{\vb*{\theta}}\otimes \mathcal{I})\circ \mathcal{C}_1(\hat{\rho}_0) \\
        &+ \cdots + \mathcal{C}_{N+1}\circ(\mathcal{E}_{\vb*{\theta}}\otimes \mathcal{I})\circ \cdots\circ\mathcal{C}_2\circ(\partial_{\theta_i}\mathcal{E}_{\vb*{\theta}}\otimes \mathcal{I})\circ \mathcal{C}_1(\hat{\rho}_0),
        \end{split}
    \end{align}
    we obtain Eq.~\eqref{M:eq_cors1_pr_4} with $d_i = c_i$.
\end{proof}

\subsection{{Details of Unbiased Estimability in Cycle Benchmarking for the Pauli-Z Rotation Gate}}\label{methods:cycle}
In the main text, we discussed the {unbiased estimation of noise parameters} affecting the Pauli-Z rotation gate using cycle benchmarking~\cite{learnability-erhard2019characterizing}.
In this section, we provide a more detailed analysis.

As in the main text, we consider {unbiased estimation of noise parameters} on the Pauli-$Z$ rotation gate $\hat{U} = \hat{R}_z(\phi) = e^{-i\phi/2 \hat{Z}}$ using cycle benchmarking.
As shown in Sec.~V, the noise channel $\mathcal{N}$ affecting $U$ can be twirled, such that its Pauli transfer matrix is expressed as
\begin{align}
    \begin{pmatrix}
        1 & 0 & 0 & 0 \\
        0 & \lambda_1\cos\theta & \lambda_1\sin\theta & 0 \\
        0 & -\lambda_1\sin\theta & \lambda_1\cos\theta & 0 \\
        \alpha & 0 & 0 & \lambda_2 \\
    \end{pmatrix},
\end{align}
which is parameterized by four independent parameters: $\lambda_1, \lambda_2, \alpha$, and $\theta$.
The noise channel $\mathcal{N}$ represents a mixture of depolarizing noise, dephasing noise, amplitude damping noise, and coherent noise due to over-rotation.

We consider {estimation of} the twirled noise $\mathcal{N}$ using cycle benchmarking under state preparation and measurement (SPAM) errors.
We assume that the state preparation error $\mathcal{N}_S$ and measurement error $\mathcal{N}_M$ are Pauli-twirled to Pauli noise, whose Pauli transfer matrices are represented as
\begin{align}
    \begin{pmatrix}
        1 & 0 & 0 & 0 \\
        0 & \lambda_{1S} & 0 & 0 \\
        0 & 0 & \lambda_{2S} & 0 \\
        0 & 0 & 0 & \lambda_{3S} \\
    \end{pmatrix},
    \quad
    \begin{pmatrix}
        1 & 0 & 0 & 0 \\
        0 & \lambda_{1M} & 0 & 0 \\
        0 & 0 & \lambda_{2M} & 0 \\
        0 & 0 & 0 & \lambda_{3M} \\
    \end{pmatrix}.
\end{align}
In cycle benchmarking, we repeatedly apply the noisy $\hat{R}_z(\phi)$ gate.
This corresponds to analyzing the channel $\mathcal{N}_{d} = \mathcal{N}_M \circ (\mathcal{N} \circ \mathcal{U})^{\circ d} \circ \mathcal{N}_S$ for various $d$, whose Pauli transfer matrix $A_d$ is given by
\begin{widetext}
\begin{align}
    \begin{pmatrix}
        1 & 0 & 0 & 0 \\
        0 & \lambda_{1M}\lambda_{1S}\lambda_1^d\cos d\theta' & \lambda_{1M}\lambda_{2S}\lambda_1^d \sin d\theta' & 0 \\
        0 & -\lambda_{2M}\lambda_{1S}\lambda_1^d\sin d\theta' & \lambda_{2M}\lambda_{2S}\lambda_1^d \cos d\theta' & 0 \\
        \lambda_{3M}\alpha\frac{1-\lambda_1^d}{1-\lambda_1} & 0 & 0 & \lambda_{3M}\lambda_{3S}\lambda_2^d \\
    \end{pmatrix}.
\end{align}
\end{widetext}
Here, $\theta' = \phi + \theta$ and $(A_d)_{ij}$ denotes the $(i+1,j+1)$-th component of $A_d$.

To encapsulate $\mathcal{N}_{d}$ for all integer values of $d$ in a single description, we introduce an augmented channel $\mathcal{N}_{\mathrm{cycle}}$, and a classical register $\hat{\rho}_{c}$. This register keeps track of the number of applied gates $d$. Formally, we define
\begin{align}
    \mathcal{N}_{\mathrm{cycle}} : \hat{\rho} \otimes \hat{\rho}_c \mapsto \sum_{d}\mathcal{N}_{d}(\hat{\rho}) \otimes \ketbra{d}\hat{\rho}_c\ketbra{d}.
\end{align}
In this representation, each outcome $\ketbra{d}$ in the register corresponds to the channel $\mathcal{N}_{d}$ acting on the initial state $\hat{\rho}$.

Since there is a one-to-one correspondence between a quantum channel and its Choi state, we consider its Choi state defined as
\begin{align}
    \begin{split}
    &\hat{\rho}_{\mathrm{cycle}} = \mathcal{N}_{\mathrm{cycle}} \otimes \mathcal{I}(\ketbra{\Psi}\otimes\ketbra{\Psi_D}) \\
    &= \frac{1}{D} \sum_{d=0}^D \hat{\rho}_d \otimes \ketbra{dd}.
    \end{split}
\end{align}
Here, $\ketbra{\Psi}$ and $\ketbra{\Psi_D}$ are maximally entangled states for single-qubit and $D$-dimensional systems, respectively, and $\hat{\rho}_d$ is the Choi state of the channel $\mathcal{N}_{d}$, defined as
\begin{align}
    \hat{\rho}_d = \mathcal{N}_{d}\otimes \mathcal{I}(\ketbra{\Psi}).
\end{align}
Since the maximally entangled state $\ketbra{\Psi}$ can be represented as
\begin{align}
    \ketbra{\Psi} = \frac{1}{4} \sum_{i=0}^{3} c_i P_i\otimes P_i,
\end{align}
where $c_i = \pm1$, $P_0 = I$, $P_1 = X$, $P_2 = Y$, and $P_3 = Z$, we can represent $\hat{\rho}_d$ as
\begin{align}
    \hat{\rho}_d 
    &= \frac{1}{4}\sum_{i=0}^{3}\sum_{j=0}^{3}c_j(A_{d})_{ij}P_i\otimes P_j.
\end{align}

Let us show that a {locally} unbiased estimator of $\alpha$ cannot be obtained from $\hat{\rho}_{\mathrm{cycle}}$.
For the parameters $\alpha, \lambda_{3M}, \lambda_{3S}$, we have
\begin{align}
    &\alpha \partial_\alpha \hat{\rho}_{\mathrm{cycle}} = \frac{1}{D} \sum_{d=1}^D c_0(A_{d})_{30}P_3\otimes P_0\otimes \ketbra{dd}, \\
    &\lambda_{3M} \partial_{\lambda_{3M}} \hat{\rho}_{\mathrm{cycle}}\\
    &= \frac{1}{D} \sum_{d=1}^D (c_0(A_{d})_{30}P_3\otimes P_0 + c_3(A_{d})_{33}P_3\otimes P_3)\otimes \ketbra{dd}, \\
    &\lambda_{3S} \partial_{\lambda_{3S}} \hat{\rho}_{\mathrm{cycle}} = \frac{1}{D} \sum_{d=1}^D (c_3(A_{d})_{33}P_3\otimes P_3)\otimes \ketbra{dd}.
\end{align}
Therefore, we obtain
\begin{align}
    \alpha \partial_\alpha \hat{\rho}_{\mathrm{cycle}} = \lambda_{3M} \partial_{\lambda_{3M}} \hat{\rho}_{\mathrm{cycle}} - \lambda_{3S} \partial_{\lambda_{3S}} \hat{\rho}_{\mathrm{cycle}},
\end{align}
which is equivalent to
\begin{align}
    \alpha \partial_\alpha \mathcal{N}_{\mathrm{cycle}} = \lambda_{3M} \partial_{\lambda_{3M}} \mathcal{N}_{\mathrm{cycle}} - \lambda_{3S} \partial_{\lambda_{3S}} \mathcal{N}_{\mathrm{cycle}}.
\end{align}
Thus, by Corollary~\ref{corollary1}, a {locally} unbiased estimator for the parameter $\alpha$ cannot be derived from $\mathcal{N}_{\mathrm{cycle}}$. 

Next, we show that $\lambda_1, \lambda_2, \theta'$ {admit locally unbiased estimators within cycle benchmarking model}.
It is straightforward to observe that $\lambda_2$ can be {unbiasedly estimated} by calculating $(A_{d+1})_{33}/ (A_d)_{33}$, so we focus on parameters $\lambda_1$ and $\theta'$.
We can determine $\lambda_{1M}\lambda_{1S}$ and $\lambda_{2M}\lambda_{2S}$ by calculating $(A_0)_{11}$ and $(A_0)_{22}$.
Additionally, the ratio $\frac{\lambda_{1M}}{\lambda_{2M}}\frac{\lambda_{2S}}{\lambda_{1S}}$ can be obtained by taking the ratio of $(A_d)_{12}$ and $(A_d)_{21}$.
From these values, we can derive $\frac{\lambda_{1M}}{\lambda_{2M}}$ and $\frac{\lambda_{2S}}{\lambda_{1S}}$ and subsequently convert $(A_d)_{12} = \lambda_{1M}\lambda_{2S}\lambda_1^d\sin d \theta'$ into $\lambda_{1M}\lambda_{1S}\lambda_1^d\sin d\theta'$.
Taking the ratio of this expression with $(A_d)_{11} = \lambda_{1M}\lambda_{1S}\lambda_1^d\cos d\theta'$, we can determine $\theta'$.
Since we can {unbiasedly estimate} $\lambda_{1M}\lambda_{1S}$ and $\theta'$, we can also extract $\lambda_1$ from $(A_d)_{11} = \lambda_{1M}\lambda_{1S}\lambda_1^d\cos d\theta'$.
We can also {determine} $\lambda_{1M}\lambda_{2S}$ and $\lambda_{2M}\lambda_{1S}$ from $(A_d)_{12} = \lambda_{1M}\lambda_{2S}\lambda_1^d\sin d \theta'$ and  $(A_d)_{21} = -\lambda_{2M}\lambda_{1S}\lambda_1^d\sin d\theta'$.
Combining these {estimated} values with $\frac{\lambda_{1M}}{\lambda_{2M}}$ and $\frac{\lambda_{2S}}{\lambda_{1S}}$, we can individually {unbiasedly estimate} $\lambda_{1S}$, $\lambda_{2S}$, $\lambda_{1M}$, and $\lambda_{2M}$.

In conclusion, we can {unbiasedly estimate} $\lambda_1,\lambda_2,\theta$, but not $\alpha$.
When considering SPAM errors, we can {unbiasedly estimate} $\lambda_1$, $\lambda_2$, $\theta$, $\lambda_{1S}$, $\lambda_{2S}$, $\lambda_{1M}$, $\lambda_{2M}$, but cannot {unbiasedly estimate} $\alpha, \lambda_{3S}, \lambda_{3M}$ individually; only the products $\lambda_{3M}\alpha$ and $\lambda_{3M}\lambda_{3S}$ can be {unbiasedly estimate}.

\bibliography{Reference.bib}

\let\oldaddcontentsline\addcontentsline
\renewcommand{\addcontentsline}[3]{}
\let\addcontentsline\oldaddcontentsline
\onecolumngrid

\clearpage
\begin{center}
	\Large
	\textbf{Supplemental Material: Criteria for unbiased estimation: applications to noise-agnostic sensing and {quantum channel estimation}}
\end{center}

\setcounter{equation}{0}
\setcounter{figure}{0}
\setcounter{table}{0}

\setcounter{page}{1}
\renewcommand{\thesection}{S\arabic{section}}
\renewcommand{\theequation}{S\arabic{equation}}
\renewcommand{\thefigure}{S\arabic{figure}}
\renewcommand{\thetable}{S\arabic{table}}

\addtocontents{toc}{\protect\setcounter{tocdepth}{0}}

\setcounter{section}{0}

\section{Fully Generalized multi-parameter Cram\'er-Rao matrix inequality}
In this section, we derive the \textit{fully generalized quantum Cramér-Rao matrix inequality}, which serves as the fundamental relation for multi-parameter estimation under non-invertible quantum Fisher information matrix.
While the main text focuses on the {locally} unbiased estimation scenario, the supplemental material extends the discussion to include biased estimation cases. In this section, we first introduce the quantum Cram\'er-Rao matrix inequality.
We then derive the \textit{perturbed quantum Cram\'er-Rao matrix inequality} which serves as a foundational inequality for deriving the fully generalized quantum Cramér-Rao matrix inequality.
Building on this bound, we then proceed to derive the fully generalized quantum Cramér-Rao matrix inequality.

\subsection{Multi-parameter quantum estimation}
Let us consider a quantum state $\hat{\rho}_{\vb*{\theta}}$ that encodes multi-parameters $ \vb*{\theta} :=(\theta_{1},\theta_{2},\cdots,\theta_{M})^{\mathrm{T}}$. To estimate $\vb*{\theta}$, a measurement is performed, characterized by a positive-operator-valued measure (POVM), denoted as $\{\hat{\Pi}_{\vb*{x}}\}_{\vb*{x}}$, where each operator $\hat{\Pi}_{\vb*{x}}$ corresponds to a possible measurement outcome $\vb*{x}$. Based on the measurement outcome $\vb*{x}$, an estimator $\tilde{\theta}_{i}=\tilde{\theta}_{i}(\vb*{x})$ for each parameter $\theta_{i}$ is constructed. The corresponding estimation error can be characterized by $M \times M$ error matrix $\mathbf{\Sigma}$, with elements defined as
\begin{align}
    [\mathbf{\Sigma}]_{ij}=\int dx ~ p_{\vb*{\theta}}(\vb*{x})\left(\tilde{\theta}_{i}-\theta_{i}\right)\left(\tilde{\theta}_{j}-\theta_{j}\right).
\end{align}
Here $p_{\vb*{\theta}}(\vb*{x}) :=\mathrm{Tr}[\hat{\Pi}_{\vb*{x}}\hat{\rho}_{\vb*{\theta}}]$ represents the probability of obtaining the measurement outcome $\vb*{x}$ for given parameters $\vb*{\theta}$.
We note that the error matrix can be decomposed into the sum of the covariance matrix $\mathbf{C}$ and the bias matrix $\mathbf{B}$ as $\mathbf{\Sigma}=\mathbf{C}+\mathbf{B}$, where the elements of the covariance matrix and the bias matrix are defined as
\begin{align}
    [\mathbf{C}]_{ij}=\int dx ~ p_{\vb*{\theta}}(\vb*{x})\left(\tilde{\theta}_{i}-\big< \tilde{\theta}_{i}\big>\right)\left(\tilde{\theta}_{j}-\big<\tilde{\theta}_{j}\big>\right), ~~[\mathbf{B}]_{ij}= \left(\big< \tilde{\theta}_{i}\big>- \theta_{i}\right)\left(\big< \tilde{\theta}_{j}\big>- \theta_{j}\right).
\end{align}
Here $\big< (\cdot) \big>$ is the average over all possible measurement outcome $\vb*{x}$ defined as $\big< (\cdot) \big>:=\int d\vb*{x}  ~ p(\vb*{x})(\cdot)$. This decomposition provides a clear distinction between the contributions to the estimation error from the statistical variance of the estimator (captured by $\mathbf{C}$) and from its bias (captured by $\mathbf{B}$).

When all the estimators satisfy the local unbiased condition (LUC) that is
\begin{align}
    \big<\tilde{\theta}_{i}(\vb*{x})\big>=\theta_{i},~\partial_{i}\big<\tilde{\theta}_{j}\big>=\delta_{ij}, \label{S:luc}
\end{align}
the multi-parameter quantum Cramér-Rao matrix inequality provides a lower bound for the error matrix:
\begin{align}
    \mathbf{\Sigma}= \mathbf{C} \succeq \mathbf{F}^{-1} \succeq \mathbf{J}^{-1}, \label{S:QCRB}
\end{align}
where $\mathbf{F}$ is the Fisher information matrix, and $\mathbf{J}$ is the quantum Fisher information matrix (QFIM) which are defined as
\begin{align}
    [\mathbf{F}]_{ij}=\int d\vb*{x} ~ \frac{\partial_{i} p_{\vb*{\theta}}(\vb*{x})\partial_{j} p_{\vb*{\theta}}(\vb*{x})}{p_{\vb*{\theta}}(\vb*{x})},~ [\mathbf{J}]_{ij}:=\mathrm{Tr}\big[\hat{\rho}_{\vb*{\theta}}\{\hat{L}_{i},\hat{L}_{j}\}{/2}\big].
\end{align}
Here $\hat{L}_{i}$ denotes the symmetric logarithmic derivative (SLD) operator satisfying $\pdv{\hat{\rho}_{\vb*{\theta}}}{\theta_{i}}=\frac{1}{2}\{\hat{L}_{i},\hat{\rho}_{\vb*{\theta}}\}$. While the quantum Cramér-Rao matrix inequality gives a lower bound for the error matrix, we emphasize that the validity of the quantum Cramér-Rao matrix inequality highly depends on the following two assumptions: (1) All the estimators should satisfy the LUC. (2) the QFIM $\mathbf{J}$ is invertible. To address these limitations, we derive the fully generalized quantum Cramér-Rao matrix inequality, which does not require these two assumptions yet still establishes a lower bound for the error matrix.

\subsection{Perturbed quantum Cram\'er-Rao matrix inequality}
To derive the fully generalized quantum Cramér-Rao matrix inequality, we first introduce the \textit{perturbed Cram\'er-Rao matrix inequality}. 

We begin by considering the diagonal form of $\mathbf{J}$. Since $\mathbf{J}$ is positive semi-definite, it can be expressed as
\begin{align}
    \mathbf{J}= \sum_{i=1}^{s} \lambda_{i} \vb*{\lambda}_{i}\vb*{\lambda}^{\mathrm{T}}_{i} + \sum_{j=s+1}^{M} 0 \vb*{\lambda}_{j}\vb*{\lambda}^{\mathrm{T}}_{j}, \label{S:diagJ}
\end{align}
where $\{\lambda_{i}\}_{i=1}^{s}$ are the positive eigenvalues, and $\{\vb*{\lambda}_{i}\}_{i=1}^{M}$ are the eigenvectors. Next, we define the \textit{perturbed FIM} $\mathbf{F}_{\vb*{\varepsilon}}$, and the \textit{perturbed QFIM} $\mathbf{J}_{\vb*{\varepsilon}}$ as 
\begin{align}
    \mathbf{F}_{\vb*{\varepsilon}}= \mathbf{F}+\varepsilon_{\mathrm{s}}\mathbf{I}_{s} + \varepsilon_{\mathrm{n}}\mathbf{I}_{n},~~\mathbf{J}_{\vb*{\varepsilon}}= \mathbf{J}+\varepsilon_{\mathrm{s}}\mathbf{I}_{\mathrm{s}} + \varepsilon_{\mathrm{n}}\mathbf{I}_{\mathrm{n}},~\varepsilon_{\mathrm{s}},\varepsilon_{\mathrm{n}}>0. \label{S:defepJ}
\end{align}
Here $\mathbf{I}_{s}:=\sum_{i=1}^{s} \vb*{\lambda}_{i}\vb*{\lambda}^{\mathrm{T}}_{i}$ and $\mathbf{I}_{n}:=\sum_{j=s+1}^{M} \vb*{\lambda}_{j}\vb*{\lambda}^{\mathrm{T}}_{j}$ represent the projection matrices onto the support and the null space of $\mathbf{J}$ respectively, and $\vb*{\varepsilon} = (\varepsilon_{\mathrm{s}}, \varepsilon_{\mathrm{n}})$. These definitions ensure that both $\mathbf{F}_{\vb*{\varepsilon}}$ and $\mathbf{J}_{\vb*{\varepsilon}}$ are invertible.
Using the perturbed FIM, we derive the following matrix inequality, defining the perturbed Cramér-Rao matrix inequality:
\begin{align}
    \mathbf{\Sigma}=\mathbf{C}+\mathbf{B} \succeq \mathbf{H}^{\mathrm{T}}\mathbf{F}^{-1}_{\vb*{\varepsilon}} \mathbf{H} + \mathbf{B}, 
\end{align}
where $\mathbf{H}$ is the Jacobian of the estimators’ averaged values, defined as $[\mathbf{H}]_{ij} := \partial_{i} \big< \tilde{\theta}_{j} \big>$.
\begin{proof}
We begin by considering the inequality:
    \begin{align}
    &\left(\vb*{w}^{\mathrm{T}}\mathbf{C}\vb*{w}\right)\left(\vb*{v}^{\mathrm{T}} \mathbf{F}_{\vb*{\varepsilon}} \vb*{v} \right)\geq \left(\vb*{w}^{\mathrm{T}}\mathbf{C}\vb*{w}\right)\left(\vb*{v}^{\mathrm{T}} \mathbf{F} \vb*{v} \right)\\
    &=\left[\int dx ~ \sum_{i,j=1}^{{M}} w_{i}p_{\vb*{\theta}}(x)\left(\tilde{\theta}_{i}- \big< \tilde{\theta}_{i} \big>\right)\left(\tilde{\theta}_{j}- \big< \tilde{\theta}_{j} \big>\right)w_{j}\right] \left[\int dx ~ \sum_{i',j'=1}^{{M}} v_{i'} \frac{\partial_{i'} p_{\vb*{\theta}}(x)\partial_{j'} p_{\vb*{\theta}}(x)}{p_{\vb*{\theta}}(x)}v_{j'}\right]\\
    &= \int dx ~ \left[\sum_{i=1}^{{M}}w_{i} \sqrt{p_{\vb*{\theta}}(x)}\left(\tilde{\theta}_{i}- \big< \tilde{\theta}_{i} \big>\right)\right] \left[\sum_{j=1}^{{M}}w_{j} \sqrt{p_{\vb*{\theta}}(x)}\left(\tilde{\theta}_{j}- \big< \tilde{\theta}_{j} \big>\right)\right] \int dx'~\left[\sum_{i'=1}^{{M}}v_{i'}\frac{\partial_{i'} p_{\vb*{\theta}}(x)}{\sqrt{p_{\vb*{\theta}}(x)}}\right]\left[\sum_{j'=1}^{{M}}v_{j'}\frac{\partial_{j'} p_{\vb*{\theta}}(x)}{\sqrt{p_{\vb*{\theta}}(x)}}\right]\\
    &\geq \left[\int dx \left(\sum_{i=1}^{{M}}w_{i}\left(\tilde{\theta}_{i}- \big< \tilde{\theta}_{i} \big>\right)\right)\left(\sum_{j=1}^{{M}}v_{j}\partial_{j} p_{\vb*{\theta}}(x)\right) \right]^{2}
    = \left[\sum_{i,j}^{M}w_{i} \left(\int dx ~ \tilde{\theta}_{i} \partial_{j} p_{\vb*{\theta}}(x) \right) v_{j} \right]^{2} \equiv \left(\vb*{w}^{\mathrm{T}}  \mathbf{H}^{\mathrm{T}} \vb*{v}\right)^{2}. \label{S:regula}
\end{align}
The first inequality follows from $\vb*{v}^{\mathrm{T}}\left(\varepsilon_{\mathrm{s}}\mathbf{I}_{s}+\varepsilon_{\mathrm{n}}\mathbf{I}_{n}\right)\vb*{v}>0$, and the second inequality can be obtained from the Cauchy-Schwarz inequality. Here, we assume the regularity condition $\int dx~\partial_{i}p_{\vb*{\theta}}(x)=0~\forall i$ to obtain the first equality in Eq. \eqref{S:regula}.
Here we note that we have not made any assumptions about $\vb*{w}$ and $\vb*{v}$. Therefore, by setting $\vb*{v} = \mathbf{F}^{-1}_{\vb*{\varepsilon}}  \mathbf{H}  \vb*{w}$, we obtain
\begin{align}
    \left(\vb*{w}^{\mathrm{T}}  \mathbf{C}  \vb*{w}\right) \left(\vb*{w}^{\mathrm{T}}  \mathbf{H}^{\mathrm{T}}  \mathbf{F}^{-1}_{\vb*{\varepsilon}} \mathbf{F}_{\vb*{\varepsilon}}  \mathbf{F}^{-1}_{\vb*{\varepsilon}}  \mathbf{H} \vb*{w}\right) = \left(\vb*{w}^{\mathrm{T}}  \mathbf{C}  \vb*{w}\right)\left(\vb*{w}^{\mathrm{T}}  \mathbf{H}^{\mathrm{T}}  \mathbf{F}^{-1}_{\vb*{\varepsilon}} \mathbf{H}  \vb*{w}\right) \geq \left(\vb*{w}^{\mathrm{T}}  \mathbf{H}^{\mathrm{T}}  \mathbf{F}^{-1}_{\vb*{\varepsilon}} \mathbf{H}  \vb*{w}\right)^{2}~\rightarrow ~ \vb*{w}^{\mathrm{T}}  \mathbf{C}  \vb*{w} \geq \vb*{w}^{\mathrm{T}}  \mathbf{H}^{\mathrm{T}}  \mathbf{F}^{-1}_{\vb*{\varepsilon}} \mathbf{H}  \vb*{w}. \label{S:dervPCRB} 
\end{align}
Since $\vb*{w}$ is an arbitrary vector, Eq. \eqref{S:dervPCRB} implies the following matrix inequality:
\begin{align}
    \mathbf{C} \succeq \mathbf{H}^{\mathrm{T}}\mathbf{F}^{-1}_{\vb*{\varepsilon}} \mathbf{H},~~\text{or equivalently,~~} \mathbf{\Sigma} \succeq \mathbf{H}^{\mathrm{T}}\mathbf{F}^{-1}_{\vb*{\varepsilon}} \mathbf{H}+\mathbf{B}. \label{S:epsilonCRB}
\end{align}
\end{proof}
Next, we derive the perturbed quantum Cram\'er-Rao matrix inequality which is
\begin{align}
    \mathbf{\Sigma}=\mathbf{C}+\mathbf{B} \succeq \mathbf{H}^{\mathrm{T}}\mathbf{F}^{-1}_{\vb*{\varepsilon}} \mathbf{H} + \mathbf{B}\succeq \mathbf{H}^{\mathrm{T}}\mathbf{J}^{-1}_{\vb*{\varepsilon}} \mathbf{H} + \mathbf{B}, \label{S:epsilonQCRB}
\end{align}
\begin{proof}
    Since $\mathbf{F} \preceq \mathbf{J}$, it follows that $\mathbf{F}_{\vb*{\varepsilon}} \preceq \mathbf{J}_{\vb*{\varepsilon}}$. Thus, the invertibility of $\mathbf{F}_{\vb*{\varepsilon}}$ and $\mathbf{J}_{\vb*{\varepsilon}}$ implies $\mathbf{F}^{-1}_{\vb*{\varepsilon}} \succeq \mathbf{J}^{-1}_{\vb*{\varepsilon}}$. Combining this with the perturbed Cram\'er-Rao matrix inequality in Eq. \eqref{S:epsilonCRB}, we obtain the perturbed quantum Cram\'er-Rao matrix inequality in Eq. \eqref{S:epsilonQCRB}.
\end{proof}

\subsection{Fully generalized quantum Cramér-Rao matrix inequality}
To derive the fully generalized quantum Cramér-Rao matrix inequality from the perturbed quantum Cram\'er-Rao matrix inequality, we begin by revisiting the perturbed quantum Cram\'er-Rao matrix inequality which is expressed as:
\begin{align}
    \mathbf{\Sigma} \succeq \mathbf{H}^{\mathrm{T}}\mathbf{J}^{-1}_{\vb*{\varepsilon}} \mathbf{H} + \mathbf{B}. \label{S:simpleepsilonQCRB}
\end{align}
We emphasize that this inequality holds for any positive values of the perturbation parameters $\varepsilon_{\mathrm{s}}>0$ and $\varepsilon_{\mathrm{n}}>0$. From the diagonal form of $\mathbf{J}$ and the definition of $\mathbf{J}_{\vb*{\varepsilon}}$, one can easily find that the inverse of $\mathbf{J}_{\vb*{\varepsilon}}$ can be expressed as 
\begin{align}
    \mathbf{J}^{-1}_{\vb*{\varepsilon}}=\sum_{i=1}^{s} \frac{1}{\lambda_{i}+\varepsilon_{\mathrm{s}}} \vb*{\lambda}_{i}\vb*{\lambda}_{i}^{\mathrm{T}} + \sum_{j=s+1}^{M} \frac{1}{\varepsilon_{\mathrm{n}}} \vb*{\lambda}_{j}\vb*{\lambda}_{j}^{\mathrm{T}}.
\end{align}
Next, let us consider $\mathbf{J}^{-1}_{\vb*{\varepsilon}}$ in the following limit:
\begin{align}
    \lim_{\varepsilon_{\mathrm{s}}\to 0}\lim_{\varepsilon_{\mathrm{n}} \to \infty}\mathbf{J}^{-1}_{\vb*{\varepsilon}}=\sum_{i=1}^{s} \frac{1}{\lambda_{i}} \vb*{\lambda}_{i}\vb*{\lambda}_{i}^{\mathrm{T}} + \sum_{j=s+1}^{M} 0 \vb*{\lambda}_{j}\vb*{\lambda}_{j}^{\mathrm{T}} = \mathbf{J}^{+}
\end{align}
where $\mathbf{J}^{+}$ is the Moore-Penrose pseudo inverse of $\mathbf{J}$. Applying the limit ${\varepsilon_{\mathrm{s}}\to 0}$, and ${\varepsilon_{\mathrm{n}}} \to \infty$ to the perturbed quantum Cram\'er-Rao matrix inequality in Eq. \eqref{S:simpleepsilonQCRB}, we obtain the following matrix inequality, which we define as the fully generalized quantum Cram\'er-Rao matrix inequality:
\begin{align}
    \mathbf{\Sigma} \succeq \mathbf{H}^{\mathrm{T}}\mathbf{J}^{+} \mathbf{H} + \mathbf{B}, \label{S:FGQCRM}
\end{align}
where $\mathbf{J}^{+}=\sum_{i=1}^{s} \frac{1}{\lambda_{i}} \vb*{\lambda}_{i}\vb*{\lambda}_{i}^{\mathrm{T}} + \sum_{j=s+1}^{M} 0 \vb*{\lambda}_{j}\vb*{\lambda}_{j}^{\mathrm{T}}$ is the Moore-Penrose pseudo inverse of $\mathbf{J}$.
Finally, from Eq. \eqref{S:FGQCRM}, we establish the following inequality:
\begin{align}
    \vb*{w}^{\mathrm{T}}\mathbf{\Sigma}\vb*{w} \succeq \vb*{w}^{\mathrm{T}}\mathbf{H}^{\mathrm{T}}\mathbf{J}^{+} \mathbf{H}\vb*{w} + \vb*{w}^{\mathrm{T}}\mathbf{B}\vb*{w}. \label{S:FGQCRB}
\end{align}
We note that if the parameter of interest $\phi=\vb*{w}^{\mathrm{T}}\vb*{\theta}$ satisfy the LUC given by Eq. \eqref{S:luc}, equivalently, $\mathbf{H}\vb*{w}=\vb*{w}$ and $\vb*{w}^{\mathrm{T}}\mathbf{B}\vb*{w}=0$, this bound reduces to \textit{generalized quantum Cramér-Rao bound}:
\begin{align}
    \vb*{w}^{\mathrm{T}}\mathbf{C}\vb*{w} \succeq \vb*{w}^{\mathrm{T}}\mathbf{J}^{+} \vb*{w}. \label{S:GQCRB}
\end{align}

\section{Alternative proof of the achievability of the generalized quantum Cramér-Rao bound}
In this section, we discuss when the equality in the generalized quantum Cramér-Rao bound Eq.~\eqref{S:GQCRB} can be achieved, i.e., when we can perform a locally unbiased estimation with finite variance.
Although the necessary and sufficient condition was originally proved in Refs.~\cite{lemma-PhysRevX.10.031023, albarelli2019upper}, we offer an alternative derivation based on the theory of nuisance parameters~\cite{est-Suzuki_2020}.

We now formally state the results as the following lemma:
\setcounter{lemma}{0}
\begin{lemma}\label{S:lemma1}
There exists a locally unbiased estimator of $\phi=\vb*{w}^{\mathrm{T}}\vb*{\theta}$ with finite estimation error if and only if $\vb*{w} \in \mathrm{supp}\left(\mathbf{J}\right)$.
In other words
\begin{align}
    \mathbf{J}^{+}\mathbf{J}\vb*{w}  = \vb*{w}. \label{S:NScondit1_revised}
\end{align}
Furthermore, when Eq.~\eqref{S:NScondit1_revised} is satisfied, the achievable lower bound on the estimation error of $\phi=\vb*{w}^{\mathrm{T}}\vb*{\theta}$ is given by generalized quantum Cramér-Rao bound:
\begin{align}
    \Delta^{2}\phi=\vb*{w}^{\mathrm{T}}\mathbf{C}\vb*{w} \geq \vb*{w}^{\mathrm{T}}\mathbf{J}^{+}\vb*{w}. \label{S:lemmaeq}
\end{align}
\end{lemma}

For simplicity, throughout the remainder of this manuscript, we assume that $\theta_{1}=\vb*{w}_{1}^{\mathrm{T}}\vb*{\theta}$ is the parameter of interest where $\vb*{w}_{1}:=(1,0,0,\cdots,0)^{\mathrm{T}}$. In other words, we treat the remaining parameters $\theta_{2},\theta_{3},\cdots,\theta_{M}$ as nuisance parameters~\cite{est-Suzuki_2020}. It is worth emphasizing, however, that our analysis can be generalized to the estimation of an arbitrary parameter $\phi=\phi(\vb*{\theta})$.
Therefore, instead of Lemma~\ref{S:lemma1}, we prove the following simplified lemma:
\begin{lemma}\label{S:lemma2}
There exists a locally unbiased estimator of $\theta_1$ at $\vb*{\theta} = \vb*{\theta}_{0}$ with finite estimation error if and only if $\vb*{w}_1 :=(1,0,0,\cdots,0)^{\mathrm{T}} \in \mathrm{supp}\left(\left.\mathbf{J}\right|_{\vb*{\theta} = \vb*{\theta}_0}\right)$.
In other words
\begin{align}
    \left.(\mathbf{J}^{+}\mathbf{J})\right|_{\vb*{\theta} = \vb*{\theta}_0}\vb*{w}_1  = \vb*{w}_1. \label{S:NScondit1_revised_w1}
\end{align}
Furthermore, when Eq.~\eqref{S:NScondit1_revised_w1} is satisfied, the achievable lower bound on the estimation error of $\theta_1$ is given by generalized quantum Cramér-Rao bound:
\begin{align}
    \Delta^{2}\theta_1 = \big<(\tilde{\theta}_{1}-\theta_{1})^{2}\big>=\vb*{w}^{\mathrm{T}}_{{1}}\mathbf{C}\vb*{w}_{{1}} \geq \left.([\mathbf{J}^{+}]_{11})\right|_{\vb*{\theta} = \vb*{\theta}_0}. \label{S:lemmaeq_w1}
\end{align}
Here, $\vb*{\theta}_0=(\theta_{1,0},\theta_{2,0},\cdots,\theta_{M,0})^{\mathrm{T}}$ is a fixed point, and the locally unbiased estimator of $\theta_1$ at $\vb*{\theta} = \vb*{\theta}_0$ is an estimator $\tilde{\theta}_1$ that satisfies
\begin{align}
    \big<\tilde{\theta}_{1}\big> =\theta_{1},~\partial_{i}\big<\tilde{\theta}_{1}\big>=\delta_{i1} \quad (i = 1,\ldots,M) \label{S:locallyunbiased_1}
\end{align}
at $\vb*{\theta} = \vb*{\theta}_0$.
\end{lemma}

In the following proof, we present the optimal estimator and the POVM that achieve equality in Eq. \eqref{S:lemmaeq_w1}. Specifically, the optimal estimator is provided in Eq. \eqref{S:optimalest}, and the optimal POVM is constructed from the eigenbasis of the SLD operator defined in Eq. \eqref{S:optimalSLD}.

\subsection{Proof of \textit{if} direction}
We note that when $\mathbf{J}$ is invertible, the achievability of the quantum Cramér-Rao bound has already been studied particularly in Ref. \cite{est-Suzuki_2020}. In this work, building on the proof and results of Ref. \cite{est-Suzuki_2020}, we extend the analysis to establish the achievability of the generalized quantum Cramér-Rao bound in Eq. \eqref{S:lemmaeq}, even when $\mathbf{J}$ is not invertible.

\begin{proof}
Let us assume  $\vb*{w}_1 \in \mathrm{supp}\left(\left.\mathbf{J}\right|_{\vb*{\theta} = \vb*{\theta}_0}\right)$. We prove the if direction in three steps. In \textit{Step 1}, we provide a brief explanation of the re-parametrization process $\vb*{\theta}\mapsto \vb*{\xi}$ and the corresponding QFIM. In \textit{Step 2}, we consider a new parameter set $\vb*{\xi}$ that block-diagonalizes $\mathbf{J}(\vb*{\xi})$. In \textit{Step 3}, leveraging the results in \textit{Step 2}, we construct the estimator and the optimal measurement that achieve the equality in the generalized quantum Cramér-Rao bound in Eq. \eqref{S:lemmaeq}.

\textit{Step 1. Re-parametrization.---}
We introduce the following change of variables $\vb*{\theta}=(\theta_{1},\vb*{\theta}_{\mathrm{n}})\mapsto \vb*{\xi}=(\xi_{1},\vb*{\xi}_{\mathrm{n}})$ (where $\vb*{\theta}_{\mathrm{n}}=(\theta_{2},\theta_{3},\cdots,\theta_{M})^{\mathrm{T}}$, and $\vb*{\xi}_{\mathrm{n}}=(\xi_{2},\xi_{3},\cdots,\xi_{M})^{\mathrm{T}}$) such that
\begin{align}
    \theta_{1}=\xi_{1},~\theta_{2}=\theta_{2}(\xi_{1},\xi_{2}),~\cdots,~ \theta_{M}=\theta_{M}(\xi_{1},\xi_{2},\cdots,\xi_{M}). \label{S:reparameter}
\end{align}
We note that QFIM typically depends on a chosen parameter set. In the previous section, however, we did not explicitly denote the QFIM $\mathbf{J}$ as a function of $\vb*{\theta}$, as the discussion was only confined to the specific parameter set $\vb*{\theta}$. Consequently, there was no ambiguity in the representation. In this section, however, we consider two distinct parameter sets, $\vb*{\theta}$ and $\vb*{\xi}$. To avoid confusion, we explicitly denote the QFIMs corresponding to these parameter sets as $\mathbf{J}(\vb*{\theta})$ and $\mathbf{J}(\vb*{\xi})$, respectively. The relationship between these two QFIMs is given by
\begin{align}
    \mathbf{J}(\vb*{\xi}) =\mathbf{T} \mathbf{J}(\vb*{\theta}) \mathbf{T}^{\mathrm{T}}, \label{S:newpara}
\end{align}
where $\mathbf{T}$ is the transformation matrix with its components are given by $[\mathbf{T}]_{ij}:=\pdv{\theta_{j}}{\xi_{i}}$. We observe that $\mathbf{T}$ is an upper triangular matrix since $\theta_{i}$ does not depend on $\xi_{j>i}$ for all $1 \leq i \leq M$. (See Eq. \eqref{S:reparameter}.) Furthermore, we restrict our analysis to re-parameterizations that maintain a \textit{one-to-one correspondence} between $\vb*{\theta}$ and $\vb*{\xi}$. This ensures that $\mathbf{T}$ is invertible, preserving the integrity of the transformation. 

For the further proof, we represent $\mathbf{J}(\vb*{\theta})$ as 
\begin{align}
    \mathbf{J}(\vb*{\theta})= \begin{pmatrix}
        [\mathbf{J}(\vb*{\theta})]_{11} & \mathbf{J}^{[1,\mathrm{n}]}(\vb*{\theta}) \\
        \mathbf{J}^{[\mathrm{n},1]}(\vb*{\theta}) & \mathbf{J}^{[\mathrm{n,n}]}(\vb*{\theta})
    \end{pmatrix},
\end{align}
where $[\mathbf{J}(\vb*{\theta})]_{11}$ is $(1,1)$-component of $\mathbf{J}(\vb*{\theta})$, $\mathbf{J}^{[\mathrm{n},1]}(\vb*{\theta}) $ is $(M-1) \times 1$ matrix corresponding to $[\mathbf{J}(\vb*{\theta})]_{i1}$ for $i=2,3,\cdots,M$, $\mathbf{J}^{[1,\mathrm{n}]}(\vb*{\theta}) = \mathbf{J}^{[\mathrm{n},1]}(\vb*{\theta})^{\mathrm{T}}$, and $\mathbf{J}^{[\mathrm{n,n}]}_{\varepsilon}(\vb*{\theta})$ is $(M-1) \times (M-1)$ matrix associated with $ [\mathbf{J}_{\varepsilon}(\vb*{\theta})]_{ij}$ for $i,j=2,3,\cdots,M$.
For this block representation, we have the following property:
\begin{property}\label{S:prop1}
    If $\vb*{w}_1 \in \mathrm{supp}\left(\mathbf{J}(\vb*{\theta}_0)\right)$, then $\mathbf{J}^{[\mathrm{n},1]}(\vb*{\theta}_0) \in \mathrm{supp}\left(\mathbf{J}^{[\mathrm{n,n}]}(\vb*{\theta}_0) \right)$.
\end{property}
We prove Property~\ref{S:prop1} at the end of this subsection.

\textit{Step 2. Block-diagonalization.---}
Next, we find a parameter set $\vb*{\xi}$ that achieves a block-diagonalization of $\mathbf{J}(\vb*{\xi})$, separating its $(1,1)$-component from the rest of the matrix blocks as $[\mathbf{J}(\vb*{\xi})]_{1i}=[\mathbf{J}(\vb*{\xi})]_{i1}=0$ for all $2\leq i \leq M$.
Especially, we focus on block-diagonalizing $\mathbf{J}(\vb*{\xi})$ at a fixed point $\vb*{\theta} = \vb*{\theta}_{0}$.
{Let us define a new parameter set $\vb*{\xi}$ as
\begin{align}
    \xi_{1}=\theta_{1},\quad \xi_{i}=\theta_{i}+\sum_{j=1}^{M-1}[\mathbf{J}^{[\mathrm{n,n}]}(\vb*{\theta}_{0})^{+}]_{i-1,j}[\mathbf{J}^{[\mathrm{n},1]}(\vb*{\theta_{0}})]_{j}(\theta_{1}-\theta_{1,0}) \quad (i= 2,3,\cdots,M). \label{S:newparameterlocal}
\end{align}}

For this parameter set, we have
\begin{align}
    \pdv{\vb*{\theta}^{[n]}}{\xi_{1}} = -\mathbf{J}^{[\mathrm{n,n}]}(\vb*{\theta}_0)^{+}\mathbf{J}^{[\mathrm{n},1]}(\vb*{\theta}_0),
\end{align}
where we have defined $\pdv{\vb*{\theta}^{[n]}}{\xi_{1}} = (\pdv{\theta_2}{\xi_{1}}, \ldots, \pdv{\theta_M}{\xi_{1}})$.
Under the assumption of $\vb*{w}_1 \in \mathrm{supp}\left(\mathbf{J}(\vb*{\theta}_0)\right)$, from Property~\ref{S:prop1}, we obtain
\begin{align}
    \mathbf{J}^{[\mathrm{n},1]}(\vb*{\theta}_0)+\mathbf{J}^{[\mathrm{n,n}]}(\vb*{\theta}_0)\pdv{\vb*{\theta}^{[n]}}{\xi_{1}}=0. \label{S:PLdifferetial}
\end{align}
Therefore, $(i,1)$-th element for $i = 2, \ldots, M$ of the QFIM $\mathbf{J}(\vb*{\xi})$ satisfies
\begin{align}
    [\mathbf{J}(\vb*{\xi})]_{i1}
    =\sum_{\alpha,\beta =1}^{M}\pdv{\theta_{\alpha}}{\xi_{i}}[\mathbf{J}(\vb*{\theta})]_{\alpha\beta}\pdv{\theta_{\beta}}{\xi_{1}}
    =\sum_{\alpha=i}^{M}\left(\pdv{\theta_{\alpha}}{\xi_{i}}\right)\left([\mathbf{J}(\vb*{\theta})]_{\alpha1}+\sum_{\beta=2}^{M}[\mathbf{J}(\vb*{\theta})]_{\alpha\beta}\pdv{\theta_{\beta}}{\xi_{1}}\right)=0 \label{S:Jxi1i}
\end{align}
at $\vb*{\theta} =  \vb*{\theta}_0$.
In the third term, the summation over $\alpha$ starts from $i$ because $\pdv{\theta_{\alpha}}{\xi_{i}}=0$ for $\alpha < i$.
Therefore, the new parameter set defined in Eq.~\eqref{S:newparameterlocal} block-diagonalizes $\mathbf{J}(\vb*{\xi})$ at $\vb*{\theta} =  \vb*{\theta}_0$.
We note that the QFIM $\mathbf{J}(\vb*{\xi})$ for this parameter $\vb*{\xi}$ satisfies the following property:
\begin{property}\label{S:prop2}
    For a set of parameter $\vb*{\xi}$ satisfying Eq.~\eqref{S:newparameterlocal}, $([\mathbf{J}(\vb*{\xi})]_{11})^{-1} = [\mathbf{J}(\vb*{\xi})^{+}]_{11}=[\mathbf{J}(\vb*{\theta})^{+}]_{11} > 0$ at $\vb*{\theta} =  \vb*{\theta}_0$.
\end{property}
We prove Property~\ref{S:prop2} at the end of this subsection.

\textit{Step 3. Estimation error.---}In the following, we set $\vb*{\theta} = \vb*{\theta}_0$. We now analyze the estimation error for $\theta_{1}(=\xi_{1})$. For this purpose, we employ the POVM set $\{\hat{\Pi}^{\mathrm{SLD}}_{\vb*{x}}\}_{\vb*{x}}$, which is the eigenbasis of SLD operator $\hat{L}_{\xi_{1}}$ that satisfies $\pdv{\hat{\rho}_{\vb*{\xi}}}{\xi_{1}}=\frac{1}{2}\{\hat{L}_{\xi_{1}},\hat{\rho}_{\vb*{\xi}}\}$. As an estimator for $\theta_1 = \xi_{1}$, we consider 
\begin{align}
    \tilde{\theta}_{1} = \tilde{\xi}_{1}=\left.\left(\xi_{1}+[\mathbf{J}(\vb*{\xi})^{+}]_{11}\pdv{p_{\vb*{\xi}}}{\xi_{1}}(\vb*{x})\frac{1}{p_{\vb*{\xi}}(\vb*{x})}\right)\right|_{\vb*{\theta}=\vb*{\theta}_0}, \label{S:optimalest}
\end{align}
where $p_{\vb*{\xi}}(\vb*{x})=\mathrm{Tr}[\hat{\Pi}^{\mathrm{SLD}}_{\vb*{x}}\hat{\rho}_{\vb*{\xi}}]$.
It is easy to see that this estimator is locally unbiased at $\vb*{\theta} = \vb*{\theta}_0$.

The corresponding estimation error is then
\begin{align}
    \big<(\tilde{\theta}_{1}-\theta_{1})^{2}\big>=\big<(\tilde{\xi}_{1}-\xi_{1})^{2}\big> = \left([\mathbf{J}(\vb*{\xi})^{+}]_{11}\right)^{2}\int \frac{\left(\pdv{}{\xi_{1}}p_{\vb*{\xi}}(\vb*{x})\right)^{2}}{p_{\vb*{\xi}}(\vb*{x})}d\vb*{x}.
\end{align}
Here, the integral $\int \frac{\left(\pdv{}{\xi_{1}}p_{\vb*{\xi}}(\vb*{x})\right)^{2}}{p_{\vb*{\xi}}(\vb*{x})}d\vb*{x}$ represents the Fisher information for the parameter $\xi_{1}$ corresponding to the POVM $\{\hat{\Pi}^{\mathrm{SLD}}_{x}\}_{x}$. Since we have chosen the POVM set $\{\hat{\Pi}^{\mathrm{SLD}}_{x}\}_{x}$ as the eigenbasis of the SLD operator 
\begin{align}
    \hat{L}_{\xi_{1}}=\sum_{i=1}^{M}\pdv{\theta_{i}}{\xi_{1}}\hat{L}_{\theta_i}, \label{S:optimalSLD}
\end{align}
the Fisher information associated with $\xi_{1}$ is equal to the quantum Fisher information for $\xi_{1}$, which is $[\mathbf{J}(\vb*{\xi})]_{11}$.
Therefore, from Property~\ref{S:prop2}, we obtain
\begin{align}
    \big<(\tilde{\theta}_{1}-\theta_{1})^{2}\big> = [\mathbf{J}(\vb*{\xi})^{+}]_{11} = [\mathbf{J}(\vb*{\theta})^{+}]_{11}.
\end{align}
\end{proof}

\textit{proof of Property \ref{S:prop1}}.---
In this proof, we set $\vb*{\theta} = \vb*{\theta}_0$.
Let us assume $\vb*{w}_1 \in \mathrm{supp}\left(\mathbf{J}\right)$.
Then, there exists a vector $\vb*{x}^T = (x_1, \ldots, x_n)^T\in\mathbb{R}^n$ satisfying
$\mathbf{J}\vb*{x} = \vb*{w}_1$.
When we define $\vb*{x}^{[n,1]T} = (x_2, \ldots, x_n)^T\in\mathbb{R}^{n-1}$, we obtain
\begin{align}
        x_1\mathbf{J}^{[\mathrm{n},1]} + \mathbf{J}^{[\mathrm{n,n}]}\vb*{x}^{[n,1]} = 0.
\end{align}
Since $\mathbf{J}$ is a positive semidefinite matrix and $\vb*{x}$ is not in the kernal of $\mathbf{J}$ from the definition, we have $x_1 = \vb*{x}^T\vb*{w}_1 = \vb*{x}^T\mathbf{J}\vb*{x} > 0$.
Therefore,
\begin{align}
         \mathbf{J}^{[\mathrm{n,n}]}(\vb*{\theta})\left(\frac{-1}{x_1}\vb*{x}^{[n,1]}\right) = \mathbf{J}^{[\mathrm{n},1]}(\vb*{\theta}).
\end{align}
This means that $\mathbf{J}^{[\mathrm{n},1]} \in \mathrm{supp}\left(\mathbf{J}^{[\mathrm{n,n}]} \right)$ at $\vb*{\theta} = \vb*{\theta}_0$.

\qed

{\textit{proof of Property \ref{S:prop2}}.--- In this proof, we set $\theta=\theta_0$.
Write the QFIM in the original parametrization as
\begin{align}
\mathbf{J}(\vb*{\theta})=
\begin{pmatrix}
a & \vb*{b}^\mathrm{T}\\
\vb*{b} & \mathbf{C}
\end{pmatrix},
\end{align}
where $a=[\mathbf{J}(\vb*{\theta})]_{11}$, $\vb*{b}=\mathbf{J}^{[n,1]}(\vb*{\theta})$, and
$\mathbf{C}=\mathbf{J}^{[n,n]}(\vb*{\theta})$. From Property 1, we have
$\vb*{b}\in \operatorname{supp}(\mathbf{C})$. Let
\begin{align}
\vb*{q}:=\mathbf{C}^{+}\vb*{b} .
\end{align}
The reparametrization introduced above corresponds to
\begin{align}
\mathbf{T}=
\begin{pmatrix}
1 & -\vb*{q}^\mathrm{T}\\
0 & \mathbf{I}
\end{pmatrix},
\end{align}
so that
\begin{align}
\mathbf{J}(\vb*{\xi})=\mathbf{T}\mathbf{J}(\vb*{\theta})\mathbf{T}^\mathrm{T}
=
\begin{pmatrix}
s & 0\\
0 & \mathbf{C}
\end{pmatrix},
\qquad
s:=a-\vb*{b}^T\mathbf{C}^{+}\vb*{b} .
\end{align}
Since $\vb*{w}_1\in\operatorname{supp}(\mathbf{J}(\vb*{\theta}))$ and $\mathbf{T}\vb*{w}_1=\vb*{w}_1$,
we have $\vb*{w}_1\in\operatorname{supp}(\mathbf{J}(\vb*{\xi}))$, which implies
$\vb*{s}>0$. Hence
\begin{align}
[\mathbf{J}(\vb*{\xi})^+]_{11}=s^{-1}.
\end{align}

It remains to show that $[\mathbf{J}(\vb*{\theta})^+]_{11}=s^{-1}$.
Since $\vb*{w}_1\in\operatorname{supp}(\mathbf{J}(\vb*{\theta}))$, the equation
$\mathbf{J}(\vb*{\theta})\vb*{y}=\vb*{w}_1$ has a solution. For any solution
$\vb*{y}=(y_1,\vb*{y}_n)^\mathrm{T}$, the block equations are
\begin{align}
ay_1+\vb*{b}^T\vb*{y}_n=1,\qquad \vb*{b}y_1+\mathbf{C}\vb*{y}_n=0 .
\end{align}
Using $\vb*{b}=\mathbf{C}\vb*{q}$, the second equation gives $\mathbf{C}\vb*{y}_n=-\mathbf{C}\vb*{q}y_1$.
Therefore,
\begin{align}
\vb*{b}^\mathrm{T}\vb*{y}_n=\vb*{q}^\mathrm{T}\mathbf{C}\vb*{y}_n=-\vb*{q}^\mathrm{T}\mathbf{C}\vb*{q}y_1=-\vb*{b}^\mathrm{T}\mathbf{C}^+\vb*{b}\,y_1 .
\end{align}
Substituting this into the first equation yields
\begin{align}
(a-\vb*{b}^\mathrm{T}\mathbf{C}^{+}\vb*{b})y_1=sy_1=1 .
\end{align}
Thus every solution of $\mathbf{J}(\vb*{\theta})\vb*{y}=\vb*{w}_1$ has first component
$y_1=s^{-1}$. In particular, the minimum-norm solution
$\mathbf{J}(\vb*{\theta})^{+}\vb*{w}_1$ also has first component $s^{-1}$, and hence
\begin{align}
[\mathbf{J}(\vb*{\theta})^+]_{11}
=\vb*{w}_1^\mathrm{T}\mathbf{J}(\vb*{\theta})^+\vb*{w}_1
=s^{-1}
=[\mathbf{J}(\vb*{\xi})^+]_{11}.
\end{align}
This proves Property 2.}

\qed

\subsection{proof of \textit{only if} direction}
In this subsection, we demonstrate that there exists a locally unbiased estimator of $\theta_1$ at $\vb*{\theta} = \vb*{\theta}_{0}$ with finite estimation error only if $\vb*{w}_1 :=(1,0,0,\cdots,0)^{\mathrm{T}} \in \mathrm{supp}\left(\left.\mathbf{J}\right|_{\vb*{\theta} = \vb*{\theta}_0}\right)$.
\begin{proof}
In this proof, we set $\vb*{\theta} = \vb*{\theta}_0$.
Here we show the proof by contradiction. We begin by assuming that there exists a locally unbiased estimator of $\theta_{1}$ with finite estimation error, but $\vb*{w}_{1} \notin \mathrm{supp}\left(\mathbf{J}\right)$. Since the estimator satisfies Eq.~\eqref{S:locallyunbiased_1}, we have $[\mathbf{H}]_{i1}=\delta_{1i}$ and $[\mathbf{B}]_{11}=0$. Under these conditions, the perturbed quantum Cramér-Rao matrix inequality, as defined in Eq. \eqref{S:epsilonQCRB}, yields the following inequality for the estimation error:
\begin{align}
    \big<(\tilde{\theta}_{1}-\theta_{1})^{2}\big> = \vb*{w}_{1}^{\mathrm{T}}\mathbf{C}\vb*{w}_{1} \geq \vb*{w}_{1}^{\mathrm{T}}\mathbf{J}^{-1}_{\vb*{\varepsilon}}\vb*{w}_{1}=\sum_{i=1}^{s} \frac{(\vb*{w}_{1}^{\mathrm{T}} \vb*{\lambda}_{i})^{2}}{\lambda_{i}+\varepsilon_{\mathrm{s}}}+\sum_{j=s+1}^{M} \frac{(\vb*{w}_{1}^{\mathrm{T}} \vb*{\lambda}_{j})^{2}}{\varepsilon_{\mathrm{n}}},~~\forall ~ \varepsilon_{\mathrm{s}},\varepsilon_{\mathrm{n}}>0. \label{S:onlyifepsi}
\end{align}
Here, we emphasize that since $\vb*{w}_{1}\notin \mathrm{supp}(\mathbf{J})$, there exists at least one component $(\vb*{w}_{1}^{\mathrm{T}} \vb*{\lambda}_{j})$ for $(s+1)\leq j \leq M$ being non-zero. This implies that in the limit of $\varepsilon_{\mathrm{s}},\varepsilon_{\mathrm{n}} \to 0$, the second summation term in Eq. \eqref{S:onlyifepsi} diverges due to the denominator $\varepsilon_{\mathrm{n}}$ approaching zero. Consequently, the lower bound for the estimation error diverges, which contradicts the initial premise that a locally unbiased estimator of $\theta_{1}$ exists with finite estimation error. Thus, the assumption that $\vb*{w}_{1}\notin\mathrm{supp}(\mathbf{J})$ leads to a contradiction. Therefore, we conclude that a locally unbiased estimator of $\theta_{1}$ with finite estimation error can exist only if $\vb*{w}_{1}\in\mathrm{supp}(\mathbf{J})$.    
\end{proof}

\section{Proof of Theorem \ref{S:theorem1}}
In this section, we provide a proof of Theorem~\ref{S:theorem1} in the main text, which is stated as follows:
\setcounter{theorem}{0}
\begin{theorem}\label{S:theorem1}
    There exists a locally unbiased estimator of $\theta_{1}$ at $\vb*{\theta} = \vb*{\theta}_0$ with finite estimation error if and only if $\left.\pdv{\hat{\rho}_{\vb*{\theta}}}{\theta_{1}}\right|_{\vb*{\theta} = \vb*{\theta}_0}$ cannot be expressed with a linear combinations of $\left.\pdv{\hat{\rho}_{\vb*{\theta}}}{\theta_{i}}\right|_{\vb*{\theta} = \vb*{\theta}_0}$ for $i \neq 1$, i.e.,
    \begin{align}
    \left.\pdv{\hat{\rho}_{\vb*{\theta}}}{\theta_{1}}\right|_{\vb*{\theta} = \vb*{\theta}_0} \neq \sum_{i = 2}^{M} c_{i} \left.\pdv{\hat{\rho}_{\vb*{\theta}}}{\theta_{i}}\right|_{\vb*{\theta} = \vb*{\theta}_0}
    \end{align}
    for all $c_i\in\mathbb{C}$.
\end{theorem}

\subsection{proof of \textit{if} direction}
We aim to show that if $\left.\pdv{\hat{\rho}_{\vb*{\theta}}}{\theta_{1}}\right|_{\vb*{\theta} = \vb*{\theta}_0} \neq \sum_{i = 2}^{M} c_{i} \left.\pdv{\hat{\rho}_{\vb*{\theta}}}{\theta_{i}}\right|_{\vb*{\theta} = \vb*{\theta}_0}$, there exists a locally unbiased estimator of $\theta_{1}$ at $\vb*{\theta} = \vb*{\theta}_0$ with finite estimation error.
\begin{proof}
In this proof, we set $\vb*{\theta} = \vb*{\theta}_0$.
According to Lemma \ref{S:lemma1}, the existence of a {locally} unbiased estimator of $\theta_{1}$ with finite estimation error is equivalent to the condition $\vb*{w}_{1} \in \mathrm{supp}(\mathbf{J})$. Therefore, the statement is equivalent to proving if $\pdv{\hat{\rho}_{\vb*{\theta}}}{\theta_{1}} \neq \sum_{i \neq 1} c_{i} \pdv{\hat{\rho}_{\vb*{\theta}}}{\theta_{i}}$, then $\vb*{w}_{1} \in \mathrm{supp}(\mathbf{J})$. For this proof, we use the contrapositive method. Specifically, we demonstrate the contrapositive statement: if $\vb*{w}_{1} \notin \mathrm{supp}(\mathbf{J})$, then $\pdv{\hat{\rho}_{\vb*{\theta}}}{\theta_{1}} = \sum_{i \neq 1} c_{i} \pdv{\hat{\rho}_{\vb*{\theta}}}{\theta_{i}}$.

Let us consider $\pdv{\hat{\rho}_{\vb*{\theta}}}{\theta_{1}}$:
\begin{align}
    \pdv{\hat{\rho}_{\vb*{\theta}}}{\theta_{1}} = \sum_{i=1}^{M}\delta_{1i} \pdv{\hat{\rho}_{\vb*{\theta}}}{\theta_{i}} =  \sum_{i=1}^{M}\big(\vb*{w}_{1}^{\mathrm{T}}\vb*{w}_{i}\big) \pdv{\hat{\rho}_{\vb*{\theta}}}{\theta_{i}} = \sum_{i,k=1}^{M}  \big(\vb*{w}_{1}^{\mathrm{T}}\vb*{\lambda}_{k}\big) \big(\vb*{\lambda}_{k}^{\mathrm{T}}\vb*{w}_{i}\big)\pdv{\hat{\rho}_{\vb*{\theta}}}{\theta_{i}} = \sum_{k=1}^{M}  \big(\vb*{w}_{1}^{\mathrm{T}}\vb*{\lambda}_{k}\big) \hat{\rho}_{\vb*{\lambda}_{k}},
\end{align}
where we define 
\begin{align}
    \hat{\rho}_{\vb*{\lambda}_{k}} :=\sum_{i=1}^{M}\big(\vb*{\lambda}_{k}^{\mathrm{T}}\vb*{w}_{i} \big)\partial_{i}\hat{\rho}_{\vb*{\theta}}= \sum_{i=1}^{M}\big(\vb*{w}_{i}^{\mathrm{T}} \vb*{\lambda}_{k}\big)\partial_{i}\hat{\rho}_{\vb*{\theta}}.
\end{align}
Here, $\vb*{w}_{i}$ is $M \times 1$ vector with elements $[\vb*{w}_{i}]_{k}=\delta_{ik}$ for all $1 \leq i \leq M$, and $\{\vb*{\lambda}_{i}\}_{i=1}^{M}$ is the eigenvector set of the QFIM $\mathbf{J}$ defined in Eq. \eqref{S:diagJ}. Each of the two vector sets forms an orthonormal basis for an $M$-dimensional real vector. 
Next, by leveraging Property \ref{S:prop4}, we complete our proof. The proof of Property \ref{S:prop4} will be provided at the end of this subsection.
\begin{property}\label{S:prop4}
    $\hat{\rho}_{\vb*{\lambda}_{k}}=0$ for $k \geq (s+1)$.
\end{property}
Using \textit{Property} \ref{S:prop4}, we further obtain
\begin{align}
    \pdv{\hat{\rho}_{\vb*{\theta}}}{\theta_{1}}=\sum_{k=1}^{s}  \big(\vb*{w}_{1}^{\mathrm{T}}\vb*{\lambda}_{k}\big) \hat{\rho}_{\vb*{\lambda}_{k}} = \sum_{i=1}^{M}\sum_{k=1}^{s} \big(\vb*{w}_{1}^{\mathrm{T}} \vb*{\lambda}_{k})\big(\vb*{\lambda}_{k}^{\mathrm{T}} \vb*{w}_{i})\pdv{\hat{\rho}_{\vb*{\theta}}}{\theta_{i}}= \sum_{k=1}^{s} \big(\vb*{w}_{1}^{\mathrm{T}} \vb*{\lambda}_{k})\big(\vb*{\lambda}_{k}^{\mathrm{T}} \vb*{w}_{1})\pdv{\hat{\rho}_{\vb*{\theta}}}{\theta_{1}} + \sum_{i=2}^{M}\sum_{k=1}^{s} \big(\vb*{w}_{1}^{\mathrm{T}} \vb*{\lambda}_{k})\big(\vb*{\lambda}_{k}^{\mathrm{T}} \vb*{w}_{i})\pdv{\hat{\rho}_{\vb*{\theta}}}{\theta_{i}}.
\end{align}
Since we are considering the case where $\vb*{w}_{1} \notin \mathrm{supp}(\mathbf{J})$, it follows that $\sum_{k=1}^{s} \big(\vb*{w}_{1}^{\mathrm{T}} \vb*{\lambda}_{k})\big(\vb*{\lambda}_{k}^{\mathrm{T}} \vb*{w}_{1}) <1$. As a consequence, we can express $\pdv{\hat{\rho}_{\vb*{\theta}}}{\theta_{1}}$ as
\begin{align}
    \pdv{\hat{\rho}_{\vb*{\theta}}}{\theta_{1}} = \sum_{i=2}^{M}\left( \frac{\sum_{k=1}^{s} \big(\vb*{w}_{1}^{\mathrm{T}} \vb*{\lambda}_{k})\big(\vb*{\lambda}_{k}^{\mathrm{T}} \vb*{w}_{i})}{1-\sum_{k=1}^{s} \big(\vb*{w}_{1}^{\mathrm{T}} \vb*{\lambda}_{k})\big(\vb*{\lambda}_{k}^{\mathrm{T}} \vb*{w}_{1})}\right)\pdv{\hat{\rho}_{\vb*{\theta}}}{\theta_{i}}.
\end{align}
Therefore, if $\vb*{w}_{1} \notin \mathrm{supp}(\mathbf{J})$, there exists a set of coefficient $\{c_{i}\}_{i=1}^{M}$ that satisfies $\pdv{\hat{\rho}_{\vb*{\theta}}}{\theta_{1}} = \sum_{i = 2}^{M} c_{i} \pdv{\hat{\rho}_{\vb*{\theta}}}{\theta_{i}}$.
\end{proof}

\textit{proof of Property \ref{S:prop4}}.---{Let us express the components of the QFIM using the eigenvectors of the finite-dimensional encoded state $\hat{\rho}_{\vb*{\theta}}$.}
The diagonal form of the $d$-dimensional encoded state $\hat{\rho}_{\vb*{\theta}}$ is expressed as
\begin{align}
    \hat{\rho}_{\vb*{\theta}}=\sum_{\alpha=1}^{d}\rho_{\alpha}\dyad{\psi_{\alpha}}=\sum_{\alpha=1}^{S}\rho_{\alpha}\dyad{\psi_{\alpha}}+\sum_{\beta=S+1}^{d}0\dyad{\psi_{\beta}},
\end{align}
where $\rho_{\alpha}>0$ for all $ 1\leq \alpha \leq S$. If $\hat{\rho}_{\vb*{\theta}}$ is not full-rank in the Hilbert space, then $S < d$. In addition, we emphasize that $\{\ket{\psi_{\alpha}}\}_{\alpha=1}^{d}$ forms an orthonormal basis for the $d$-dimensional Hilbert space. By using $\{\ket{\psi_{\alpha}}\}_{\alpha=1}^{d}$ , the elements of the QFIM can be expressed as
\begin{align}
    \vb*{w}_{i}^{\mathrm{T}}\mathbf{J}\vb*{w}_{j} = [\mathbf{J}]_{ij}=\sum_{\rho_{\alpha}+\rho_{\beta} \neq 0}\frac{2\Re(\langle \psi_{\alpha} \vert \partial_{i}\hat{\rho}_{\vb*{\theta}} \vert \psi_{\beta} \rangle \langle \psi_{\beta} \vert \partial_{j}\hat{\rho}_{\vb*{\theta}} \vert \psi_{\alpha} \rangle)}{\rho_{\alpha}+\rho_{\beta}},
\end{align}
where the summation is over all $1 \leq \alpha,\beta \leq d$ such that $\rho_{\alpha}+\rho_{\beta} \neq 0$.

Next, let us consider $\vb*{\lambda}_{k}^{\mathrm{T}}\mathbf{J} \vb*{\lambda}_{k}$. For $k\geq (s+1)$, $\vb*{\lambda}_{k}$ is the eigenvector of $\mathbf{J}$ with eigenvalue $0$, therefore, $\vb*{\lambda}_{k}^{\mathrm{T}}\mathbf{J} \vb*{\lambda}_{k}=0$. This can be expressed in terms of the basis set $\{\vb*{w}_{i}\}_{i=1}^{M}$ as follows:
\begin{align}
    0=\vb*{\lambda}_{k}^{\mathrm{T}}\mathbf{J} \vb*{\lambda}_{k} = \vb*{\lambda}_{k}^{\mathrm{T}}\left(\sum_{i=1}^{M}\vb*{w}_{i}\vb*{w}_{i}^{\mathrm{T}}\right) \mathbf{J} \left(\sum_{j=1}^{M}\vb*{w}_{j}\vb*{w}_{j}^{\mathrm{T}}\right)\vb*{\lambda}_{k}=\sum_{i,j=1}^{M} (\vb*{\lambda}_{k}^{\mathrm{T}} \vb*{w}_{i})[\mathbf{J}]_{ij} (\vb*{w}_{j}^{\mathrm{T}}\vb*{\lambda}_{k} ) = 2\sum_{\rho_{\alpha}+\rho_{\beta} \neq 0}\frac{\abs{\langle \psi_{\alpha} \vert \hat{\rho}_{\vb*{\lambda}_{k}} \vert \psi_{\beta} \rangle}^{2}}{\rho_{\alpha}+\rho_{\beta}}, \label{S:QFIlambda}
\end{align}
Here, the numerator is equal or greater than $0$ and the denominator is positive. Therefore, for $\vb*{\lambda}_{k}^{\mathrm{T}}\mathbf{J} \vb*{\lambda}_{k}=0$ to hold for $(s+1) \leq  k \leq M$, it must be true that $\langle \psi_{\alpha} \vert \hat{\rho}_{\vb*{\lambda}_{k}} \vert \psi_{\beta} \rangle=0$ for all $1 \leq \alpha,\beta \leq d$ such that $\rho_{\alpha}+\rho_{\beta} \neq 0$.
{Next, let us show 
\(\bra{\psi_{\alpha}}\hat{\rho}_{\vb*{\lambda}_{k}}\ket{\psi_{\beta}}=0\)
for \(\alpha\) and \(\beta\) satisfying
\(\rho_{\alpha}+\rho_{\beta}=0\). We note that
\(\rho_{\alpha}+\rho_{\beta}=0\) if and only if
\(\rho_{\alpha}=0\) and \(\rho_{\beta}=0\).
From the definition of \(\hat{\rho}_{\vb*{\lambda}_{k}}\), we have
\begin{align}
    \hat{\rho}_{\vb*{\theta}+\varepsilon\vb*{\lambda}_k}
    =
    \hat{\rho}_{\vb*{\theta}}
    +
    \varepsilon \hat{\rho}_{\vb*{\lambda}_k}
    +
    O(\varepsilon^2)
\end{align}
for sufficiently small \(\varepsilon\).
When we sandwich this state with an arbitrary normalized vector
\(\ket{\psi^0}\in\ker(\hat{\rho}_{\vb*{\theta}})\), we obtain
\begin{align}
    0
    \leq
    \ev*{\hat{\rho}_{\vb*{\theta}+\varepsilon\vb*{\lambda}_k}}{\psi^0}
    =
    \varepsilon
    \ev*{\hat{\rho}_{\vb*{\lambda}_k}}{\psi^0}
    +
    O(\varepsilon^2).
\end{align}
Dividing by \(\varepsilon>0\) and taking
\(\varepsilon\to0^+\), we obtain
\begin{align}
\ev*{\hat{\rho}_{\vb*{\lambda}_k}}{\psi^0}\ge 0 .
\end{align}
Let \(\hat{\Pi}_0\) denote the projector onto
\(\ker(\hat{\rho}_{\vb*{\theta}})\), and define
\begin{align}
\hat{A}_0
:=
\hat{\Pi}_0
\hat{\rho}_{\vb*{\lambda}_k}
\hat{\Pi}_0 .
\end{align}
Since the above inequality holds for every normalized
\(\ket{\psi^0}\in\ker(\hat{\rho}_{\vb*{\theta}})\),
the null-space block \(\hat{A}_0\) is positive semidefinite.

Moreover, since \(\hat{\rho}_{\vb*{\lambda}_k}\) is a linear
combination of derivatives of a normalized density operator,
we have
\begin{align}
\mathrm{Tr}\!\left[\hat{\rho}_{\vb*{\lambda}_k}\right]=0 .
\end{align}
From the previous argument,
\begin{align}
\bra{\psi_\alpha}
\hat{\rho}_{\vb*{\lambda}_k}
\ket{\psi_\beta}
=0
\end{align}
for all \(\alpha,\beta\) satisfying
\(\rho_\alpha+\rho_\beta\neq0\). In particular, all diagonal
elements on \(\operatorname{supp}(\hat{\rho}_{\vb*{\theta}})\)
vanish. Hence,
\begin{align}
0
=
\mathrm{Tr}\!\left[\hat{\rho}_{\vb*{\lambda}_k}\right]
=
\mathrm{Tr}\!\left[
\hat{\Pi}_0
\hat{\rho}_{\vb*{\lambda}_k}
\hat{\Pi}_0
\right]
=
\mathrm{Tr}\!\left[\hat{A}_0\right].
\end{align}
Since \(\hat{A}_0\) is positive semidefinite and has zero trace,
it must be the zero matrix. Therefore,
\begin{align}
\bra{\psi_\alpha}
\hat{\rho}_{\vb*{\lambda}_k}
\ket{\psi_\beta}
=0
\end{align}
also for \(\rho_\alpha=0\) and \(\rho_\beta=0\).
Combining this with the case \(\rho_\alpha+\rho_\beta\neq0\),
we conclude that
\begin{align}
\hat{\rho}_{\vb*{\lambda}_k}=0 .
\end{align}
}

\qed

\subsection{proof of \textit{only if} direction}
We aim to show that when there exists a locally unbiased estimator of $\theta_{1}$ at $\vb*{\theta} = \vb*{\theta}_0$ with finite estimation error, then $\left.\pdv{\hat{\rho}_{\vb*{\theta}}}{\theta_{1}}\right|_{\vb*{\theta} = \vb*{\theta}_0} \neq \sum_{i = 2}^{M} c_{i} \left.\pdv{\hat{\rho}_{\vb*{\theta}}}{\theta_{i}}\right|_{\vb*{\theta} = \vb*{\theta}_0}$.
\begin{proof}
In this proof, we set $\vb*{\theta} = \vb*{\theta}_0$ and show the proof by contradiction.
Let us assume $\mathbf{J}^{+}\mathbf{J}\vb*{w}_{1} = \vb*{w}_{1}$ and there exists $c_i\in\mathbb{R}$ such that $\pdv{\hat{\rho}_{\vb*{\theta}}}{\theta_{1}} = \sum_{i \neq 1} c_{i} \pdv{\hat{\rho}_{\vb*{\theta}}}{\theta_{i}}$.
From the assumption of $\pdv{\hat{\rho}_{\vb*{\theta}}}{\theta_{1}} = \sum_{i \neq 1} c_{i} \pdv{\hat{\rho}_{\vb*{\theta}}}{\theta_{i}}$ the elements of the QFIM satisfy:
\begin{align}
    [\mathbf{J}]_{k1}=\sum_{i \neq 1} c_{i} [\mathbf{J}]_{ki}, \forall k. \label{S:detectioncondition2}
\end{align}
This is because $\mathbf{J}_{ki} \equiv \mathrm{Tr} \left[\hat{L}_{k} \pdv{\hat{\rho}_{\vb*{\theta}}}{\theta_{i}} \right]$. Next, let us inspect $\mathbf{J}^{+}\mathbf{J}\vb*{w}_{1}$. Using Eq. \eqref{S:detectioncondition2}, we can express it as
\begin{align}
    & \left[\mathbf{J}^{+}\mathbf{J}\vb*{w}_{1}\right]_{a}  =\sum_{b,c=1}^{M} [\mathbf{J}^{+}]_{ab}[\mathbf{J}]_{bc}[\vb*{w}_{1}]_{c}=\sum_{b=1}^{M}[\mathbf{J}^{+}]_{ab}[\mathbf{J}]_{b1}=\sum_{b=1}^{M}[\mathbf{J}^{+}]_{ab} \sum_{j \neq 1} c_{j} [\mathbf{J}]_{bj} =\left[\mathbf{J}^{+}\mathbf{J}\vb*{c}\right]_{a},
\end{align}
where $\vb*{c}=(0,c_{2},\cdots,c_{M})^{\mathrm{T}}$. Notably, $\vb*{c}$ is orthogonal to $\vb*{w}_{1}$, as $\vb*{w}_{1}$ has non-zero support only in the first component. Now, consider the following inner product:
\begin{align}
    \vb*{w}_{1}^{\mathrm{T}}\left(\mathbf{J}^{+}\mathbf{J}\vb*{w}_{1}\right)=\vb*{w}_{1}^{\mathrm{T}}\left(\mathbf{J}^{+}\mathbf{J}\vb*{c}\right)=\left(\vb*{w}_{1}^{\mathrm{T}}\mathbf{J}^{+}\mathbf{J}\right)\vb*{c}=\vb*{w}_{1}^{\mathrm{T}} \vb*{c}=0. \label{S:pcontra2}
\end{align}
However, this contradicts with $\vb*{w}_{1}^{\mathrm{T}}\left(\mathbf{J}^{+}\mathbf{J}\vb*{w}_{1}\right) = \vb*{w}_1^T\vb*{w}_1 = 1$.
Therefore, when $\mathbf{J}^{+}\mathbf{J}\vb*{w}_{1} = \vb*{w}_{1}$, we have $\pdv{\hat{\rho}_{\vb*{\theta}}}{\theta_{1}} \neq \sum_{i \neq 1} c_{i} \pdv{\hat{\rho}_{\vb*{\theta}}}{\theta_{i}}$ for all $c_i\in\mathbb{R}$.
\end{proof}

\section{Analysis of Quantum Fisher Information Matrix in Noise-Agnostic Sensing}
As discussed in the main text, {locally} unbiased estimation of an unknown signal $\phi$ can be achieved using an entangled quantum probe with a noiseless ancilla, even in the presence of unknown noise. In this section, we analyze the quantum Fisher information matrix (QFIM) of the entangled noisy output state
\begin{equation}
    \label{S:eq_noisy_entangled}
    \hat{\rho}^{\mathrm{ent}} = \sum_{\vb*{x}\in\{0,1\}^n}p_{\vb*{x}}\hat{\Pi}_{\vb*{x}}\frac{1}{2} \left(\hat{I} + \lambda_{\vb*{x}}\cos(n\phi) \hat{X}_\mathrm{L} +  \lambda_{\vb*{x}}\sin(n\phi) \hat{Y}_\mathrm{L} \right),
\end{equation}
where the projector onto the $2^n$ orthogonal code space is defined as $\hat{\Pi}_{\vb*{x}} :=\prod_{i=1}^{n} \frac{1}{2} (\hat{I} +(-1)^{x_i+x_{i+1}} \hat{Z}_i\hat{Z}_{i+1})$ with $x_{n+1} = 0$. Additionally, $p_{\vb*{x}} := \sum_{\vb*{z}\in\{0,1\}^n} p_{(\vb*{x},\vb*{z})}$ and $\lambda_{\vb*{x}} = \sum_{\vb*{z}\in\{0,1\}^n} (-1)^{\vb*{z}\cdot\vb*{1}} p_{(\vb*{x},\vb*{z})}/p_{\vb*{x}}$ with $p_{(\vb*{0},\vb*{0})} = 1-\sum_{\vb*{a} \neq (\vb*{0},\vb*{0})}p_{\vb*{a}}$ and $\vb*{1}$ is $n$-bit binary vector whose components are all $1$.

The original unknown Pauli noise $\mathcal{N}$ is parameterized by Pauli error rates $\{p_{\vb*{a}}\}_{\vb*{a} \neq (\vb*{0},\vb*{0})}$ or Pauli eigenvalues $\{\lambda_{\vb*{a}}\}_{\vb*{a} \neq (\vb*{0},\vb*{0})}$ with $\vb*{a} = (\vb*{x}, \vb*{z})$. However, for the state $\hat{\rho}^{\mathrm{ent}}$, we do not use these noise parameters. Instead, we regard that $\hat{\rho}^{\mathrm{ent}}$ is parameterized by $2\cdot 2^n+1$ parameters: $\phi$, $\{p_{\vb*{x}}\}_{\vb*{x}\in\{0,1\}^n}$, and $\{\lambda_{\vb*{x}}\}_{\vb*{x}\in\{0,1\}^n}$.

The derivative of the quantum state $\hat{\rho}^{\mathrm{ent}}$ for each parameter can be expressed as
\begin{align}
    \partial_{\phi}\hat{\rho}^{\mathrm{ent}} &= \sum_{\vb*{x}\in\{0,1\}^n}p_{\vb*{x}}\hat{\Pi}_{\vb*{x}}\frac{1}{2} \left(- n\lambda_{\vb*{x}}\sin(n\phi) \hat{X}_\mathrm{L} + n\lambda_{\vb*{x}}\cos(n\phi) \hat{Y}_\mathrm{L} \right),\\
    \partial_{p_{\vb*{x}}}\hat{\rho}^{\mathrm{ent}} &= \hat{\Pi}_{\vb*{x}}\frac{1}{2} \left(\hat{I} + \lambda_{\vb*{x}}\cos(n\phi) \hat{X}_\mathrm{L} +  \lambda_{\vb*{x}}\sin(n\phi) \hat{Y}_\mathrm{L} \right),\\
    \partial_{\lambda_{\vb*{x}}}\hat{\rho}^{\mathrm{ent}} &= p_{\vb*{x}}\hat{\Pi}_{\vb*{x}}\frac{1}{2} \left(\cos(n\phi) \hat{X}_\mathrm{L} + \sin(n\phi) \hat{Y}_\mathrm{L} \right).
\end{align}
Moreover, the symmetric logarithmic derivative (SLD) operator for each parameter can be calculated as
\begin{align}
    \hat{L}_\phi &= n\sum_{\vb*{x}\in\{0,1\}^n} \lambda_{\vb*{x}}\hat{\Pi}_{\vb*{x}}(-\sin(n\phi)\hat{X}_L + \cos(n\phi)\hat{Y}_L), \\
    \hat{L}_{p_{\vb*{x}}} &= \frac{1}{p_{\vb*{x}}} \hat{\Pi}_{\vb*{x}},\\
    \hat{L}_{\lambda_{\vb*{x}}} &= \hat{\Pi}_{\vb*{x}} \left(\frac{-\lambda_{\vb*{x}}}{1-\lambda_{\vb*{x}}^2}\hat{I}  + \frac{1}{1-\lambda_{\vb*{x}}^2} \cos(n\phi) \hat{X}_\mathrm{L} + \frac{1}{1-\lambda_{\vb*{x}}^2}\sin(n\phi) \hat{Y}_\mathrm{L} \right).
\end{align}
Therefore, the QFIM of $\hat{\rho}_{\mathrm{ent}}$ is a diagonal matrix with its diagonal elements represented as
\begin{align}
    \mathbf{J}_{\phi,\phi}  &= \mathrm{Tr}[\hat{L}_\phi\partial_{\phi}\hat{\rho}^{\mathrm{ent}}] =  n^2\sum_{\vb*{x}\in\{0,1\}^n}p_{\vb*{x}}\lambda_{\vb*{x}}^2,\\
    \mathbf{J}_{p_{\vb*{x}},p_{\vb*{x}}}  &= \mathrm{Tr}[\hat{L}_{p_{\vb*{x}}}\partial_{p_{\vb*{x}}}\hat{\rho}^{\mathrm{ent}}] = \frac{1}{p_{\vb*{x}}}\\
    \mathbf{J}_{\lambda_{\vb*{x}},\lambda_{\vb*{x}}}  &= \mathrm{Tr}[\hat{L}_{\lambda_{\vb*{x}}}\partial_{\lambda_{\vb*{x}}}\hat{\rho}^{\mathrm{ent}}] = \frac{p_{\vb*{x}}}{1-\lambda_{\vb*{x}}^2}.
\end{align}

Since the QFIM is diagonal and invertible, the optimal estimation protocol achieving the quantum Cram\'er-Rao bound (QCRB) is to measure in a basis diagonalizing $L_\phi$~\cite{est-Suzuki_2020}. Thus, a {locally} unbiased estimator $\phi$ saturating the QCRB can be obtained:
\begin{align}
    \mathrm{Var}[\tilde{\phi}] \geq \mathbf{J}_{\phi,\phi}^{-1} = \frac{1}{n^2}\frac{1}{\sum_{\vb*{x}\in\{0,1\}^n}p_{\vb*{x}}\lambda_{\vb*{x}}^2},
\end{align}
by first measuring $\hat{\rho}_{\mathrm{ent}}$ to determine which code space $\hat{\Pi}_{\vb*{x}}$ the state belongs to, and then measuring the logical observable $(-\sin(n\phi)\hat{X}_L + \cos(n\phi)\hat{Y}_L)$. Therefore, noise-agnostic sensing can be achieved, meaning that an optimal {locally} unbiased estimator of the signal can be obtained without requiring any prior knowledge about the unknown Pauli noise $\mathcal{N}$.

\section{Symmetric Clifford Twirling in Quantum Metrology}
In the analysis of noise-agnostic sensing, we have assumed that the noise channel $\mathcal{N}$ affecting the unitary operation $\mathcal{U}_\phi = \hat{U}_\phi \cdot \hat{U}_\phi^\dag = e^{-i\phi/2 \sum_{i=1}^{n} \hat{Z}_i} \cdot e^{i\phi/2 \sum_{i=1}^{n} \hat{Z}_i}$ is Pauli noise.
However, in a general situation, the noise $\mathcal{N}$ may not be Pauli noise and could be a more complex noise channel.
In such cases, it may be thought that the noise channel can be transformed into Pauli noise through Pauli twirling~\cite{wallman2016noise}.
Furthermore, one might consider transforming the noise into simpler depolarizing noise by applying Clifford twirling~\cite{emerson2005scalable,dankert2009exact}.
To perform Clifford twirling, it is necessary to insert a Clifford operator $\hat{C}$ between $\mathcal{U}_\phi$ and $\mathcal{N}$.
However, this is generally not possible, since the unitary channel $\mathcal{U}_\phi$ and the noise channel $\mathcal{N}$ cannot be separated.
An alternative approach is to insert a modified operator $\hat{U}_\phi^\dag \hat{C}\hat{U}_\phi$ before the unitary channel $\mathcal{U}_\phi$, but this is infeasible as $\hat{U}_\phi^\dag \hat{C}\hat{U}_\phi$ may depend on the unknown parameter $\phi$.
Therefore, in quantum metrology, both Pauli and Clifford twirling cannot be implemented in a naive way.

To address this challenge, Ref.~\cite{tsubouchi2024symmetric} introduced the concept of \textit{symmetric Clifford twirling}.
Let us define the $Z$-symmetric Clifford group $\mathcal{G}_{n,Z}$ as the set of all $n$-qubit Clifford unitaries that commute with Pauli-$Z$ operators:
\begin{align}
    \mathcal{G}_{n,Z} = \{\hat{C}\in\mathcal{G}_n \;|\; \forall i\in \{0,\ldots,n\}, [\hat{C}, Z_i] = 0  \},
\end{align}
where $\mathcal{G}_n$ represents the $n$-qubit Clifford group.
Since a symmetric Clifford operator $\hat{C} \in \mathcal{G}_{n,Z}$ commutes with the unitary $\hat{U}_\phi = e^{-i\phi/2 \sum_{i=1}^{n} \hat{Z}_i}$, inserting $\hat{C}$ before the unitary $\hat{U}_\phi$ is equivalent to inserting $\hat{C}$ after it. Therefore, by randomly applying $\hat{C}\in \mathcal{G}_{n,Z}$ and its inverse $\hat{C}^\dag$ before and after the noisy unitary $\mathcal{N} \circ \hat{U}_\phi$, we obtain
\begin{align}
    \mathbb{E}_{C\in\mathcal{G}_{n,Z}}[\mathcal{C}^\dag\circ\mathcal{N}\circ \mathcal{U}_\phi \circ \mathcal{C}] = \mathbb{E}_{C\in\mathcal{G}_{n,Z}}[\mathcal{C}^\dag\circ\mathcal{N} \circ \mathcal{C}] \circ \mathcal{U}_\phi,
\end{align}
where $\mathcal{C}(\cdot) = \hat{C}\cdot \hat{C}^\dag$ and $\mathcal{C}^\dag(\cdot) = \hat{C}^\dag\cdot \hat{C}$. Thus, by symmetrically twirling the noisy unitary $\mathcal{N} \circ \hat{U}_\phi$, we can twirl only the noise without affecting $\hat{U}_\phi$.
This method is referred to as symmetric Clifford twirling.

While the performance of symmetric Clifford twirling for Pauli noise was thoroughly analyzed in Ref.~\cite{tsubouchi2024symmetric}, its application to general noise remains an open question.
Here, we analyze how symmetric Clifford twirling using the $Z$-symmetric Clifford group $\mathcal{G}_{n,Z}$ transforms a general noise channel $\mathcal{N}$.
For this analysis, we characterize the $n$-qubit noise channel $\mathcal{N}$ using a $4^n\times 4^n$-dimensional square matrix known as the \textit{Pauli transfer matrix}.
The Pauli transfer matrix $A$ describes how each Pauli operator $\hat{P}_{(\vb*{x}, \vb*{z})}$ is mapped through the noise:
\begin{align}
    \mathcal{N}(\hat{P}_{(\vb*{x}', \vb*{z}')}) = \sum_{\vb*{x}',\vb*{z}'}A_{(\vb*{x},\vb*{z}), (\vb*{x}',\vb*{z}')}\hat{P}_{(\vb*{x}, \vb*{z})}.
\end{align}
The elements of the Pauli transfer matrix $A$ can be represented as
\begin{align}
    A_{(\vb*{x},\vb*{z}), (\vb*{x}',\vb*{z}')} = \frac{1}{2^n}\mathrm{tr}[\hat{P}_{(\vb*{x}, \vb*{z})}\mathcal{N}(\hat{P}_{(\vb*{x}', \vb*{z}')})].
\end{align}
Since $\mathcal{N}$ is a trace-preserving channel, we have
\begin{align}
    A_{(\vb*{0},\vb*{0}), (\vb*{x}',\vb*{z}')} = \delta_{\vb*{x}', \vb*{0}}\delta_{\vb*{z}', \vb*{0}}.
\end{align}
However, the remaining $16^n-4^n$ elements have no such restrictions. Thus, a general $n$-qubit quantum noise channel $\mathcal{N}$ contains $16^n - 4^n$ independent parameters.

Meanwhile, for the twirled noise $\mathcal{N}_{\mathrm{twirl}} := \mathbb{E}_{C\in\mathcal{G}_{n,Z}}[\mathcal{C}^\dag\circ\mathcal{N}\circ\mathcal{C}]$ using the $Z$-symmetric Clifford group $\mathcal{G}_{n,Z}$, we obtain the following theorem.
\begin{theorem}
    \label{thm_sct}
    Let $A_{(\vb*{x},\vb*{z}), (\vb*{x}',\vb*{z}')} = \frac{1}{2^n}\mathrm{tr}[\hat{P}_{(\vb*{x}, \vb*{z})}\mathcal{N}(\hat{P}_{(\vb*{x}', \vb*{z}')})]$ be the elements of the Pauli transfer matrix of an $n$-qubit noise channel $\mathcal{N}$. 
    Then, the Pauli transfer matrix $(A_{\mathrm{twirl}})_{(\vb*{x},\vb*{z}), (\vb*{x}',\vb*{z}')} = \frac{1}{2^n}\mathrm{tr}[\hat{P}_{(\vb*{x}, \vb*{z})}\mathcal{N}_\mathrm{twirl}(\hat{P}_{(\vb*{x}', \vb*{z}')})]$ of the twirled noise $\mathcal{N}_{\mathrm{twirl}} = \mathbb{E}_{C\in\mathcal{G}_{n,Z}}[\mathcal{C}^\dag\circ\mathcal{N}\circ\mathcal{C}]$ satisfies
    \begin{align}
        (A_{\mathrm{twirl}})_{(\vb*{x},\vb*{z}), (\vb*{x}',\vb*{z}')} = 
        \left\{
        \begin{array}{ll}
            A_{(\vb*{0},\vb*{z}), (\vb*{0},\vb*{z}')} & (\vb*{x} = \vb*{x}' = \vb*{0}) \\
            (-1)^{\vb*{z}\cdot(\vb*{x}\odot\Delta\vb*{z})}\mathbb{E}_{\vb*{v}\in\{0,1\}^n}[(-1)^{\vb*{v}\cdot(\vb*{x}\odot\Delta\vb*{z})} A_{(\vb*{x},\vb*{v}), (\vb*{x},\vb*{v}+\Delta\vb*{z})})] & (\vb*{x} = \vb*{x}' \neq \vb*{0}) \\
            0 & (\vb*{x} \neq \vb*{x}') \\
        \end{array}
        \right.
         ,
    \end{align}
    where $\Delta \vb*{z} = \vb*{z} + \vb*{z}'$ is the difference between $\vb*{z}$ and $\vb*{z}'$ (note that we are calculating on modulo 2) and $\vb*{a} \odot \vb*{b} = \sum_i a_ib_i \vb*{e}_i$ is an elementwise product of vectors.
\end{theorem}

\begin{proof}
    We first show $(A_{\mathrm{twirl}})_{(\vb*{x},\vb*{z}), (\vb*{x}',\vb*{z}')} = A_{(\vb*{0},\vb*{z}), (\vb*{0},\vb*{z}')}$ for $\vb*{x} = \vb*{x}' = \vb*{0}$.
    When $\vb*{x} = \vb*{x}' = \vb*{0}$, we have
    \begin{align}
        (A_{\mathrm{twirl}})_{(\vb*{0},\vb*{z}), (\vb*{0},\vb*{z}')} 
        &= \frac{1}{2^n}\mathbb{E}_{\hat{C}\in\mathcal{G}_{n,Z}}\left[\mathrm{tr}[\hat{C}\hat{P}_{(\vb*{0}, \vb*{z})}\hat{C}^\dag\mathcal{N}(\hat{C}\hat{P}_{(\vb*{0}, \vb*{z}')}\hat{C}^\dag)]\right] \\
        &= \frac{1}{2^n}\mathbb{E}_{\hat{C}\in\mathcal{G}_{n,Z}}\left[\mathrm{tr}[\hat{P}_{(\vb*{0}, \vb*{z})} \mathcal{N}(\hat{P}_{(\vb*{0}, \vb*{z}')})]\right] \\
        &= A_{(\vb*{0},\vb*{z}), (\vb*{0},\vb*{z}')},
    \end{align}
    where we use that Z-symmetric Clifford operator $\hat{C}\in\mathcal{G}_{n,Z}$ commutes with Pauli-Z operator $\hat{P}_{(\vb*{0}, \vb*{z})}$ at the second equality.
    
    Next, we prove $(A_{\mathrm{twirl}})_{(\vb*{x},\vb*{z}), (\vb*{x}',\vb*{z}')} = 0$ for $\vb*{x} \neq \vb*{x}'$.
    When the twirled noise channel $\mathcal{N}_{\mathrm{twirl}}$ is further twirled by random Pauli-Z operators $\hat{P}_{(\vb*{0}, \vb*{v})}$, the noise channel stays the same because $\hat{P}_{(\vb*{0}, \vb*{v})} \in \mathcal{G}_{n,Z}$.
    Therefore, we have
    \begin{align}
        \mathcal{N}_{\mathrm{twirl}}(\cdot) = \mathbb{E}_{\vb*{v}\in\{0,1\}^n}[\hat{P}_{(\vb*{0}, \vb*{v})}\mathcal{N}_{\mathrm{twirl}}(\hat{P}_{(\vb*{0}, \vb*{v})}\cdot \hat{P}_{(\vb*{0}, \vb*{v})})\hat{P}_{(\vb*{0}, \vb*{v})}].
    \end{align}
    By calculating the Pauli transfer matrix of the LHS and RHS of this equation, we obtain
    \begin{align}
        (A_{\mathrm{twirl}})_{(\vb*{x},\vb*{z}), (\vb*{x}',\vb*{z}')} 
        &= \frac{1}{2^n}\mathbb{E}_{\vb*{v}\in\{0,1\}^n}\left[\mathrm{tr}[\hat{P}_{(\vb*{0}, \vb*{v})}\hat{P}_{(\vb*{x}, \vb*{z})}\hat{P}_{(\vb*{0}, \vb*{v})} \mathcal{N}_{\mathrm{twirl}}(\hat{P}_{(\vb*{0}, \vb*{v})}\hat{P}_{(\vb*{x}', \vb*{z}')}\hat{P}_{(\vb*{0}, \vb*{v})})]\right] \\
        &= \mathbb{E}_{\vb*{v}\in\{0,1\}^n}\left[(-1)^{(\vb*{x}+\vb*{x}')\cdot \vb*{v}}\right] \frac{1}{2^n}\mathrm{tr}[\hat{P}_{(\vb*{x}, \vb*{z})} \mathcal{N}_{\mathrm{twirl}}(\hat{P}_{(\vb*{x}', \vb*{z}')})] \\
        &= \delta_{\vb*{x}, \vb*{x}'} (A_{\mathrm{twirl}})_{\vb*{x},\vb*{z}, \vb*{x},\vb*{z}'}.
        \label{eq_SCTproof_ztwirl}
    \end{align}
    Thus, we have
    \begin{align}
        (A_{\mathrm{twirl}})_{(\vb*{x},\vb*{z}), (\vb*{x}',\vb*{z}')} = 0
    \end{align}
    for $\vb*{x} \neq \vb*{x}'$.

    Finally, we consider the case where $\vb*{x} = \vb*{x}' \neq \vb*{0}$.
    To analyze this case, we remind that any Z-symmetric Clifford operator $\hat{C}\in\mathcal{G}_{n,Z}$ can be uniquely represented as
    \begin{equation}
        \hat{C} = \prod_{\substack{i,j\in\qty{1,\ldots,n}\\i<j}}\hat{CZ}_{ij}^{\nu_{ij}} \prod_{i\in\qty{1,\ldots,n}}\hat{S}_{i}^{\dag\mu_{i}}\prod_{i\in\qty{1,\ldots,n}}\hat{Z}_{i}^{\xi_{i}}
    \end{equation}
    up to phase, where $\hat{CZ}_{ij}$ is the CZ gate acting on the $i$-th and $j$-th qubits, $\hat{S}_i$ and $\hat{Z}_i$ is the S gate and Z gate acting on $i$-th qubit, and $\nu_{ij}, \mu_{i}, \xi_{i}\in\{0,1\}$~\cite{mitsuhashi2023clifford}.
    When  $\vb*{x} = \vb*{x}'$, twirling through Pauli-Z operator does not affect the Pauli transfer matrix (see Eq.~\eqref{eq_SCTproof_ztwirl}), so we only consider the twirling with $\hat{S}_i^\dag$ and $\hat{CZ}_{ij}$.
    Conjugating Pauli operator $\hat{P}_{(\vb*{x}, \vb*{z})}$ with $\hat{S}_{i}^\dag$ flips $z_i$ if $x_i = 1$ with additional phase $(-1)^{z_i}$.
    Conjugating Pauli operator $\hat{P}_{(\vb*{x}, \vb*{z})}$ with $\hat{CZ}_{ij}$ flips $z_i$ if $x_j = 1$ and $z_j$ if $x_i = 1$, with no additional phase.
    Therefore, conjugating Pauli operator $\hat{P}_{(\vb*{x}, \vb*{z})}$ with $\prod\hat{CZ}_{ij}^{\nu_{ij}} \prod\hat{S}_{i}^{\dag\mu_{i}}$ results in
    \begin{align}
        &(-1)^{\sum_{i,x_i=1} \mu_iz_i} \hat{P}_{(\vb*{x},\vb*{z} + \sum_{i, x_i = 1} \mu_i\vb*{e}_i + \sum_{i<j, x_i = 1} \nu_{ij}\vb*{e}_j + \sum_{i<j, x_j = 1} \nu_{ij}\vb*{e}_i)}\\
        = &(-1)^{\sum_{i,x_i=1} \mu_iz_i} \hat{P}_{(\vb*{x},\vb*{z} + \sum_{i, x_i = 1} \mu_i\vb*{e}_i + \sum_{i<j, x_i = x_j = 1} \nu_{ij}(\vb*{e}_i+\vb*{e}_j) + \sum_{i\neq j, x_i=0,x_j = 1} \nu_{ij}\vb*{e}_i)},
    \end{align}
    where we define $\nu_{ij} := \nu_{ji}$ for $i>j$.
    When we randomly sample $\nu_{ij}, \mu_{ij} \in\{0,1\}$, the last term $\sum\nu_{ij}\vb*{e}_i$ randomly flips the elements of $\vb*{z}$ that satisfies $x_i = 0$, and the middle two terms $\sum \mu_i\vb*{e}_i + \sum \nu_{ij}(\vb*{e}_i+\vb*{e}_j)$ randomly flips the elements of $\vb*{z}$ that satisfies $x_i = 1$ with additional phase of $(-1)^{\sum_{i:\mathrm{flipped}}z_i}$.
    In other words, all the elements of $\vb*{z}$ is randomly flipped with additional phase of $(-1)^{\sum_{i:\mathrm{flipped}}x_iz_i}$.
    This means that, the random conjugation with $\prod\hat{CZ}_{ij}^{\nu_{ij}} \prod\hat{S}_{i}^{\dag\mu_{i}}$ results in a uniform mixture of
    \begin{align}
        (-1)^{\vb*{v}\cdot(\vb*{x}\odot\vb*{z})} \hat{P}_{(\vb*{x}, \vb*{z}+\vb*{v})}.
    \end{align}
    Therefore, the Pauli transfer matrix of the twirled noise $\mathcal{N}_{\mathrm{twirl}}$ is
    \begin{align}
        (A_{\mathrm{twirl}})_{(\vb*{x},\vb*{z}), (\vb*{x},\vb*{z}')} 
        &= \frac{1}{2^n}\mathbb{E}_{\vb*{v}\in\{0,1\}^n}\left[\mathrm{tr}[(-1)^{\vb*{v}\cdot(\vb*{x}\odot\vb*{z})} \hat{P}_{(\vb*{x},\vb*{z} + \vb*{v})} \mathcal{N}((-1)^{\vb*{v}\cdot(\vb*{x}\odot\vb*{z}')} \hat{P}_{(\vb*{x}',\vb*{z}' + \vb*{v})})]\right] \\
        &= \frac{1}{2^n}\mathbb{E}_{\vb*{v}\in\{0,1\}^n}\left[(-1)^{\vb*{v}\cdot(\vb*{x}\odot(\vb*{z}+\vb*{z}')) }\mathrm{tr}[\hat{P}_{(\vb*{x},\vb*{z} + \vb*{v})} \mathcal{N}(\hat{P}_{(\vb*{x}',\vb*{z}' + \vb*{v})})]\right] \\
        &= \frac{1}{2^n}(-1)^{\vb*{z}\cdot(\vb*{x}\odot(\vb*{z}+\vb*{z}'))}\mathbb{E}_{\vb*{v}\in\{0,1\}^n}\left[(-1)^{\vb*{v}\cdot(\vb*{x}\odot(\vb*{z}+\vb*{z}')) }\mathrm{tr}[\hat{P}_{(\vb*{x},\vb*{v})} \mathcal{N}(\hat{P}_{(\vb*{x}',\vb*{z} + \vb*{z}' + \vb*{v})})]\right] \\
        &= (-1)^{\vb*{z}\cdot(\vb*{x}\odot(\vb*{z}+\vb*{z}'))}\mathbb{E}_{\vb*{v}\in\{0,1\}^n}\left[(-1)^{\vb*{v}\cdot(\vb*{x}\odot(\vb*{z}+\vb*{z}'))} A_{(\vb*{x},\vb*{v}), (\vb*{x},\vb*{z}+\vb*{z}'+\vb*{v})}\right] \\
        &= (-1)^{\vb*{z}\cdot(\vb*{x}\odot\Delta\vb*{z})}\mathbb{E}_{\vb*{v}\in\{0,1\}^n}[(-1)^{\vb*{v}\cdot(\vb*{x}\odot\Delta\vb*{z})} A_{(\vb*{x},\vb*{v}), (\vb*{x},\vb*{v}+\Delta\vb*{z})}].
    \end{align}

\end{proof}

From Theorem~\ref{thm_sct}, we can conclude that symmetric Clifford twirling has two effects: (1) block-diagonalizing the Pauli transfer matrix by eliminating the elements satisfying $\vb*{x} \neq \vb*{x}'$, and (2) mixing the elements within the same block.
For example, the Pauli transfer matrix $A_{\mathrm{twirl}}$ for a single-qubit twirled noise channel is represented as
\begin{align}
    \label{eq_PTM_singlequbit}
    \begin{pmatrix}
        1 & 0 & 0 & 0 \\
        0 & \frac{1}{2}(A_{xx} + A_{yy}) & \frac{1}{2}(A_{xy} - A_{yx}) & 0 \\
        0 & \frac{1}{2}(A_{xy} - A_{yx}) & \frac{1}{2}(A_{xx} + A_{yy}) & 0\\
        A_{zi} & 0 & 0 & A_{zz} \\
    \end{pmatrix}
    =
    \begin{pmatrix}
        1 & 0 & 0 & 0 \\
        0 & \lambda_1\cos\theta  & \lambda_1\sin\theta & 0\\
        0 & -\lambda_1\sin\theta & \lambda_1\cos\theta & 0 \\
        \alpha & 0 & 0 & \lambda_2 \\
    \end{pmatrix},
\end{align}
which can be characterized by four independent parameters: $\lambda_1, \lambda_2, \alpha, \theta$.
Here, we have denoted $i = (0,0), x = (1,0), y = (1,1), z = (0,1)$, with the first, second, third, and fourth rows and columns representing the elements corresponding to $i$, $x$, $y$, and $z$, respectively.
Therefore, the twirled noise $\mathcal{N}_{\mathrm{twirl}}$ can be interpreted as a combination of depolarizing noise, dephasing noise, amplitude damping noise, and Pauli-$Z$ rotation.

\section{Noise-Agnostic Sensing Using Symmetric Clifford Twirling}
Let us now return to the problem of estimating the signal $\phi$ from the unitary channel $\mathcal{U}_\phi = \hat{U}_\phi \cdot \hat{U}_\phi^\dag = e^{-i\phi/2 \sum_{i=1}^{n} \hat{Z}_i} \cdot e^{i\phi/2 \sum_{i=1}^{n} \hat{Z}_i}$ affected by a general noise channel $\mathcal{N}$.
For simplicity, we consider the case of $n=1$, though our discussion can be easily generalized to arbitrary $n$.
Instead of directly sensing from the noisy channel $\mathcal{N} \circ \mathcal{U}_{\phi}$, we apply symmetric Clifford twirling using the $Z$-symmetric Clifford group $\mathcal{G}_{n,Z}$.
As a result, the noise channel $\mathcal{N}$ is transformed into $\mathcal{N}_{\mathrm{twirl}}$, whose Pauli transfer matrix is represented as Eq.~\eqref{eq_PTM_singlequbit}.
Furthermore, the Pauli transfer matrix of the noisy channel $\mathcal{N} \circ \mathcal{U}_{\phi}$ is given by
\begin{align}
    \begin{pmatrix}
        1 & 0 & 0 & 0 \\
        0 & \lambda_1\cos(\phi + \theta)  & \lambda_1\sin(\phi + \theta) & 0\\
        0 & -\lambda_1\sin(\phi + \theta) & \lambda_1\cos(\phi + \theta) & 0 \\
        \alpha & 0 & 0 & \lambda_2 \\
    \end{pmatrix}.
\end{align}

When a Pauli-$Z$ rotation component $\theta$ is present in the noise, it becomes indistinguishable from the signal $\phi$.
In other words, we have $\partial_\theta \mathcal{N} \circ \mathcal{U}_{\phi} = \partial_\phi \mathcal{N} \circ \mathcal{U}_{\phi}$, which prevents {locally} unbiased estimation of $\phi$.
However, if the original noise $\mathcal{N}$ contains no Pauli-$Z$ rotation component, i.e., $\theta = 0$, {locally} unbiased estimation of $\phi$ can be achieved without the knowledge of the noise.

When the input state is the $\ket{+}$ state with density matrix $\rho_0 = \frac{1}{2} (\hat{I} + \hat{X})$, the output state $\hat{\rho}_{\mathrm{twirl}} = \mathcal{N}_{\mathrm{twirl}} \circ \mathcal{U}_{\phi}(\hat{\rho}_0)$ satisfies
\begin{align}
    \hat{\rho}_{\mathrm{twirl}} = \frac{1}{2}(\hat{I} + \lambda_1\cos(\phi)\hat{X} + \lambda_1\sin(\phi)\hat{Y} + \alpha\hat{Z}).
\end{align}
The QFIM of this state with respect to the parameters $\phi$, $\lambda$, and $\alpha$ satisfies
\begin{align}
    \begin{pmatrix}
        \lambda_1^2 & 0 & 0 \\
        0 & \frac{1-\alpha^2}{1-\lambda_1^2-\alpha^2} & \frac{-\alpha\lambda_1}{1-\lambda_1^2-\alpha^2} \\
        0 & \frac{-\alpha\lambda_1}{1-\lambda_1^2-\alpha^2} & \frac{1-\lambda_1^2}{1-\lambda_1^2-\alpha^2}\\
    \end{pmatrix}.
\end{align}
Since the QFIM is block-diagonal and invertible, a {locally} unbiased estimator $\tilde{\phi}$ that saturates the quantum Cram\'er-Rao bound (QCRB) can be obtained:
\begin{align}
    \mathrm{Var}[\tilde{\phi}] \geq \mathbf{J}_{\phi,\phi}^{-1} =\frac{1}{\lambda_1^2},
\end{align}
by measuring $\hat{\rho}_{\mathrm{twirl}}$ with a basis diagonalizing the SLD operator $\hat{L}_\phi = \lambda_1(-\sin(\phi)\hat{X} + \cos(\phi)\hat{Y})\propto-\sin(\phi)\hat{X} + \cos(\phi)\hat{Y}$~\cite{est-Suzuki_2020}.
Therefore, by applying twirling, we can achieve a {locally} unbiased estimation of the signal $\phi$ with finite variance.

Compared to {locally} unbiased estimation using entanglement, symmetric Clifford twirling offers two advantages.
The first advantage is that the noise does not need to be Pauli noise; the only restriction is the absence of an over-rotation term.
The second advantage is that it does not require the preparation of a noiseless ancilla, which may be challenging in some scenarios.
Nevertheless, there are drawbacks to {locally} unbiased estimation with twirling. Specifically, the Fisher information for the entanglement-based protocol is higher than that of the twirling-based protocol.

\setcounter{corollary}{0}
\section{Proof of Corollary 1 in the Main Text}
In this section, we provide the proof of Corollary 1 stated in the main text.
For clarity,  we restate the corollary below.
We note that the Corollary stated in the main text is the simplified version, and the following Corollary is more general.

\begin{corollary}\label{S:corollary1}
    Let $\mathcal{E}_{\vb*{\theta}}$ be an unknown quantum channel parameterized by unknown parameters $\vb*{\theta} = (\theta_1, \ldots, \theta_M)^T$.
    Then, given multiple accesses to the unknown channel $\mathcal{E}_{\vb*{\theta}}$, a locally unbiased estimator of parameter $\theta_1$ at $\vb*{\theta} = \vb*{\theta}_0$ with finite estimation error exists if and only if
    \begin{align}
    \label{eq_cors1_1}
    \left.\pdv{\mathcal{E}_{\vb*{\theta}}}{\theta_{1}}\right|_{\vb*{\theta} = \vb*{\theta}_0} \neq \sum_{i \neq 1} c_{i} \left.\pdv{\mathcal{E}_{\vb*{\theta}}}{\theta_{i}}\right|_{\vb*{\theta} = \vb*{\theta}_0}
    \end{align}
    for all $c_i\in\mathbb{C}$. Here, $\vb*{\theta}_0=(\theta_{1,0},\theta_{2,0},\cdots,\theta_{M,0})^{\mathrm{T}}$ is a fixed point.
\end{corollary}

\begin{proof}
    In this proof, we set $\vb*{\theta} = \vb*{\theta}_0$.
    We employ the fact that the most general estimation protocol for unknown quantum channels is represented by the \textit{sequential strategy}~\cite{giovannetti2006quantum, demkowicz2014using}.
    In this strategy, we access $\mathcal{E}_{\vb*{\theta}}$ for $N$ times and prepare the quantum state
    \begin{align}
        \label{eq_cors1_pr_1}
        \hat{\rho}_{\mathrm{seq}} = \mathcal{C}_{N+1}\circ(\mathcal{E}_{\vb*{\theta}}\otimes \mathcal{I})\circ \cdots\circ\mathcal{C}_2\circ(\mathcal{E}_{\vb*{\theta}}\otimes \mathcal{I})\circ \mathcal{C}_1(\hat{\rho}_0).
    \end{align}
    Here, $\hat{\rho}_0$ is the initial state of the system and ancilla, $\mathcal{C}_i$ is an arbitrary CPTP channel, and $\mathcal{I}$ denotes the identity channel on the ancilla.
    Since the sequential strategy is the most general estimation protocol, a {locally} unbiased estimator of $\theta_1$ with finite estimation error exists if and only if there exists a sequential strategy that allows {locally} unbiased estimation of $\theta_1$ from the output state $\hat{\rho}_{\mathrm{seq}}$.
    This condition is equivalent to the existence of $N$, $\hat{\rho}_0$, and $\mathcal{C}_i$ such that
    \begin{align}
        \label{eq_cors1_pr_2}
        \pdv{\theta_{1}}\mathcal{C}_{N+1}\circ(\mathcal{E}_{\vb*{\theta}}\otimes \mathcal{I})\circ \cdots\circ\mathcal{C}_2\circ(\mathcal{E}_{\vb*{\theta}}\otimes \mathcal{I})\circ \mathcal{C}_1(\hat{\rho}_0)
        \neq \sum_{i \neq 1} d_{i} \pdv{\theta_{i}}\mathcal{C}_{N+1}\circ(\mathcal{E}_{\vb*{\theta}}\otimes \mathcal{I})\circ \cdots\circ\mathcal{C}_2\circ(\mathcal{E}_{\vb*{\theta}}\otimes \mathcal{I})\circ \mathcal{C}_1(\hat{\rho}_0)
    \end{align}
    for all $d_i\in\mathbb{C}$.
    Therefore, it suffices to show the equivalence between Eq.~\eqref{eq_cors1_1} and Eq.~\eqref{eq_cors1_pr_2}.

    To establish this equivalence, we demonstrate that parameters $c_i\in\mathbb{C}$ satisfying
    \begin{align}
    \label{eq_cors1_pr_3}
    \pdv{\mathcal{E}_{\vb*{\theta}}}{\theta_{1}} = \sum_{i \neq 1} c_{i} \pdv{\mathcal{E}_{\vb*{\theta}}}{\theta_{i}}
    \end{align}
    exist if and only if, for all $N$, $\hat{\rho}_0$, and $\mathcal{C}_i$, parameters $d_i\in\mathbb{C}$ exist such that
    \begin{align}
        \label{eq_cors1_pr_4}
        \pdv{\theta_{1}}\mathcal{C}_{N+1}\circ(\mathcal{E}_{\vb*{\theta}}\otimes \mathcal{I})\circ \cdots\circ\mathcal{C}_2\circ(\mathcal{E}_{\vb*{\theta}}\otimes \mathcal{I})\circ \mathcal{C}_1(\hat{\rho}_0)
        = \sum_{i \neq 1} d_{i} \pdv{\theta_{i}}\mathcal{C}_{N+1}\circ(\mathcal{E}_{\vb*{\theta}}\otimes \mathcal{I})\circ \cdots\circ\mathcal{C}_2\circ(\mathcal{E}_{\vb*{\theta}}\otimes \mathcal{I})\circ \mathcal{C}_1(\hat{\rho}_0).
    \end{align}

    We first prove the ``if" direction.
    Assume that for all $N$, $\hat{\rho}_0$, and $\mathcal{C}_i$, there exist parameters $d_i\in\mathbb{C}$ satisfying Eq.~\eqref{eq_cors1_pr_4}.
    In particular, consider the case where $N = 1$, $\hat{\rho}_0 = \ketbra{\Psi}$ with $\ket{\Psi} = \sum_i \ket{i}\ket{i}$ being the maximally entangled state, and where $\mathcal{C}_1$ and $\mathcal{C}_2$ are identity operations on the system and the ancilla.
    In this case, we obtain
    \begin{align}
        \label{S:eq_cors1_pr_5}
        \left(\pdv{\theta_{1}}\mathcal{E}_{\vb*{\theta}}\otimes \mathcal{I}\right)(\ketbra{\Psi})
        = \left(\sum_{i \neq 1} d_{i} \pdv{\theta_{i}}\mathcal{E}_{\vb*{\theta}}\otimes \mathcal{I}\right)(\ketbra{\Psi}),
    \end{align}
    for some $d_i\in\mathbb{C}$.
    This equality implies that the Choi state of the map $\partial_{\theta_{1}}\mathcal{E}_{\vb*{\theta}}$ coincides with the Choi state of the map $\sum_{i \neq 1} d_{i} \partial_{\theta_{i}}\mathcal{E}_{\vb*{\theta}}$.
    Therefore, we obtain Eq.~\eqref{eq_cors1_pr_3} with $c_i = d_i$.
    
    Next, we prove the ``only if" direction.
    Assume that there exist parameters $c_i\in\mathbb{C}$ satisfying Eq.~\eqref{eq_cors1_pr_3}.
    Then, since we have
    \begin{align}
        \label{S:eq_cors1_pr_6}
        &\;\;\;\pdv{\theta_{i}}\mathcal{C}_{N+1}\circ(\mathcal{E}_{\vb*{\theta}}\otimes \mathcal{I})\circ \cdots\circ\mathcal{C}_2\circ(\mathcal{E}_{\vb*{\theta}}\otimes \mathcal{I})\circ \mathcal{C}_1(\hat{\rho}_0) \\
        &= \mathcal{C}_{N+1}\circ(\partial_{\theta_i}\mathcal{E}_{\vb*{\theta}}\otimes \mathcal{I})\circ \cdots\circ\mathcal{C}_2\circ(\mathcal{E}_{\vb*{\theta}}\otimes \mathcal{I})\circ \mathcal{C}_1(\hat{\rho}_0) + \cdots + \mathcal{C}_{N+1}\circ(\mathcal{E}_{\vb*{\theta}}\otimes \mathcal{I})\circ \cdots\circ\mathcal{C}_2\circ(\partial_{\theta_i}\mathcal{E}_{\vb*{\theta}}\otimes \mathcal{I})\circ \mathcal{C}_1(\hat{\rho}_0),
    \end{align}
    we obtain Eq.~\eqref{eq_cors1_pr_4} with $d_i = c_i$.
\end{proof}

\section{Details of Noise parameter estimation in Cycle Benchmarking for the Pauli-Z Rotation Gate}
In the main text, we discussed the {unbiased estimation} of noise {parameters} affecting the Pauli-Z rotation gate using cycle benchmarking~\cite{learnability-erhard2019characterizing}.
In this section, we provide a more detailed analysis.

As in the main text, we consider {the estimation of} the noise on the Pauli-$Z$ rotation gate $U = \hat{R}_z(\phi) = e^{-i\phi/2 Z}$ using cycle benchmarking.
As shown in Sec.~V, the noise channel $\mathcal{N}$ affecting $U$ can be twirled, such that its Pauli transfer matrix is expressed as
\begin{align}
    \begin{pmatrix}
        1 & 0 & 0 & 0 \\
        0 & \lambda_1\cos\theta & \lambda_1\sin\theta & 0 \\
        0 & -\lambda_1\sin\theta & \lambda_1\cos\theta & 0 \\
        \alpha & 0 & 0 & \lambda_2 \\
    \end{pmatrix},
\end{align}
which is parameterized by four independent parameters: $\lambda_1, \lambda_2, \alpha$, and $\theta$.
The noise channel $\mathcal{N}$ represents a mixture of depolarizing noise, dephasing noise, amplitude damping noise, and coherent noise due to over-rotation.

We consider {the unbiased estimation of} the twirled noise $\mathcal{N}$ using cycle benchmarking under state preparation and measurement (SPAM) errors.
We assume that the state preparation error $\mathcal{N}_S$ and measurement error $\mathcal{N}_M$ are Pauli-twirled to Pauli noise, whose Pauli transfer matrices are represented as
\begin{align}
    \begin{pmatrix}
        1 & 0 & 0 & 0 \\
        0 & \lambda_{1S} & 0 & 0 \\
        0 & 0 & \lambda_{2S} & 0 \\
        0 & 0 & 0 & \lambda_{3S} \\
    \end{pmatrix},
    \quad
    \begin{pmatrix}
        1 & 0 & 0 & 0 \\
        0 & \lambda_{1M} & 0 & 0 \\
        0 & 0 & \lambda_{2M} & 0 \\
        0 & 0 & 0 & \lambda_{3M} \\
    \end{pmatrix}.
\end{align}
In cycle benchmarking, we repeatedly apply the noisy $\hat{R}_z(\phi)$ gate.
This corresponds to analyzing the channel $\mathcal{N}_{d} = \mathcal{N}_M \circ (\mathcal{N} \circ \mathcal{U})^{\circ d} \circ \mathcal{N}_S$ for various $d$, whose Pauli transfer matrix $A_d$ is given by
\begin{align}
    \begin{pmatrix}
        1 & 0 & 0 & 0 \\
        0 & \lambda_{1M}\lambda_{1S}\lambda_1^d\cos d\theta' & \lambda_{1M}\lambda_{2S}\lambda_1^d \sin d\theta' & 0 \\
        0 & -\lambda_{2M}\lambda_{1S}\lambda_1^d\sin d\theta' & \lambda_{2M}\lambda_{2S}\lambda_1^d \cos d\theta' & 0 \\
        \lambda_{3M}\alpha\frac{1-\lambda_1^d}{1-\lambda_1} & 0 & 0 & \lambda_{3M}\lambda_{3S}\lambda_2^d \\
    \end{pmatrix}.
\end{align}
Here, $\theta' = \phi + \theta$ and $(A_d)_{ij}$ denotes the $(i+1,j+1)$-th component of $A_d$.

To encapsulate $\mathcal{N}_{d}$ for all integer values of $d$ in a single description, we introduce an augmented channel $\mathcal{N}_{\mathrm{cycle}}$, and a classical register $\hat{\rho}_{c}$. This register keeps track of the number of applied gates $d$. Formally, we define
\begin{align}
    \mathcal{N}_{\mathrm{cycle}} : \hat{\rho} \otimes \hat{\rho}_c \mapsto \sum_{d}\mathcal{N}_{d}(\hat{\rho}) \otimes \ketbra{d}\hat{\rho}_c\ketbra{d}.
\end{align}
In this representation, each outcome $\ketbra{d}$ in the register corresponds to the channel $\mathcal{N}_{d}$ acting on the initial state $\hat{\rho}$.

Since there is a one-to-one correspondence between a quantum channel and its Choi state, we consider its Choi state defined as
\begin{align}
    \hat{\rho}_{\mathrm{cycle}} = \mathcal{N}_{\mathrm{cycle}} \otimes \mathcal{I}(\ketbra{\Psi}\otimes\ketbra{\Psi_D}) = \frac{1}{D} \sum_{d=0}^D \hat{\rho}_d \otimes \ketbra{dd}.
\end{align}
Here, $\ketbra{\Psi}$ and $\ketbra{\Psi_D}$ are maximally entangled states for single-qubit and $D$-dimensional systems, respectively, and $\hat{\rho}_d$ is the Choi state of the channel $\mathcal{N}_{d}$, defined as
\begin{align}
    \hat{\rho}_d = \mathcal{N}_{d}\otimes \mathcal{I}(\ketbra{\Psi}).
\end{align}
Since the maximally entangled state $\ketbra{\Psi}$ can be represented as
\begin{align}
    \ketbra{\Psi} = \frac{1}{4} \sum_{i=0}^{3} c_i P_i\otimes P_i,
\end{align}
where $c_i = \pm1$, $P_0 = I$, $P_1 = X$, $P_2 = Y$, and $P_3 = Z$, we can represent $\hat{\rho}_d$ as
\begin{align}
    \hat{\rho}_d 
    &= \frac{1}{4}\sum_{i=0}^{3}\sum_{j=0}^{3}c_j(A_{d})_{ij}P_i\otimes P_j.
\end{align}

Let us show that a {locally} unbiased estimator of $\alpha$ cannot be obtained from $\hat{\rho}_{\mathrm{cycle}}$.
For the parameters $\alpha, \lambda_{3M}, \lambda_{3S}$, we have
\begin{align}
    \alpha \partial_\alpha \hat{\rho}_{\mathrm{cycle}} &= \frac{1}{D} \sum_{d=1}^D c_0(A_{d})_{30}P_3\otimes P_0\otimes \ketbra{dd}, \\
    \lambda_{3M} \partial_{\lambda_{3M}} \hat{\rho}_{\mathrm{cycle}} &= \frac{1}{D} \sum_{d=1}^D (c_0(A_{d})_{30}P_3\otimes P_0 + c_3(A_{d})_{33}P_3\otimes P_3)\otimes \ketbra{dd}, \\
    \lambda_{3S} \partial_{\lambda_{3S}} \hat{\rho}_{\mathrm{cycle}} &= \frac{1}{D} \sum_{d=1}^D (c_3(A_{d})_{33}P_3\otimes P_3)\otimes \ketbra{dd}.
\end{align}
Therefore, we obtain
\begin{align}
    \alpha \partial_\alpha \hat{\rho}_{\mathrm{cycle}} = \lambda_{3M} \partial_{\lambda_{3M}} \hat{\rho}_{\mathrm{cycle}} - \lambda_{3S} \partial_{\lambda_{3S}} \hat{\rho}_{\mathrm{cycle}},
\end{align}
which is equivalent to
\begin{align}
    \alpha \partial_\alpha \mathcal{N}_{\mathrm{cycle}} = \lambda_{3M} \partial_{\lambda_{3M}} \mathcal{N}_{\mathrm{cycle}} - \lambda_{3S} \partial_{\lambda_{3S}} \mathcal{N}_{\mathrm{cycle}}.
\end{align}
Thus, by Corollary~\ref{S:corollary1}, a {locally} unbiased estimator for the parameter $\alpha$ cannot be derived from $\mathcal{N}_{\mathrm{cycle}}$.

Next, we show that $\lambda_1, \lambda_2, \theta'$ {admit locally unbiased estimators within cycle benchmarking model}.
It is straightforward to observe that $\lambda_2$ can be {estimated} by calculating $(A_{d+1})_{33}/ (A_d)_{33}$, so we focus on parameters $\lambda_1$ and $\theta'$.
We can determine $\lambda_{1M}\lambda_{1S}$ and $\lambda_{2M}\lambda_{2S}$ by calculating $(A_0)_{11}$ and $(A_0)_{22}$.
Additionally, the ratio $\frac{\lambda_{1M}}{\lambda_{2M}}\frac{\lambda_{2S}}{\lambda_{1S}}$ can be obtained by taking the ratio of $(A_d)_{12}$ and $(A_d)_{21}$.
From these values, we can derive $\frac{\lambda_{1M}}{\lambda_{2M}}$ and $\frac{\lambda_{2S}}{\lambda_{1S}}$ and subsequently convert $(A_d)_{12} = \lambda_{1M}\lambda_{2S}\lambda_1^d\sin d \theta'$ into $\lambda_{1M}\lambda_{1S}\lambda_1^d\sin d\theta'$.
Taking the ratio of this expression with $(A_d)_{11} = \lambda_{1M}\lambda_{1S}\lambda_1^d\cos d\theta'$, we can determine $\theta'$.
Since we can {unbiasedly estimate} $\lambda_{1M}\lambda_{1S}$ and $\theta'$, we can also extract $\lambda_1$ from $(A_d)_{11} = \lambda_{1M}\lambda_{1S}\lambda_1^d\cos d\theta'$.
We can also {determine} $\lambda_{1M}\lambda_{2S}$ and $\lambda_{2M}\lambda_{1S}$ from $(A_d)_{12} = \lambda_{1M}\lambda_{2S}\lambda_1^d\sin d \theta'$ and  $(A_d)_{21} = -\lambda_{2M}\lambda_{1S}\lambda_1^d\sin d\theta'$.
Combining these {estimated} values with $\frac{\lambda_{1M}}{\lambda_{2M}}$ and $\frac{\lambda_{2S}}{\lambda_{1S}}$, we can individually {unbiasedly estimate} $\lambda_{1S}$, $\lambda_{2S}$, $\lambda_{1M}$, and $\lambda_{2M}$.

In conclusion, we can {unbiasedly estimate} $\lambda_1,\lambda_2,\theta$, but not $\alpha$.
When considering SPAM errors, we can {unbiasedly estimate} $\lambda_1$, $\lambda_2$, $\theta$, $\lambda_{1S}$, $\lambda_{2S}$, $\lambda_{1M}$, $\lambda_{2M}$, but cannot {unbiasedly estimate} $\alpha, \lambda_{3S}, \lambda_{3M}$ individually; only the products $\lambda_{3M}\alpha$ and $\lambda_{3M}\lambda_{3S}$ can be {unbiasedly estimate}.

\section{Unbiased Estimation of Noise Parameters in Cycle Benchmarking for the CNOT Gate}
In the previous section, we analyzed the {unbiased estimation of noise parameters} affecting the Pauli-$Z$ rotation gate using cycle benchmarking~\cite{learnability-erhard2019characterizing}.
Our analysis was based on Corollary~\ref{S:corollary1}, and a similar analysis can be applied to noise affecting general gates.
As another application of Corollary~\ref{S:corollary1}, we discuss the {unbiased estimation} of Pauli noise affecting the CNOT gate using cycle benchmarking.

Consider the 2-qubit noise channel $\mathcal{N}$ acting on the CNOT gate using cycle benchmarking.
Since the CNOT gate is a Clifford gate, we can Pauli twirl the noise channel $\mathcal{N}$ such that it becomes a Pauli noise~\cite{emerson2005scalable}.
Thus, $\mathcal{N}$ can be characterized by 15 Pauli eigenvalues $\{\lambda_i\}_{i=1}^{15}$.
We aim to {unbiasedly estimate} the noise parameters of $\mathcal{N}$ (i.e., the 15 Pauli eigenvalues $\{\lambda_i\}_{i=1}^{15}$) under state preparation and measurement (SPAM) errors.
We assume that the state preparation error $\mathcal{N}_S$ and measurement error $\mathcal{N}_M$ are also Pauli noise, whose Pauli eigenvalues are denoted by $\{\lambda_{iS}\}_{i=1}^{15}$ and $\{\lambda_{iM}\}_{i=1}^{15}$, respectively.

In cycle benchmarking, the noisy CNOT gate is applied repeatedly.
This means we have access to the channel $\mathcal{N}_{d} = \mathcal{N}_M \circ (\mathcal{N} \circ \mathcal{U})^{\circ d} \circ \mathcal{N}_S$ for various values of $d$, where $\mathcal{U}$ represents the CNOT gate.
To encapsulate $\mathcal{N}_{d}$ for all integer values of $d$ in a single description, we introduce an augmented channel $\mathcal{N}_{\mathrm{cycle}}$, and a classical register $\hat{\rho}_{c}$. This register keeps track of the number of applied gates $d$. Formally, we define
\begin{align}
    \mathcal{N}_{\mathrm{cycle}} : \hat{\rho} \otimes \hat{\rho}_c \mapsto \sum_{d}\mathcal{N}_{d}(\hat{\rho}) \otimes \ketbra{d}\hat{\rho}_c\ketbra{d}.
\end{align}
In this representation, each outcome $\ketbra{d}$ in the register corresponds to the channel $\mathcal{N}_{d}$ acting on the initial state $\hat{\rho}$.

Since there is a one-to-one correspondence between a quantum channel and its Choi state, we consider the Choi state defined as
\begin{align}
    \hat{\rho}_{\mathrm{cycle}} = \mathcal{N}_{\mathrm{cycle}} \otimes \mathcal{I}(\ketbra{\Psi}\otimes\ketbra{\Psi_D}) = \frac{1}{D} \sum_{d=1}^D \hat{\rho}_d \otimes \ketbra{dd}.
\end{align}
Here, $\ketbra{\Psi}$ and $\ketbra{\Psi_D}$ are maximally entangled states for the 2-qubit and $D$-dimensional systems, respectively, and $\hat{\rho}_d$ is the Choi state of the channel $\mathcal{N}_{d}$, defined as
\begin{align}
    \hat{\rho}_d = \mathcal{N}_{d}\otimes \mathcal{I}(\ketbra{\Psi}).
\end{align}
Since the maximally entangled state $\ketbra{\Psi}$ can be represented as
\begin{align}
    \ketbra{\Psi} = \frac{1}{16} \sum_{i=0}^{15} c_i P_i\otimes P_i,
\end{align}
where $c_i = \pm1$, we can represent $\hat{\rho}_d$ as
\begin{align}
    \hat{\rho}_d 
    &= \frac{1}{16}\sum_{i=0}^{15}c_i\mathcal{N}_d(P_i)\otimes P_i \\
    &= 
    \left\{
    \begin{array}{ll}
        \frac{1}{16}\sum_{i=0}^{15}c_{i,d}  P_i\otimes P_i & (d\text{: even}) \\
        \frac{1}{16}\sum_{i=0}^{15}c_{i,d}  P_{\mathrm{CNOT}(i)}\otimes P_i & (d\text{: odd}) \\
    \end{array}
    \right.
    ,
\end{align}
where $\mathrm{CNOT}(i)$ represents the index of the Pauli operator obtained by conjugating $P_i$ with the CNOT gate, and the coefficient $c_{i,d}$ is defined as
\begin{align}
    c_{i,d} = 
    \left\{
    \begin{array}{ll}
        c_i \lambda_{iM}\lambda_{iS}(\lambda_{i}\lambda_{\mathrm{CNOT}(i)} )^{d/2} & (d\text{: even}) \\
        c_i \lambda_{\mathrm{CNOT}(i)M}\lambda_{iS}\lambda_{\mathrm{CNOT}(i)}(\lambda_{i}\lambda_{\mathrm{CNOT}(i)} )^{(d-1)/2}  & (d\text{: odd}) \\
    \end{array}
    \right.
    .
\end{align}

We first show that if $\mathrm{CNOT}(i) = i$, i.e., if the Pauli operator $P_i$ commutes with the CNOT gate, then the corresponding Pauli eigenvalue $\lambda_i$ can be {unbiasedly estimated}.
For such a parameter $\lambda_i$, we have
\begin{align}
    \lambda_{i}\partial_{\lambda_i} \hat{\rho}_{\mathrm{cycle}} = \frac{1}{D} \sum_{d=1}^D dc_{i,d} P_i\otimes P_i \otimes \ketbra{dd}.
\end{align}
It is evident that such an expression for $\lambda_{i}\partial_{\lambda_i} \hat{\rho}_{\mathrm{conc}}$ cannot be written as a linear combination of the derivatives of other parameters due to the presence of the coefficient $d$.
Indeed, a {locally} unbiased estimator of $\lambda_i$ can be obtained by first estimating $c_{i,d}$ and $c_{i,d+2}$, and then calculating $\sqrt{c_{i,d+2}/c_{i,d}}$.

Next, we show that a {locally} unbiased estimator of $\lambda_i$ cannot be obtained if $\mathrm{CNOT}(i) \neq i$, i.e., if the Pauli operator $P_i$ does not commute with the CNOT gate.
For such a parameter $\lambda_i$, we have
\begin{align}
    \lambda_{i}\partial_{\lambda_i} \hat{\rho}_{\mathrm{cylce}} = \frac{1}{D} \Bigg(&\sum_{d\text{: even}} \left(\frac{d}{2} c_{i,d} P_i\otimes P_i + \frac{d}{2} c_{j,d} P_{j}\otimes P_{j}\right) \otimes \ketbra{dd} \\
    + &\sum_{d\text{: odd}} \left(\frac{d-1}{2} c_{i,d} P_{j}\otimes P_i + \frac{d+1}{2} c_{j,d} P_{i}\otimes P_{j} \right) \otimes \ketbra{dd}
    \Bigg),
\end{align}
where we define $j = \mathrm{CNOT}(i)$.
Furthermore, we have
\begin{align}
    \lambda_{j}\partial_{\lambda_{j}} \hat{\rho}_{\mathrm{cycle}} = \frac{1}{D} \Bigg(&\sum_{d\text{: even}} \left(\frac{d}{2} c_{i,d} P_i\otimes P_i + \frac{d}{2} c_{j,d} P_{j}\otimes P_{j}\right) \otimes \ketbra{dd} \\
    + &\sum_{d\text{: odd}} \left(\frac{d+1}{2} c_{i,d} P_{j}\otimes P_i + \frac{d-1}{2} c_{j,d} P_{i}\otimes P_{j} \right) \otimes \ketbra{dd}
    \Bigg),
\end{align}
and
\begin{align}
    \lambda_{iS}\partial_{\lambda_{iS}} \hat{\rho}_{\mathrm{cycle}}
    &= \frac{1}{D} \Bigg(\sum_{d\text{: even}} \left(c_{i,d} P_i\otimes P_i\right) \otimes \ketbra{dd} +\sum_{d\text{: odd}} \left(c_{i,d} P_{j}\otimes P_i \right) \otimes \ketbra{dd} \Bigg), \\
    \lambda_{iM}\partial_{\lambda_{iM}} \hat{\rho}_{\mathrm{cycle}}
    &= \frac{1}{D} \Bigg(\sum_{d\text{: even}} \left(c_{i,d} P_i\otimes P_i\right) \otimes \ketbra{dd} +\sum_{d\text{: odd}} \left(c_{j,d} P_{i}\otimes P_j \right) \otimes \ketbra{dd} \Bigg). \\
\end{align}
Therefore, we obtain
\begin{align}
    \lambda_{i}\partial_{\lambda_i} \hat{\rho}_{\mathrm{cycle}} = \lambda_{j}\partial_{\lambda_j} \hat{\rho}_{\mathrm{cylce}} - \lambda_{iS}\partial_{\lambda_{iS}} \hat{\rho}_{\mathrm{cylce}} + \lambda_{iM}\partial_{\lambda_{iM}} \hat{\rho}_{\mathrm{cylce}},
\end{align}
which is equivalent to 
\begin{align}
    \lambda_{i}\partial_{\lambda_i} \mathcal{N}_{\mathrm{cycle}} = \lambda_{j}\partial_{\lambda_j} \mathcal{N}_{\mathrm{cycle}} - \lambda_{iS}\partial_{\lambda_{iS}} \mathcal{N}_{\mathrm{cycle}} + \lambda_{iM}\partial_{\lambda_{iM}} \mathcal{N}_{\mathrm{cycle}}.
\end{align}
Thus, by Corollary~\ref{S:corollary1}, a {locally} unbiased estimator of $\lambda_i$ for $i \neq \mathrm{CNOT}(i)$ cannot be derived from $\mathcal{N}_{\mathrm{cycle}}$. 

To summarize, the Pauli noise $\mathcal{N}$ of the CNOT gate cannot be {unbiasedly estimated} by simply concatenating the gate in the presence of SPAM errors.
More precisely, a {locally} unbiased estimator can be obtained for the Pauli eigenvalues corresponding to Pauli operators that commute with the CNOT gate, while the remaining parameters cannot be {unbiasedly estimated}.
It is worth noting that similar results were obtained in Ref.~\cite{learnability-chen2023learnability}, where graph-theoretic tools were employed.
{Here, we provide a new proof based on the locally unbiased estimability criterion.}

\end{document}